\documentclass[aps,prx,twocolumn, superscriptaddress]{revtex4-2}
\usepackage{amsmath}
\usepackage{amssymb}
\usepackage{amsthm}
\usepackage{graphicx}
\usepackage{dcolumn}
\usepackage{bm}
\usepackage{enumerate}
\usepackage{physics} 
\usepackage{subfigure}
\usepackage[percent]{overpic}
\usepackage{adjustbox}
\usepackage{hyperref}
\usepackage{extarrows}
\hypersetup{colorlinks = true, 
	    linkcolor = blue, 
	    urlcolor = blue,
            citecolor = blue}
\newtheorem*{theorem}{Theorem}
\newtheorem*{lemma}{Lemma}
\newtheorem{corollary}{Corollary}

\theoremstyle{definition}
\newtheorem{definition}{Definition}

\DeclareMathOperator{\Heff}{H_{\mathrm{eff}}}

\theoremstyle{definition}

\begin{document}
\preprint{APS/123-QED}
\title{Optimal Effective Hamiltonian for Quantum Computing and Simulation}
\author{Hao-Yu Guan}
\affiliation{International Quantum Academy, Shenzhen, 518048, China}
\author{Xiao-Long Zhu}
\affiliation{International Quantum Academy, Shenzhen, 518048, China}
\affiliation{Shenzhen Institute for Quantum Science and Engineering, Southern University of Science and Technology, Shenzhen, 518048, China}%
\author{Yu-Hang Dang}
\affiliation{International Quantum Academy, Shenzhen, 518048, China}
\author{Xiu-Hao Deng}
\email{dengxiuhao@iqasz.cn}
\affiliation{International Quantum Academy, Shenzhen, 518048, China}
\affiliation{Shenzhen Branch, Hefei National Laboratory, Shenzhen, 518048, China.}

\date{\today}
\begin{abstract}
The effective Hamiltonian serves as the conceptual pivot of quantum engineering, transforming physical complexity into programmable logic; yet, its construction remains compromised by the mathematical non-uniqueness of block diagonalization, which introduces an intrinsic ``gauge freedom" that standard methods fail to resolve. We address this by establishing the Least Action Unitary Transformation (LAUT) as the fundamental principle for effective models. By minimizing geometric action, LAUT guarantees dynamical fidelity and inherently enforces the preservation of symmetries—properties frequently violated by conventional Schrieffer-Wolff and Givens rotation techniques. We identify the Bloch-Brandow formalism as the natural perturbative counterpart to this principle, yielding analytic expansions that preserve symmetries to high order. We validate this framework against experimental data from superconducting quantum processors, demonstrating that LAUT quantitatively reproduces interaction rates in driven entangling gates where standard approximations diverge. Furthermore, in tunable coupler architectures, we demonstrate that the LAUT approach captures essential non-rotating-wave contributions that standard models neglect; this inclusion is critical for quantitatively reproducing interaction rates and revealing physical multi-body interactions such as $XZX+YZY$, which are verified to be physical rather than gauge artifacts. By reconciling variational optimality with analytical tractability, this work provides a systematic, experimentally validated route for high-precision system learning and Hamiltonian engineering.
\end{abstract}
\maketitle

\section{Introduction}
\begin{figure*}[t]
    \centering
    \includegraphics[width=1.4\columnwidth]{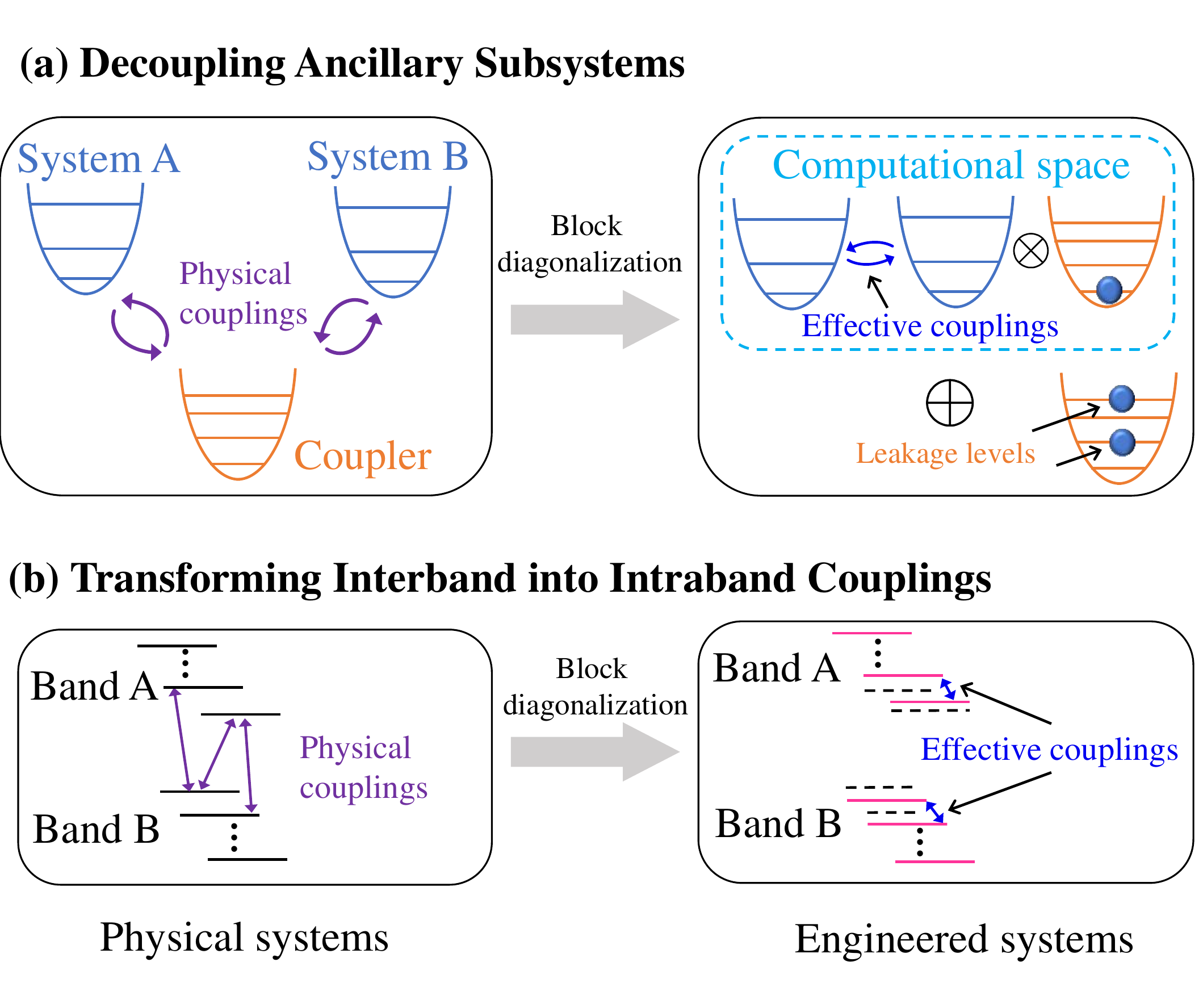}
    \caption{(Color online) Block diagonalization for effective Hamiltonian construction. The two representative cases introduced below exemplify how effective Hamiltonians reconcile microscopic complexity with low-dimensional physical intuition.
    \textbf{(a)}~Decoupling ancillary subsystems: A physical quantum circuit with couplers (left) is mapped to an effective model (right) where computational states are dynamically isolated. 
    \textbf{(b)}~Transforming interband into intraband couplings: A multiband system with hybridization (left) is transformed into an effective theory with only intraband dynamics (right). 
    Dashed lines indicate bare energies; pink solid lines indicate renormalized energies; blue arrows denote effective couplings. 
    }
    \label{Fig:IllustratingBD}
\end{figure*}
The effective Hamiltonian is the renormalized interface between microscopic complexity and macroscopic control; it is the distillation of relevant information that allows us to treat physical hardware as programmable logic. Historically, effective Hamiltonians have played a foundational role in nuclear structure theory~\cite{bloch1958theorie}, quantum chemistry~\cite{shavitt2009many}, condensed‑matter physics~\cite{kohn1959theory}, and atomic spectroscopy~\cite{schucan1972effective}, where they distill essential physics by eliminating irrelevant or virtual degrees of freedom. In the pursuit of scalable quantum information processing, the ability to derive faithful effective models for targeted subspaces—whether they be low-energy qubit manifolds, specific topological bands, or high-excitation Rydberg sectors—is paramount. Whether in superconducting circuits~\cite{blais2021circuit}, trapped ions~\cite{bruzewicz2019trapped}, or neutral atom arrays~\cite{morgado2021quantum}, the engineering of high-fidelity quantum computers~\cite{magesan2020effective,ku2020suppression,xu2021zz,xu2023parasitic,xu2024lattice} or quantum simulators~\cite{dang2024weaving,momoi2025numerical} relies entirely on our ability to dynamically isolate a computational manifold from a distinct ‘‘environment" of leakage levels and ancillary modes, as illustrated in Fig.~\ref{Fig:IllustratingBD}. However, as the field advances toward the strong-coupling and near-degenerate regimes required for fast gates and high-connectivity architectures, the standard theoretical frameworks used to construct this interface are facing a fundamental crisis of ambiguity.

The central difficulty lies in the mathematical non-uniqueness of block diagonalization. While the condition for decoupling a target subspace is precise, the unitary transformation that achieves it possesses an intrinsic ``gauge freedom"—an arbitrary rotation within the target subspace~\cite{di2019resolution}. In the weak-coupling limit, this freedom is often negligible. However, whether utilizing frequency-selective driving~\cite{wei2022hamiltonian,Schafer2018,Kawakami2014} or tunable couplers~\cite{krantz2019quantum,Malinowski2019} to engineer faster entanglement operations, interaction strengths increasingly rival spectral gaps. In this regime, perturbative approaches like Schrieffer-Wolff (SWT)~\cite{bravyi2011schrieffer} suffer from convergence issues, while numerical techniques such as Givens rotations (GR)~\cite{li2022nonperturbative}, though exact in spectral resolution, rely on path-dependent gauge choices. This lack of a canonical geometric constraint can result in spurious symmetry breaking and the loss of adiabatic continuity with the laboratory frame. Consequently, extracted interaction parameters cease to be objective physical invariants, becoming instead dependent on the algorithmic path of the chosen scheme. This gauge ambiguity renders the effective model physically opaque, introducing systematic uncertainties that fundamentally limit the precision required for Hamiltonian learning~\cite{qi2019determining,huang2020predicting} and calibration in quantum computing and simulation~\cite{di2019resolution,eisert2020quantum}.

In this work, we resolve this ambiguity by elevating the effective Hamiltonian from a phenomenological tool to a variationally optimized entity. We posit that a rigorous effective model must satisfy a Least Action Unitary Transformation (LAUT) principle: among the manifold of all block-diagonalizing transformations, the physical solution is the unique unitary that minimizes the geometric distance to the identity~\cite{cederbaum1989block, kvaal2008geometry}. Although the mathematical concept of minimizing this unitary distance has roots in chemical physics, its role as a fundamental selection rule for quantum engineering and its implications for dynamical fidelity have remained unexplored.

We present a unified framework that reconciles this variational optimality with analytical tractability, delivering three distinct advances:

First, we establish a rigorous link between static geometry and dynamical fidelity. We prove a theoretical lower bound showing that the transformation minimizing the Frobenius distance to the identity (the ``least action") also maximizes the long-time evolution fidelity. This result provides the missing physical justification for the LAUT geometry, linking it directly to the observable reliability of quantum gates.

Second, we address the critical issue of symmetry preservation. We prove that, unlike sequential numerical diagonalization, which introduces arbitrary phase artifacts, the unique solution selected by the LAUT criterion inherently respects all symmetries shared by the Hamiltonian and the subspace partition. This ensures that effective interactions—such as mediated exchange couplings in multi-qubit lattices—are protected from spurious symmetry-breaking errors.

Third, we bridge the gap between non-perturbative numerics and analytical theory by identifying the Bloch-Brandow (BB) formalism~\cite{bloch1958theorie} as the natural perturbative counterpart to the LAUT principle. We demonstrate that the BB expansion, unlike the canonical Schrieffer-Wolff transformation , remains variationally optimal and symmetry-preserving to high orders. This provides a systematic route to deriving closed-form, high-precision effective Hamiltonians that resolve the inconsistencies inherent in standard perturbative techniques.

We validate this framework through detailed case studies on superconducting quantum hardware, including the driven frequency-selective entangling gates, cross-resonance gate~\cite{magesan2020effective, wei2022hamiltonian} as an example, and multi-body coupler architectures~\cite{sung2021realization}. We demonstrate that the LAUT-based approach quantitatively reproduces experimental interaction rates in regimes where standard approximations diverge and successfully identifies novel three-body interactions ($XZX + YZY$) mediated by higher-order virtual processes. To ensure rigorous applicability, we further introduce a static soundness metric that efficiently certifies the validity domain of the effective model without necessitating computationally expensive time-domain simulations. By providing a mathematically rigorous and experimentally validated standard for model reduction, this framework lays the foundation for high-precision Hamiltonian engineering across diverse quantum platforms.

\section{Theory}
\label{sec: EH theory and application}

\subsection{Effective Hamiltonian and Least-action Criterion}
\label{subsec:definition_of_EH}

The goal of effective Hamiltonian theory is to represent the dynamics of a complex quantum system within a reduced subspace, capturing the relevant physics while faithfully reproducing observable quantities. Formally, one seeks an operator $\Heff$ acting on the subspace $\mathcal{H}_P$ such that the projected evolution generated by the full Hamiltonian $H(t)$ can be approximated as 
\[PU(t)\ket{\psi} = P\mathcal{T}e^{-i\int_0^t H(\tau)d\tau }\ket{\psi} \approx e^{-iH_\mathrm{eff}t}P\ket{\psi}.\]
The precise meaning of $\Heff$, however, depends on how spectral correspondence and block-diagonal structure are enforced. A formal definition consistent with conventional theory is given below.

\begin{definition}{\textbf{Effective Hamiltonian:}}\label{def:general_EH}
Given a quantum system with a full Hamiltonian $ H $ acting on a Hilbert space $\mathcal{H}=\mathcal{H}_P\oplus\mathcal{H}_Q$, which is partitioned into  a relevant subspace $ \mathcal{H}_P=\bigoplus_i \mathcal{H}_{P_i} $ (e.g., low-energy sectors of interest) and its orthogonal complement $ \mathcal{H}_Q $, with projectors $ P=\sum_i P_i $ and $ Q = \mathcal{I} - P $, the effective Hamiltonian $ \Heff $ is an operator acting on $ \mathcal{H}_P $ that satisfies the following criteria, as drawn from conventional theory~\cite{soliverez1981general}:
\begin{enumerate}[(1)]
    \item The eigenvalues of $ \Heff $ match the eigenvalues of $ H $ associated with the states that have significant support in $ \mathcal{H}_P $.
    \item There is a one-to-one correspondence between the eigenstates of $ \Heff $ and the projections onto $ \mathcal{H}_P $ of the corresponding eigenstates of $ H $.
    \item \( \Heff \) is block-diagonal with respect to the decomposition \( \mathcal{H}=(\bigoplus_i \mathcal{H}_{P_i})\oplus\mathcal{H}_{Q}  \).
\end{enumerate}
\end{definition}

In practice, solving an effective Hamiltonian involves seeking a block-diagonalizing (BD) transformation $T$ that block-diagonalizes $H$ with respect to the decomposition $\mathcal{H} = \mathcal{H}_{\mathrm{P}} \oplus \mathcal{H}_{\mathrm{Q}}$, yielding 
\begin{equation}
    \tilde{H} = T^\dagger H T =
    \begin{pmatrix}
        H_{\mathrm{P}} & 0 \\
        0 & H_{\mathrm{Q}}
    \end{pmatrix}.
\end{equation}
And then $H_P = H_\mathrm{eff}$ is taken as the effective model of the target system. This manuscript focuses on unitary BD transformations and, hence, Hermitian effective Hamiltonians.

While this definition ensures spectral equivalence and block diagonalization, the transformation is highly non-unique, as different choices of $T$ can yield distinct effective Hamiltonians that differ in higher-order terms, gauge conventions within each subspace $\mathcal{H}_i$, and symmetry properties. In strongly hybridized or near-degenerate systems—prevalent in quantum hardware and many-body simulations—this ambiguity manifests as inconsistent predictions for observables, such as interaction strengths or dynamical phases, underscoring the need for a principled selection of $T$ that prioritizes physical fidelity and minimizes structural alteration.

To better understand the origin of this ambiguity, it is instructive to identify the precise mathematical freedom allowed in the block–diagonalization process. The lemma below makes this freedom explicit.

\begin{lemma}[\textbf{Gauge Freedom in Block Diagonalization}] \label{lemma:uniqueness}
Consider a Hilbert space partition $\mathcal{H} = \mathcal{H}_P \oplus \mathcal{H}_Q$. Let $U$ and $U'$ be two unitary transformations that both block-diagonalize the Hamiltonian $H$ with respect to this partition, sorting the same set of eigenvalues into the target subspace $\mathcal{H}_P$. The relationship between these two transformations, $U' = W U$, depends on the spectral structure across the partition boundary:
\begin{enumerate}[(1)]
    \item \textbf{The Gapped Regime.} In the absence of any spectral degeneracy intersecting the partition boundary, the connecting unitary $W$ is strictly block-diagonal:
$$W = \begin{pmatrix} W_P & 0 \\ 0 & W_Q \end{pmatrix}$$, 
where $W_P$ represents an arbitrary internal basis rotation (gauge freedom) within the target subspace.
\item \textbf{The Resonant Regime.} If a degenerate eigenspace spans across both blocks, $W$ is block-diagonal only up to a mixing within that degenerate subspace.
\end{enumerate}
\end{lemma}

\begin{proof}
Define \(W:=U'U^{\dagger}\). By construction, \(W\) is unitary and satisfies
\(\Lambda'=W\Lambda W^{\dagger}\).
Write the spectral decomposition
\(\Lambda=\sum_{\lambda}\lambda\Pi_{\lambda}\).
Then \(\Lambda'=W\Lambda W^{\dagger}=\sum_{\lambda}\lambda\,W\Pi_{\lambda}W^{\dagger}\),
so \(W\Pi_{\lambda}W^{\dagger}\) is the spectral projector of \(\Lambda'\) corresponding to \(\lambda\).
Thus, \(W\) maps each eigenspace \(E_{\lambda}(\Lambda)=\operatorname{Ran}\Pi_{\lambda}\)
onto \(E_{\lambda}(\Lambda')\).

If every \(\Pi_{\lambda}\) is contained in exactly one of the blocks, then \(W\) cannot mix vectors from \(\mathcal H_P\) and \(\mathcal H_Q\) without leaving the eigen-space, so it must be block–diagonal with respect to \((P,Q)\), yielding item~(1).
If instead some \(\Pi_{\lambda}\) overlap both blocks, then the corresponding degenerate eigenspace can be rotated arbitrarily by \(W\) while still producing a block–diagonal form for \(H\), giving item~(2).
\end{proof}

This lemma formalizes the intuition that the block–diagonalizing transformation is unique only up to local rotations inside each invariant subspace. In the gapped regime, this restricts the freedom strictly to $W_P \oplus W_Q$. However, if degeneracy crosses the partition, this structure collapses. For instance, in a three–level system with degeneracy $\lambda_1=\lambda_2$:$$\Lambda = \mathrm{diag}(\lambda_1, \lambda_2, \mu), \quad P = |\lambda_1\rangle\langle\lambda_1|, \quad Q = |\lambda_2\rangle\langle\lambda_2| + |\mu\rangle\langle\mu|$$, any $U(2)$ rotation within $\mathrm{span}\{|\lambda_1\rangle, |\lambda_2\rangle\}$ preserves $\Lambda$ but mixes the blocks, thereby breaking strict block–diagonality. This confirms that in the physically relevant gapped regime, the effective Hamiltonian contains an intrinsic local gauge freedom $W_P$.

The LAUT principle addresses this by serving as a variational criterion that selects the transformation $T$ with the least deviation from the identity among all block-diagonalizing maps, thereby preserving the original Hamiltonian's structure and symmetries while achieving exact decoupling. Such a minimal-deformation principle is particularly vital in quantum physics, as it ensures interpretability, avoids unphysical artifacts, and facilitates robust modeling in regimes where perturbative or ad hoc methods falter, such as strong-coupling scenarios or multi-level architectures~\cite{magesan2020effective}. 

A well-established formulation~\cite{cederbaum1989block} minimizes the Frobenius norm
\begin{equation}
    \mathop{\arg\min}_T \|T - \mathcal{I}\|_F,
    \label{eq:frobenius_min}
\end{equation}
subject to the constraint that $T^\dagger H T$ is block diagonal. Notably, this criterion aligns with the principle of least information~\cite{cederbaum1989block}, ensuring that the resulting EH contains no arbitrary gauge information or spurious correlations beyond what is strictly necessary to block-diagonalize the system. Unlike trace fidelity, the LAUT framework based on the Frobenius norm remains well-defined for non-unitary but invertible transformations $T$, thus enabling extensions to non-Hermitian EHs in the future.

To resolve the gauge ambiguity identified in the Lemma and enforce a unique physical description, we propose that the definition of an effective Hamiltonian is incomplete without a variational criterion. Accordingly, we augment the EH construction by strictly enforcing a fourth defining condition:

\begin{enumerate}[(1)] 
    \setcounter{enumi}{3} 
    \item \textbf{Least Action Criterion:} Among all block-diagonalizing transformations satisfying conditions (1)-(3), the physical transformation \( T \) is the unique unitary that minimizes the geometric distance to the identity,
    \[
    T = \underset{T' \in \mathcal{T}_{BD}}{\arg \min} \|T' - \mathcal{I}\|_F,
    \]
    yielding the optimal effective Hamiltonian \( H_{\text{eff}} = P T^{\dagger} H T P \).
\end{enumerate}

\begin{figure}[t]
    \centering
    \includegraphics[width=0.85\columnwidth]{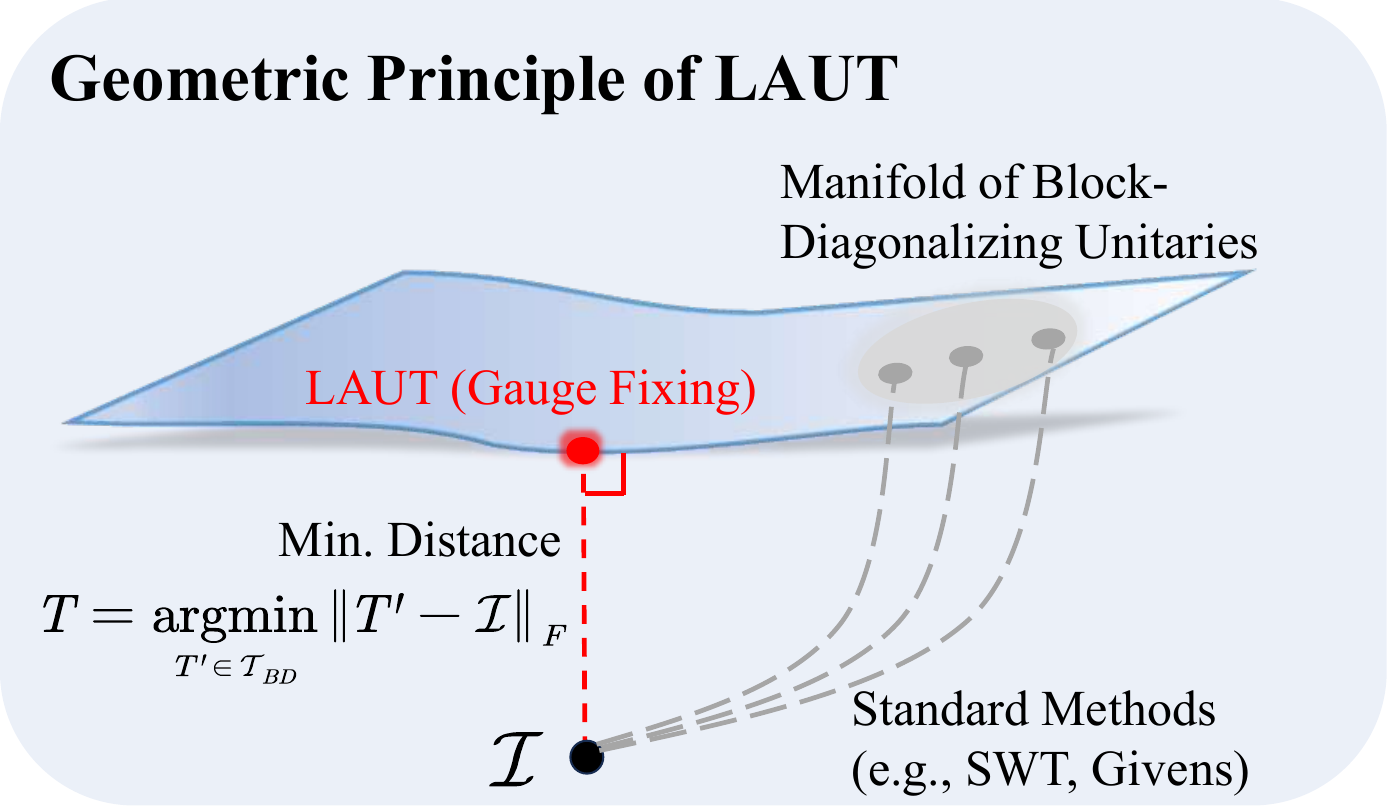}
    \caption{(Color online) Geometric interpretation of the Least Action Unitary Transformation (LAUT) principle.
As described in the Lemma, the set of all unitaries that block-diagonalize the Hamiltonian forms a continuous manifold $\mathcal{T}_{BD}$ (blue surface). Standard methods (grey dashed line) traverse this manifold via unconstrained trajectories. In contrast, the LAUT criterion (Definition 1) uniquely identifies the transformation $T$ (red dot) that minimizes the geometric distance $||T-\mathcal{I}||_F$ to the identity (black dot), thereby fixing the gauge freedom with minimal structural deformation.
    }
    \label{Fig:laut_geometry}
\end{figure}

This condition effectively ``legislates" the optimal frame, elevating the effective Hamiltonian from a phenomenological tool to a variationally optimized entity. To clarify this geometric intuition, we visualize the principle in Fig.~\ref{Fig:laut_geometry}. While conventional methods traverse the continuous solution manifold $\mathcal{T}_{BD}$ along arbitrary trajectories (gray dashed lines), the LAUT criterion acts as a strict gauge-fixing condition, isolating the unique transformation connected to the identity via the shortest geodesic (red dashed line).

\noindent\textit{Remark (Generalization to Hamiltonian Engineering).} 
Although formulated explicitly for subspace decoupling, the \textbf{LAUT} principle constitutes a universal variational foundation for Hamiltonian engineering. It generalizes naturally to tasks requiring targeted term cancellation---such as the suppression of counter-rotating wave contributions---where global subspace isolation is not strictly required. By replacing the block-diagonal constraint $\mathcal{T}_{\mathrm{BD}}$ with a specific interaction-vanishing constraint, the LAUT principle identifies the unique frame that eliminates unwanted terms with minimal geometric deformation, thereby minimizing artificial renormalization of the remaining physical parameters. This distinction between global decoupling and targeted cancelation is schematically illustrated in Fig.~\ref{Fig:IllustratingBD}.

This variational framework fundamentally refines the conceptual relationship between effective Hamiltonians and block diagonalization. In conventional formulations, the two are often regarded as synonymous, implying that \textit{any} transformation decoupling the subspaces yields a valid effective Hamiltonian. However, as formalized in Definition~\ref{def:general_EH}, block diagonalization is merely a mathematical procedure that yields a gauge-dependent manifold of solutions, whereas the Effective Hamiltonian (in a strict physical sense) must be a unique entity. 

By enforcing the LAUT principle as the fourth defining condition, we filter this manifold to select the unique transformation \( T \) with minimal geometric deformation. Thus, EH~\cite{bloch1958theorie,killingbeck2003bloch} and BD~\cite{li2022nonperturbative,bravyi2011schrieffer} must be conceptually distinguished: a rigorous one-to-one correspondence between the mathematical procedure and the physical model is established \textit{only} when the transformation is constrained by the LAUT criterion.

The explicit solution to the LAUT criterion via the exact block diagonalization (EBD-LAUT) method, following Ref.~\cite{cederbaum1989block}, is reviewed in Appendix~\ref{app:ebd-la}. Mathematically, this solution relies on the polar decomposition~\footnote{Analogous to the polar form of a complex number $z = r e^{i\theta}$, the polar decomposition uniquely factorizes the projection matrix into a unitary ``rotation'' component and a positive-semidefinite ``stretching'' component, isolating the former as the optimal transformation.} of the projection matrix linking the target eigenstates to the computational basis. Crucially, provided that these eigenstates maintain a one-to-one correspondence with the computational subspace, this projection is non-singular, thereby guaranteeing a unique unitary $T$ that minimizes $\| T - \mathcal{I}\|_F$. This provides a robust non-perturbative framework for EH construction. In Sec.~\ref{subsec:exact_EH_benchmark}, we benchmark LAUT against GR and dynamical spectral analysis based on Fourier transforms of population dynamics (hereafter referred to as FFT analysis), demonstrating that LAUT consistently surpasses GR in accuracy while preserving symmetry.


Unless otherwise specified, the block-diagonalizing transformation $T$ refers to the unique solution derived from the LAUT principle. In the resulting effective Hamiltonian, diagonal elements correspond to renormalized frequencies (or energies), while off-diagonal elements represent effective couplings.

\subsection{Accuracy of Effective Hamiltonian and Near Optimal Condition}
\label{subsec:LAUT_lower_bound}

A central question is how the LAUT criterion relates to the dynamical accuracy of the resulting effective Hamiltonian $H_{\mathrm{eff}}$. While the LAUT principle minimizes the structural deviation of the block-diagonalizing transformation, its impact on long-time evolution is not a priori evident. Here, we establish this connection by interpreting the LAUT objective as eigenbasis alignment and deriving a lower bound on the time-averaged evolution error. This analysis quantitatively links the static variational principle to dynamical fidelity, thereby justifying the LAUT-based definition of the effective Hamiltonian and clarifying its relevance for strongly driven entangling gates such as the CR gate, where perturbative methods often fail.

\paragraph*{Eigenbasis alignment and basis mismatch.-}

Let the full Hamiltonian \( H \) and its effective counterpart \( H_{\mathrm{eff}} \) share the same spectrum \( E \), with corresponding eigenstate matrices \( S \) and \( \mathcal{E} \), such that
\begin{align}
    S^{\dagger} H S &= \mathcal{E}^{\dagger} H_{\mathrm{eff}} \mathcal{E} = E.
\end{align}
It follows that $H_{\mathrm{eff}} = T^\dagger H T$ and \( T = S \mathcal{E}^\dagger \), which act as a unitary change of basis from \( \mathcal{E} \) to \( S \). The Frobenius norm of the distance between transformation $T$ and the identity quantifies the mismatch:
\begin{equation}
    \| T - \mathcal{I} \|_F^2 = \| S - \mathcal{E} \|_F^2 = \Tr(2\mathcal{I} - \mathcal{E}^\dagger S - S^\dagger \mathcal{E}).
\end{equation}
Minimizing this quantity guarantees optimal eigenbasis alignment, preserving the spectrum and improving dynamical fidelity.

Consider a Hermitian Hamiltonian $H$ acting on a $D$-dimensional Hilbert space $\mathcal{H}$, decomposed into a computational subspace $\mathcal{H}_C$ of dimension $d$ and its orthogonal complement $\mathcal{H}_{NC}$. Let $T$ be a unitary block-diagonalizing transformation satisfying the least-action (LAUT) criterion, i.e., $T$ minimizes the Frobenius norm $\|T - \mathcal{I}\|_F$ among all unitaries that render $T^\dagger H T$ block-diagonal with respect to $\mathcal{H} = \mathcal{H}_C \oplus \mathcal{H}_{NC}$. Define the effective Hamiltonian in $\mathcal{H}_C$ as $H_{\mathrm{eff}} = P (T^\dagger H T) P$, where $P$ is the projector onto $\mathcal{H}_C$.
Let $U(t) = e^{-i H t}$ and $U_{\mathrm{eff}}(t) = e^{-i H_{\mathrm{eff}} t}$ be the unitary time-evolution operators generated by $H$ and $H_{\mathrm{eff}}$, respectively.

\paragraph*{Trace fidelity as a dynamical benchmark.-}

To evaluate the quality of effective dynamics, we use the trace fidelity:
\begin{equation}
    F_{\mathrm{tr}}(t) = \frac{1}{D} \left| \Tr\left[ U_{\mathrm{eff}}^\dagger(t) U(t) \right] \right|,
\end{equation}
which quantifies the average overlap between the full and effective propagators over the full Hilbert space of dimension \( D \). Its long-time average,
\begin{equation}\label{eq:long_time_averaged_fidelity}
    \langle F_{\mathrm{tr}} \rangle_t := \lim_{T \to \infty} \frac{1}{T} \int_0^T F_{\mathrm{tr}}(t)\, dt,
\end{equation}
serves as a robust metric of global dynamical consistency. This trace fidelity is directly related to the average fidelity over all pure input states~\cite{pedersen2007fidelity}, a standard figure of merit in quantum simulation, randomized benchmarking~\cite{knill2008randomized}, and Hamiltonian learning~\cite{poulin2011quantum}. 

\paragraph*{Assumptions and scope of the fidelity bound.-}
The bound below relies on mild spectral and labeling assumptions that ensure a consistent comparison between the exact eigenbasis $S$ and the LAUT-optimized basis $\mathcal{E}$. First, we assume that the spectrum of $H$ admits a well-defined eigenbasis such that accidental degeneracies do not introduce ambiguities in matching the diagonal terms that survive long-time averaging. Exact or near-degenerate \emph{states within} a block are allowed, while persistent cross-block degeneracies may invalidate the diagonal-time average simplification. Second, we assume that the eigenbases $S$ and $\mathcal{E}$ can be continuously aligned along the parameter manifold, which holds generically except at true eigenvalue crossings. Finally, the long-time average is taken over a window $T$ satisfying $T \gg 1/|E_m - E_n|$ for all $m\neq n$, so that oscillatory off-block contributions $\exp[i(E_n - E_m)t]$ average out. These assumptions are standard in time-averaged fidelity analyzes and hold throughout all benchmark regimes considered in this work.

\begin{theorem}[\textbf{Geometric Lower Bound on Dynamical Fidelity}]
\label{thm:LAUT_bound_trace_fidelity}
Under the assumptions above, the least-action (LAUT) block-diagonalizing transformation $T$ yields an effective
Hamiltonian $H_{\mathrm{eff}}=T^\dagger H T$ whose long-time averaged trace fidelity satisfies
\begin{equation}
    \left\langle F_{\mathrm{tr}} \right\rangle_t
    \;\geqslant\;
    \left( 1 - \frac{1}{2D}\,\|T - \mathcal{I}\|_F^2 \right)^2,
    \label{eq:key_bound}
\end{equation}
where $D$ is the Hilbert-space dimension.
The right-hand side becomes an asymptotically accurate estimate in the perturbative regime
$\|T-\mathcal{I}\|_F \ll 1$, where all inequalities used in the derivation coincide up to
second order in the singular-value deviations. No general tightness condition is implied beyond this perturbative limit.
\end{theorem}

\textit{Sketch of the proof.-} Let $H=S E S^\dagger$ and $H_{\mathrm{eff}}=\mathcal{E} E \mathcal{E}^\dagger$ share the eigenvalue spectrum $E$, and define
$W = \mathcal{E}^\dagger S$. The exact and effective propagators satisfy
\[
U(t)=S e^{-i E t} S^\dagger,\qquad 
U_{\mathrm{eff}}(t)=\mathcal{E} e^{-i E t} \mathcal{E}^\dagger,
\]
from which
\[
\Tr\!\left[ U^\dagger(t) U_{\mathrm{eff}}(t) \right]
= \sum_{m,n} e^{i(E_n - E_m)t} |W_{mn}|^2.
\]
By the long-time averaging assumption, all oscillatory off-diagonal terms vanish, giving
\[
\Big\langle \Tr(U^\dagger U_{\mathrm{eff}}) \Big\rangle_t
= \sum_n |W_{nn}|^2.
\]
Applying Jensen’s inequality and Cauchy–Schwarz yields
\[
\left\langle \left| \Tr(U^\dagger U_{\mathrm{eff}}) \right| \right\rangle_t
\geqslant \frac{1}{D}\left| \Tr(W) \right|^2.
\]
Let the singular value decomposition be $S = U_0 \Sigma V_0^\dagger$.
The LAUT principle implies $\mathcal{E} = U_0 V_0^\dagger$, so $W = \mathcal{E}^\dagger S$ has diagonal entries $\sigma_i$, giving
\[
\Tr(W)=\sum_i \sigma_i,
\qquad
\|T - \mathcal{I}\|_F^2 = 2D - 2 \sum_{i=1}^D \sigma_i.
\]
A direct algebraic manipulation then leads to the bound in Eq.~\eqref{eq:key_bound}.
A full proof is provided in Appendix~\ref{app:fidelity_bound}. \hfill$\blacksquare$

While the Theorem holds mathematically for any LAUT $T$, its physical utility is strictly constrained by the energy scale separation between the target and leakage subspaces. In the strong-coupling or resonant regime (e.g., $g/\Delta \gtrsim 1$), this scale separation collapses, leading to extensive hybridization that renders the concept of an isolated ``computational subspace" physically ill-defined. As a result, in this limit, the required transformation $T$ deviates significantly from the identity ($||T-\mathcal{I}||_F^2 \sim \mathcal{O}(D)$), causing the fidelity lower bound in Eq.~\eqref{eq:key_bound} to become loose and uninformative. Thus, the bound—and the effective Hamiltonian description itself—is maximally relevant in the dispersive and quasi-dispersive regimes, where the scale separation is preserved and the LAUT framework demonstrates its distinct advantage over conventional perturbative methods.

To provide intuition for the Theorem, 
Fig.~\ref{Fig:Illustrating_bound} shows a minimal three-level Q–C–Q example in which both sides of Eq.~\eqref{eq:key_bound} are evaluated across a broad frequency sweep. The inequality is satisfied throughout, and the two curves converge as $\|T-\mathcal{I}\|_F^2 \rightarrow 0$. The inset reveals a clear quadratic dependence of the pointwise difference between the two sides in the small-rotation regime, consistent with the qualitative expectation that the bound tightens smoothly as the transformation approaches the identity. Additional numerical details and the comparison between symmetric and weakly asymmetric configurations are provided in 
Appendix~\ref{app:bound_verification}.
\begin{figure}[t]
    \centering
    \includegraphics[width=0.7\columnwidth]{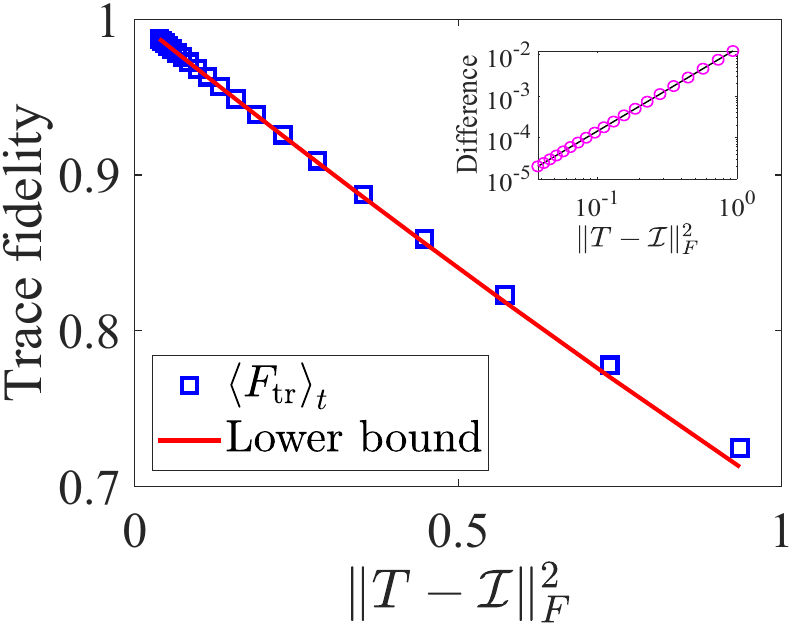}
    \caption{
Illustration of the LAUT-based fidelity bound in a minimal three-level model. The left- (blue squares) and right- (red line) hand sides of Eq.~(\ref{eq:key_bound}) are shown as functions of the coupler frequency in a symmetric configuration ($\omega_1/2\pi=\omega_2/2\pi=4~\mathrm{GHz}$). The inequality holds for all parameters, and the two curves converge as $\|T-\mathcal{I}\|_F^2 \rightarrow 0$, reflecting the expected tightening of the bound in the small-rotation regime. The inset displays the point-wise difference (pink circles) between the two sides of the inequality; the fitted line (black) exhibits a clear quadratic trend.}
    \label{Fig:Illustrating_bound}
\end{figure}

The bound shows that the long-time dynamical accuracy of the effective Hamiltonian is controlled entirely by the eigenbasis mismatch quantified through $\|T-\mathcal{I}\|_F^2$, independent of the fast oscillatory coherence terms that wash out under time averaging. Consequently, minimizing this Frobenius deviation is equivalent to maximizing the dynamical fidelity guaranty. This result provides the formal justification for adopting the LAUT criterion as \textit{the definition of the effective Hamiltonian}: the LAUT-selected transformation is precisely the one that optimizes the lower bound of long-time dynamical fidelity among all admissible block-diagonalizing unitaries.

This connection is particularly relevant for driven multilevel operations, such as the CR gate~\cite{magesan2020effective}, where non-rotating-wave-approximation (RWA) processes and higher-level hybridization substantially limit perturbative approaches. The LAUT framework naturally captures these contributions, providing a unified and robust tool for Hamiltonian design and entangling-gate calibration, as demonstrated in Sec.~\ref{sec:benchmarking}.

\subsection{Symmetry Preservation Under the LAUT Principle}
\label{subsec:symmetry_preserve}

Having established the LAUT criterion as the rigorous standard for structural fidelity, a critical question arises: does this construction respect the fundamental symmetries of the original system? Symmetry preservation is particularly paramount in analog quantum simulators and the study of symmetry-protected topological phases, where any violation of conserved quantities can induce unphysical symmetry-breaking interactions.

Here, we demonstrate that the LAUT framework \textit{inherently guarantees} symmetry preservation whenever the block partition is symmetry-adapted ($[\mathcal{S}, P]=0$). 

Crucially, this property stems directly from the \textit{uniqueness} of the variational solution (see Corollary~\ref{cor:symmetry}), a feature not shared by procedural numerical approaches. For instance, local optimization methods such as sequential Givens rotations (GR) often introduce arbitrary gauge phases that disrupt global symmetries; we provide concrete counterexamples illustrating this violation in Appendix~\ref{app:symmetry_GR}.

\begin{corollary}
[\textbf{Symmetry Preservation under Least‑Action Gauge Fixing}]
\label{cor:symmetry}
Let $\mathcal{S}$ be a unitary symmetry of the system such that $[\mathcal{S}, H] = 0$. If the symmetry respects the partitioning of the Hilbert space (i.e., $[\mathcal{S}, P] = 0$), then the optimal transformation $T_{\mathrm{opt}}$ defined by the LAUT principle commutes with the symmetry:
\begin{equation}
    [\mathcal{S}, T_{\mathrm{opt}}] = 0.
\end{equation}
Consequently, the resulting effective Hamiltonian preserves the symmetry: $[H_{\mathrm{eff}}, P\mathcal{S}P] = 0$.
\end{corollary}

\begin{proof}
The proof relies on the uniqueness of the variational solution and the unitary invariance of the Frobenius norm. Let $T_{\mathrm{opt}}$ be the unique minimizer of the action $\mathcal{A}(T) = \|T - \mathcal{I}\|_F$ within the manifold of block-diagonalizing unitaries $\mathcal{T}_{\mathrm{BD}}$.

Consider the rotated transformation $T' \equiv \mathcal{S}^\dagger T_{\mathrm{opt}} \mathcal{S}$. We verify two properties of $T'$:
\begin{enumerate}[(1)]
    \item \textbf{Validity:} Since $H$ and $P$ commute with $\mathcal{S}$, the block-diagonal structure is invariant under conjugation by $\mathcal{S}$. Specifically,
    \begin{equation}
        (T')^\dagger H T' = \mathcal{S}^\dagger (T_{\mathrm{opt}}^\dagger H T_{\mathrm{opt}}) \mathcal{S}.
    \end{equation}
    Because $T_{\mathrm{opt}}^\dagger H T_{\mathrm{opt}}$ is block-diagonal and $[\mathcal{S}, P]=0$, the product remains block-diagonal. Thus, $T' \in \mathcal{T}_{\mathrm{BD}}$.

    \item \textbf{Optimality:} The Frobenius norm is unitarily invariant ($\|U A V\|_F = \|A\|_F$). The action of the rotated transformation is:
    \begin{align}
        \|T' - \mathcal{I}\|_F &= \|\mathcal{S}^\dagger T_{\mathrm{opt}} \mathcal{S} - \mathcal{S}^\dagger \mathcal{I} \mathcal{S}\|_F \nonumber \\
        &= \|\mathcal{S}^\dagger (T_{\mathrm{opt}} - \mathcal{I}) \mathcal{S}\|_F \nonumber \\
        &= \|T_{\mathrm{opt}} - \mathcal{I}\|_F.
    \end{align}
    Thus, $T'$ achieves the same minimal action value as $T_{\mathrm{opt}}$.
\end{enumerate}
Since the LAUT criterion identifies a unique variational solution (as established in Sec.~\ref{subsec:definition_of_EH}), the optimal transformation $T_{\mathrm{opt}}$ must inherently respect any symmetry $\mathcal{S}$ compatible with the partitioning ($[\mathcal{S}, P]=0$). If it did not ($[T_{\mathrm{opt}}, \mathcal{S}] \neq 0$), the symmetry operation would generate a distinct transformation $T' = \mathcal{S}^\dagger T_{\mathrm{opt}} \mathcal{S}$ with identical geometric action, contradicting the proven uniqueness.
\end{proof}


In contrast, numerical approaches relying on local operations—such as sequential Givens rotations or Schur complements—often traverse the solution manifold arbitrarily. These methods may disrupt symmetries even when the block partition nominally respects them.
Therefore, we conclude that a rigorous \textbf{one-to-one correspondence} between the mathematical procedure of block diagonalization and the physical construction of a symmetry-preserving effective Hamiltonian is established \textit{if and only if} the transformation is constrained by the LAUT criterion.

\subsection{Perturbative Block Diagonalization Aligned with LAUT}
\label{subsec:pbd}

With the LAUT framework and its symmetry-preserving advantages in place, a natural question arises: how can one construct effective Hamiltonians analytically while remaining aligned with the LAUT principle?
Exact block diagonalization (EBD-LAUT) provides numerically exact results but suffers from a high computational cost and limited analytical transparency.
Perturbative block diagonalization (PBD) addresses these limitations by offering compact, analytically tractable expressions that are widely used in theoretical modeling and experimental calibration.
Its applicability, however, is contingent on perturbation strength, expansion order, and the spectral properties of the unperturbed Hamiltonian.
We begin by summarizing these validity conditions and discussing the known limitations of the commonly employed SWT; technical details are deferred to Appendix~\ref{app:pbd-validity}.

Perturbative block diagonalization remains valid when the perturbation strength $\varepsilon$ satisfies $|\varepsilon| \leqslant \varepsilon_c = \Delta / (2|V|)$, where $\Delta$ denotes the spectral gap separating the target (low-energy) manifold from the rest of the spectrum~\cite{bravyi2011schrieffer}.
In the SWT, the effective Hamiltonian is constructed through a unitary transformation $U = e^{\mathcal{A}}$, where the anti-Hermitian generator $\mathcal{A}$ is perturbatively expanded to eliminate couplings between the low- and high-energy subspaces~\cite{bravyi2011schrieffer}.
Although the direct rotation condition underlying the SWT is formally equivalent to the geometric LAUT criterion~\cite{bravyi2011schrieffer}, the perturbative series of $\mathcal{A}$ does not, in general, satisfy the LAUT principle.
Recent analyzes further show that the SWT expansion deviates from the LAUT expansion starting at the third order~\cite{mankodi2024perturbative}.
Despite its widespread use, the SWT exhibits several notable drawbacks:
\begin{enumerate}[(1)]
    \item The complexity of computing $\mathcal{A}$ grows factorially with order~\cite{wurtz2020variational, bravyi2011schrieffer}, making manual derivations impractical and automated implementations computationally expensive~\cite{li2022nonperturbative}.
    \item SWT can introduce spurious singularities that are absent in the exact dynamics, potentially leading to unphysical predictions~\cite{sanz2016beyond}.
    \item SWT is primarily designed for two-block decompositions (e.g., separating low- and high-energy sectors). Extending it to multi-block scenarios, such as full block diagonalization across multiple manifolds, requires significant modifications and imposes heavy computational overhead~\cite{hormann2023projective, araya2025pymablock}.
\end{enumerate}

These limitations highlight the need for alternative perturbative frameworks that align with the LAUT principle while offering improved tractability and generality.

A natural candidate for a perturbative scheme aligned with the LAUT principle is the Bloch--Brandow (BB) formalism~\cite{bloch1958theorie,killingbeck2003bloch,shavitt2009many,brandow2005effective}, which constructs a similarity transformation that decouples a model space $P$ from its complement $Q = \mathcal{I} - P$. In the absence of unitarity constraints, the BB solution coincides with the LAUT result; however, the resulting effective Hamiltonian is inherently non-Hermitian and typically requires post hoc symmetrization. Notably, it has been established that the symmetrized form remains algebraically equivalent to the canonical Hermitian expansion up to the fourth order~\cite{richert1976comparison}.

Crucially, Takayanagi's refinements~\cite{takayanagi2013study,takayanagi2016effective,takayanagi2020linked} introduced efficient recursive expansions, widely used in nuclear and quantum chemistry for incorporating virtual excitations and deriving closed-form effective interactions. For a detailed review, including projection operators and eigenstate decompositions, see Appendix~\ref{app:bb-formalism}. 

While these techniques ensure analytical tractability, a pivotal advantage of adopting the BB formalism within the LAUT framework is its strict adherence to physical conservation laws. We formalize this property in the following corollary:
\begin{corollary}\label{prop:bb_symmetry}
    \textbf{(Symmetry Preservation of the BB Formalism).} The Bloch-Brandow effective Hamiltonian $H_{\mathrm{eff}}^{\mathrm{BB}}$ preserves a system symmetry $\mathcal{S}$ within the target subspace if and only if the projector $P$ commutes with the symmetry operator:
\begin{equation}
    [\,H_{\mathrm{eff}}^{\mathrm{BB}},\,\mathcal{S}\,]=0 \;\Longleftrightarrow\; [\,P,\mathcal{S}\,]=0.
\end{equation}
\end{corollary}
This corollary establishes that symmetry-adapted partitioning is sufficient to guarantee a physically consistent effective model. The rigorous proof of Corollary~\ref{prop:bb_symmetry} is provided in Appendix~\ref{app:symmetry_bb}.

Unlike the SWT, the perturbative block diagonalization via the Bloch--Brandow formalism (PBD-BB) approach satisfies the LAUT criterion by construction without enforcing unitarity, thereby simplifying implementation and enabling systematic extensions. Its analytical tractability and symmetry compatibility make it particularly attractive for superconducting circuits and quantum many-body systems, where multiple energy manifolds are weakly coupled or symmetry-partitioned.
This BB approach forms the perturbative counterpart in our unified LAUT-based framework, enabling analytic EHs that complement the non-perturbative EBD-LAUT method.

\subsection{Verification of Effective Hamiltonian}

While the LAUT principle provides a rigorous theoretical lower bound on dynamical fidelity, practical applications require a more sensitive diagnostic. The asymptotic bound becomes loose in the strong-coupling regime ($||T-\mathcal{I}||_F^2 \sim \mathcal{O}(D)$), where the distinction between ``valid" and ``invalid" effective models requires precise quantification. This motivates the introduction of a refined, physically meaningful metric that directly quantifies information retention and dynamical accuracy within the computational subspace.

To address this, we introduce a static, fidelity-inspired soundness metric $I(H_{\mathrm{eff}})$, derived from the subspace information retention~\cite{pedersen2007fidelity}:
\begin{equation}
    I(H_{\mathrm{eff}}) = \frac{1}{d(d+1)} \left[ \Tr(M M^\dagger) + |\Tr(M)|^2 \right],
    \label{eq:soundness_metric}
\end{equation}
where $d = \dim(\mathcal{H}_{\mathrm{C}})$, and $M = U_{\text{ideal}}^\dagger P U_{\text{real}} P$ capture the projection overlap onto $\mathcal{H}_{\mathrm{C}}$.
By identifying $U_{\text{ideal}}$ and $U_{\text{real}}$ with the eigenbases of $H_{\mathrm{eff}}$ and $H$, respectively, we derive the explicit form:
\begin{equation}
\label{eq:measure_of_EH}
\begin{aligned}
    I(\Heff) =& \frac{1}{d(d+1)} \Bigg[ \operatorname{Tr}\left(P S_{\text{BD}} S_{\text{BD}}^\dagger P\right) \\
    &+ \left| \operatorname{Tr}\left(P (S_{\text{BD}} S_{\text{BD}}^\dagger)^{1/2} P \right) \right|^2 \Bigg],
\end{aligned}
\end{equation}
where $S_{\text{BD}}$ denotes the block-diagonal component of the full eigenbasis. Physically, the first term quantifies the population retained within the target subspace after the basis transformation, while the second term captures the preservation of relative phase coherence. This metric satisfies $I(\Heff) \in [0,1]$, with $I(\Heff)=1$ indicating perfect recovery of the original dynamics within $\mathcal{H}_{\mathrm{C}}$.

Functionally, this metric serves as a \textbf{post hoc complement} to the variational LAUT criterion:
\begin{itemize}
    \item The \textbf{LAUT principle} guides the \textit{construction} of $H_{\mathrm{eff}}$ by minimizing geometric distortion.
    \item The \textbf{Soundness metric} $I(H_{\mathrm{eff}})$ certifies the \textit{quality} of the result, sensitively detecting leakage and hybridization errors that escape purely spectral analysis.
\end{itemize}
This dual structure enables a unified workflow for designing and benchmarking effective Hamiltonians in realistic quantum hardware. As demonstrated in Sec.~\ref{sec:benchmarking}, $I(\Heff)$ strongly correlates with the time-averaged fidelity $\bar{\mathcal{F}}$. This establishes it as a reliable and computationally efficient proxy for global dynamical accuracy (averaged over all initial states), eliminating the need for explicit time evolution.

\paragraph*{Practical considerations and numerical implementation.-}
The validity of the metrics introduced above relies on the consistent identification of the target subspace across varying parameter regimes. Ensuring the adiabatic continuity of eigenstate labeling is critical; misidentification—often triggered by level crossings—can lead to unphysical discontinuities in $H_{\mathrm{eff}}$ or erroneous assessments of leakage errors. While recent algorithmic advances~\cite{goto2024labeling} offer strategies for automated labeling, robustly handling highly hybridized spectra in large Hilbert spaces remains a significant challenge.

In addition, computational scalability warrants attention. Although the static metric $I(H_{\mathrm{eff}})$ avoids the prohibitive cost of explicit time-domain simulations, it still necessitates obtaining the relevant eigenstates of the full Hamiltonian. The cost of this partial diagonalization typically dominates the workflow for large-scale systems. Future research may address this bottleneck by exploring low-rank approximations or randomized trace estimators to evaluate $I(H_{\mathrm{eff}})$ without full spectral decomposition.

\subsection{Protocol for Solving Effective Hamiltonians}
\label{subsec:BD_protocol}
Building on the least-action framework, we present a practical protocol for systematically constructing effective Hamiltonians. These reduced models provide a principled way to simplify complex quantum many-body systems by focusing on dynamically relevant subspaces. Across condensed matter theory and quantum information science, effective Hamiltonians enable the elimination of virtual processes and irrelevant degrees of freedom, yielding low-dimensional models that retain essential physical content. While perturbative techniques such as the SWT and BB formalism have long been employed for this purpose, their systematic integration with modern symmetry-guided and fidelity-driven criteria remains underexplored. A unified protocol that reconciles these objectives is particularly valuable for preserving emergent phenomena, including symmetry-protected phases and collective excitations.

In practical quantum platforms—such as superconducting circuits, trapped ions, or Rydberg arrays—couplings to ancillary modes, higher excited states, or symmetry-breaking imperfections are often non-negligible. Neglecting these effects can lead to leakage, residual interactions, or distortions of target dynamics. Our protocol addresses this challenge by iteratively constructing effective Hamiltonians via PBD-BB while explicitly maintaining dynamical fidelity and symmetry preservation. The resulting models are analytically tractable, experimentally relevant, and scalable, enabling quantitatively accurate modeling of complex quantum systems.

\subsubsection*{Iterative Block Diagonalization Strategy}

\textit{Objective: Transform the microscopic Hamiltonian into an approximately block-diagonal form while faithfully reproducing the original dynamics.}

Let \( H_{\mathrm{full}} \) be the microscopic Hamiltonian and \( \{H_{\mathrm{ideal}}\} \) the set of block-diagonal Hamiltonians consistent with a symmetry group \( \mathcal{S} \). The target effective Hamiltonian satisfies
\begin{equation}
    H_{\mathrm{eff}} \in \{H_{\mathrm{ideal}}\},
\end{equation}
and accurately approximates dynamics within a designated computational subspace \( \mathcal{H}_C \). This is achieved via a sequence of unitary transformations \( \{T^{[n]}\} \), each designed to target and eliminate a particular class of undesired terms in the Hamiltonian. By addressing system–ancilla couplings, symmetry-breaking perturbations, and other off-block components in turn, the transformation sequence yields a progressively block-diagonal structure.
\begin{equation}
    H^{[n+1]} = [T^{[n]}]^\dagger H^{[n]} T^{[n]}, \quad H^{[0]} = H_{\mathrm{full}}.
\end{equation}
Each \( T^{[n]} \) is derived perturbatively, and under the LAUT criterion, the sequence converges toward an optimal block-diagonal approximation of \( H_{\mathrm{full}} \).

\subsubsection*{Step 1: Ancillary Subsystem Decoupling}

\textit{Objective: Remove hybridization between the system and ancillary degrees of freedom.}

Ancillary components such as couplers or resonators facilitate tunable interactions but may induce leakage or unwanted entanglement. We isolate the computational subspace
\begin{equation}
    \mathcal{H}_C = \mathrm{span}\{\ket{\phi_i}_S \otimes \ket{0}_A\},
\end{equation}
where \( \ket{0}_A \) is the ground state of the ancilla. The total Hamiltonian takes the form
\begin{equation}
    H_{\text{full}} = H_S \otimes \mathcal{I}_A + \mathcal{I}_S \otimes H_A + \varepsilon_1 V_{\text{int}},
\end{equation}
with \( \varepsilon_1 \ll 1 \) controlling the interaction strength. Applying a PBD transformation yields an effective Hamiltonian \( H_{\text{decoupled}} \) that eliminates ancilla excitations. Perturbative validity requires
\begin{equation}
    |\varepsilon_1| \leq \varepsilon_c, \quad \varepsilon_c = \frac{\Delta}{2 \|V_{\text{int}}\|},
\end{equation}
where \( \Delta \) denotes the spectral gap separating \( \mathcal{H}_C \) from its orthogonal complement.

\subsubsection*{Step 2: Interband Coupling Elimination}

\textit{Objective: Suppress transitions between symmetry sectors to preserve the block structure.}

After decoupling the ancilla modes, the reduced Hamiltonian reads:
\begin{equation}
    H_{\text{decoupled}} = H_0 + \varepsilon_2 V_2,
\end{equation}
where \( H_0 \) respects a symmetry decomposition \( \mathcal{H} = \bigoplus_\lambda \mathcal{H}_\lambda \), and \( V_2 \) introduces weak interband coupling. Applying a second BDT, the effective Hamiltonian becomes
\begin{equation}
    H_{\mathrm{eff}} = \sum_\lambda P_\lambda H_{\mathrm{eff}} P_\lambda,
\end{equation}
with \( P_\lambda \) projecting onto sector \( \mathcal{H}_\lambda \). Perturbative consistency requires
\begin{equation}
    |\varepsilon_2| \leq \varepsilon_c, \quad \varepsilon_c = \frac{\Delta_2}{2 \|V_2\|},
\end{equation}
where \( \Delta_2 \) is the minimal interband gap. This step preserves the original symmetry structure imposed by \( G_0 \) and yields a symmetry-respecting \( H_{\mathrm{eff}} \) with renormalized parameters.

\subsubsection*{Step 3: Validity Regime Assessment}

\textit{Objective: Identify the conditions under which the perturbative expansion remains reliable.}

Block-diagonal construction becomes unreliable in the presence of near-degeneracies, typically when
\begin{equation}
    |g_{ij}/(E_i^{(0)} - E_j^{(0)})| \geqslant 1, \quad \ket{i} \in \mathcal{H}_C,\ \ket{j} \in \mathcal{H}_{\mathrm{NC}},
\end{equation}
and $g_{ij}$ denotes the coupling between $\ket{i}$ and $\ket{j}$.
Such resonance conditions deteriorate the convergence of perturbative expansions due to small energy denominators. We identify the onset of breakdown both analytically—by inspecting perturbation denominators—and numerically, using the fidelity-based indicator \( I(H_{\mathrm{eff}}) \) defined in Eq.~\eqref{eq:measure_of_EH}, which quantifies how faithfully \( H_{\mathrm{eff}} \) captures the full system dynamics over long timescales. When necessary, the EBD-LAUT method provides a robust, non-perturbative alternative capable of addressing strongly hybridized or near-resonant scenarios.

\subsubsection*{Step 4: Dynamical Fidelity Benchmarking}

\textit{Objective: Quantify how well the effective Hamiltonian captures real-time evolution.}

We compare the time evolution under \( H_{\mathrm{eff}} \) and \( H_{\mathrm{full}} \):
\begin{align}
    i \frac{d}{dt} \ket{\psi_{\mathrm{eff}}(t)} &= H_{\mathrm{eff}} \ket{\psi_{\mathrm{eff}}(t)}, \\
    i \frac{d}{dt} \ket{\psi_{\mathrm{exact}}(t)} &= H_{\mathrm{full}} \ket{\psi_{\mathrm{exact}}(t)}.
\end{align}
Projecting the full dynamics onto \( \mathcal{H}_C \), define
\begin{equation}
    \ket{\psi_p(t)} = P_C \ket{\psi_{\mathrm{exact}}(t)}.
\end{equation}
The instantaneous and time-averaged fidelities are
\begin{equation}
    \mathcal{F}(t) = \left| \braket{\psi_p(t)}{\psi_{\mathrm{eff}}(t)} \right|^2, \quad
    \bar{\mathcal{F}} = \frac{1}{T} \int_0^T \mathcal{F}(t)\, \mathrm{d}t.
\end{equation}
High values of \( \bar{\mathcal{F}} \) confirm that \( H_{\mathrm{eff}} \) faithfully capture the full dynamics within the desired subspace.

\subsubsection*{Outlook: Embedding into Many-Body Chains}

\textit{Extension: Generalize the protocol to one-dimensional quantum simulators via bootstrap embedding.}

The above strategy naturally extends to chain geometries using bootstrap embedding theory (BET)~\cite{welborn2016bootstrap,ye2019bootstrap,liu2023bootstrap}. The steps are:
\begin{enumerate}
    \item \textbf{Local Construction:} Derive few-site effective Hamiltonians using PBD-BB or EBD-LAUT.
    \item \textbf{Bootstrap Embedding:} Embed the local EHs into a larger environment that is iteratively updated by local observables.
    \item \textbf{Global Convergence:} Iterate until a self-consistent many-body EH is obtained.
\end{enumerate}
This framework enables scalable and physically transparent modeling of strongly correlated systems with tunable local interactions.

\paragraph*{Summary.}
This section has outlined a general protocol for effective Hamiltonian construction via block diagonalization, along with two representative applications: ancillary subsystem decoupling and interband coupling elimination. These examples illustrate how the approach systematically simplifies complex quantum systems while preserving dynamical and symmetry fidelity, paving the way for scalable and accurate simulation frameworks.

\section{Benchmarking, Validation and Application of the LAUT Framework}
\label{sec:benchmarking}

To evaluate the practical performance of the least-action (LAUT) framework, we benchmark a set of effective-Hamiltonian (EH) constructions on experimentally relevant superconducting-circuit models.  Using the three-mode qubit–coupler–qubit (Q–C–Q) architecture as a unified platform (Appendix~\ref{app:QCQ model}), we compare exact block diagonalization via the least-action criterion (EBD-LAUT), its perturbative implementation of the Bloch–Brandow formalism (PBD-BB), and two commonly used non-variational techniques: Schrieffer–Wolff transformations (SWT) and Jacobi-style Givens rotations (GR). Our comparison emphasizes three desiderata for an EH method: \emph{accuracy}, \emph{symmetry preservation}, and \emph{robustness}, and it leverages the fidelity and soundness diagnostics introduced in Sec.~\ref{sec: EH theory and application}.

The benchmarks are organized into three layers.  
(i) An \emph{exact-level} comparison of the Q--C--Q model highlights structural differences between variational and nonvariational constructions and demonstrates the importance of symmetry constraints.  
(ii) An \emph{experiment-level} validation using the driven cross-resonance (CR) gate shows how LAUT-consistent perturbative expansions remain accurate under experimentally relevant strong drives.  
(iii) A \emph{beyond-RWA} analysis examines counter-rotating and multi-photon processes, illustrating how PBD-BB captures renormalizations that standard RWA-based models miss.

We employ Fourier-based spectral extraction (FFT analysis) of population dynamics as a numerical and experimental benchmark for coupling strengths, where applicable. FFT is used solely as a diagnostic tool, not as an EH construction method; detailed derivations are given in Appendix~\ref{app:effective_coupling_FFT}. For analytical consistency, Appendix~\ref{app:three_level} presents a three-level derivation verifying that PBD-BB reproduces the fourth-order LAUT expansion exactly, while SWT begins to deviate at that order.

\subsection{Exact EH Solving Comparison}
\label{subsec:exact_EH_benchmark}
To assess the structural and quantitative advantages of the LAUT framework, we benchmark it against two widely used non-variational approaches: (i) dynamical spectral extraction via Fourier analysis~\cite{sung2021realization}, used here purely as a numerical reference, and (ii) Jacobi-style block diagonalization based on Givens rotations (GR)~\cite{li2022nonperturbative}.  
All comparisons are performed on the full microscopic Hamiltonian of the three-mode Q--C--Q circuit,
\begin{equation}\label{eq:two_qubit_system}
\begin{aligned}
H_{2Q} = &\sum_{i\in\{1,2,c\}}
\left(\omega_i a^{\dagger}_i a_i
+ \frac{\beta_i}{2}\, a^{\dagger}_i a^{\dagger}_i a_i a_i\right)
\\[1mm]
&- g_{12}\,(a_1 - a_1^{\dagger})(a_2 - a_2^{\dagger})
\\[1mm]
&- \sum_{j\in\{1,2\}} g_j\,
(a_j - a_j^{\dagger})(a_c - a_c^{\dagger}),
\end{aligned}
\end{equation}
with parameters chosen to match typical superconducting-qubit devices.

\paragraph*{Symmetry preservation.-}
While Givens rotations (GR) provide a numerically robust and convergent approach for exact block diagonalization~\cite{li2022nonperturbative}, they typically navigate the unitary manifold without an explicit geometric constraint to fix the gauge freedom. 
Consequently, the resulting effective model may not strictly inherit the discrete symmetries of the microscopic system. 
We quantify this spurious symmetry breaking using the frequency mismatch metric $\zeta=|\tilde{\omega}_1-\tilde{\omega}_2|/2\pi$. 
In a perfectly symmetric configuration~\footnote{Parameters: $\omega_{1,2}/2\pi=4.16~\mathrm{GHz}$, $\beta_{1,2}/2\pi=-220~\mathrm{MHz}$, $g_{1,2}/2\pi=72.5~\mathrm{MHz}$, $\beta_c/2\pi=-90~\mathrm{MHz}$, and $g_{12}/2\pi=5~\mathrm{MHz}$.}, we find that EBD-LAUT yields identical effective qubit frequencies, thereby exactly preserving the exchange symmetry ($\zeta=0$). 
In contrast, standard GR algorithms generally yield $\tilde{\omega}_1\neq \tilde{\omega}_2$ ($\zeta > 0$). 
This establishes that the iterative elimination procedure in GR introduces a path-dependent gauge bias, spuriously breaking the permutation invariance. 
Finally, time-averaged fidelity benchmarks [Fig.~\ref{subfig:Fidelity_verification}] confirm that while GR produces valid spectral results, it violates the rigorous fidelity lower bound predicted by the Theorem in Sec.~\ref{thm:LAUT_bound_trace_fidelity}. 
This indicates that the bound is intrinsic to the geometric least-action path, which EBD-LAUT uniquely satisfies.

\begin{figure}[tb]
    \centering
    \subfigure{
        \begin{minipage}[b]{0.46\columnwidth}
            \centering
            \begin{overpic}[scale=0.30]{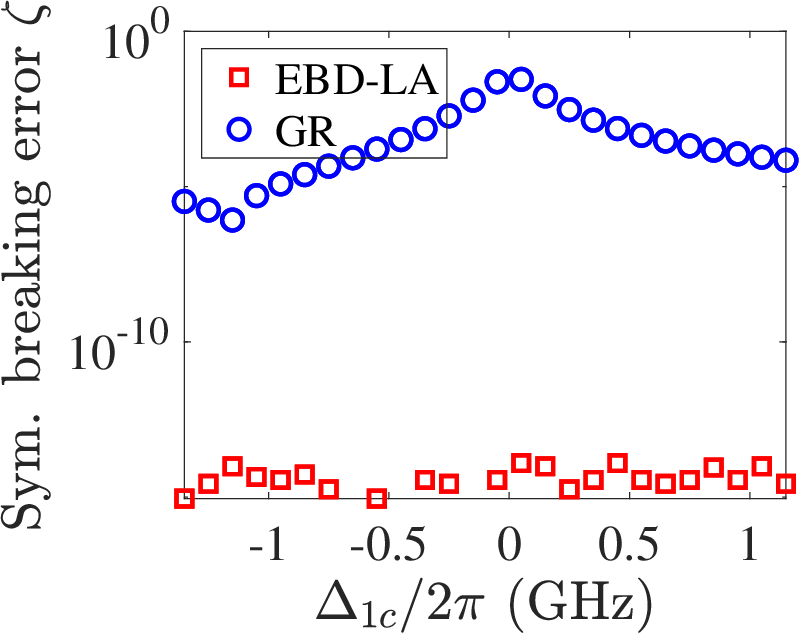}
                \put(0,82){\textbf{(a)}}
            \end{overpic}
        \end{minipage}\label{subfig:Symmetry_verification}
    }
    \subfigure{
        \begin{minipage}[b]{0.46\columnwidth}
            \centering
            \begin{overpic}[scale=0.39]{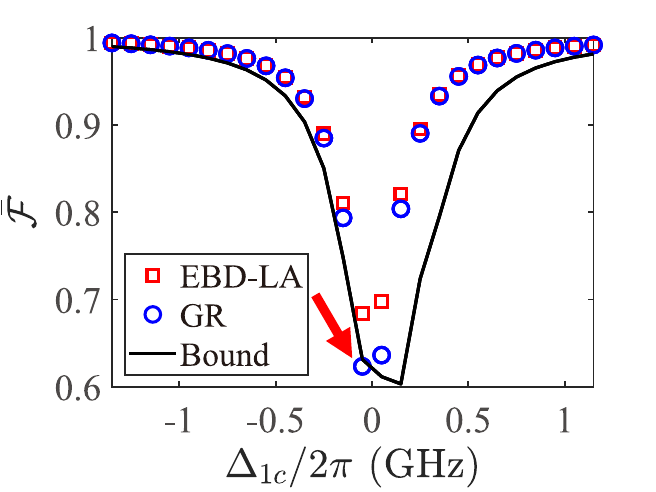}
                \put(0,75){\textbf{(b)}}
            \end{overpic}\label{subfig:Fidelity_verification}
        \end{minipage}
    }
    \caption{(Color online) Analysis of symmetry preservation and transformation fidelity. \textbf{(a)}~The symmetry-breaking metric $\zeta$. LAUT maintains exact symmetry ($\zeta \approx 0$) by construction, while GR introduces spurious asymmetry. \textbf{(b)}~LAUT yields superior fidelity in the near-resonant regime, strictly saturating the analytical lower bound of the Theorem. GR, lacking variational optimality, falls below this limit (red arrow). Data are shown versus the qubit--coupler detuning.}
\end{figure}

\paragraph*{Effective coupling extraction.-}
Figure~\ref{fig:compare_LAUT} benchmarks the effective coupling $2\tilde{g}$ derived from the LAUT and GR frameworks against a numerically exact FFT reference. 
Throughout this paper, we employ the notation $\eta_{X}$ to denote the relative deviation of a physical quantity $X$. Accordingly, the coupling accuracy is quantified by the relative error $\eta_{g} = |(\tilde{g}-\tilde{g}_{\mathrm{FFT}})/\tilde{g}_{\mathrm{FFT}}|$.
As illustrated in Fig.~\ref{fig:compare_LAUT}, LAUT maintains superior accuracy across the parameter space. Crucially, in the quasi-dispersive regime ($\Delta_{1c} \to 0$), where hybridization is pronounced, GR exhibits systematic deviations. This discrepancy arises because GR lacks the Least-Action selection criterion, whereas the LAUT framework explicitly employs this variational principle to filter and identify the optimal effective model. Consequently, LAUT suppresses the relative error by nearly two orders of magnitude (reaching $\eta_g \sim 10^{-4}$), confirming the necessity of the least-action optimization in extracting physically faithful parameters.

\begin{figure}[t]
    \centering
    \subfigure{
        \begin{minipage}[b]{0.46\columnwidth}
            \centering
            \begin{overpic}[scale=0.30]{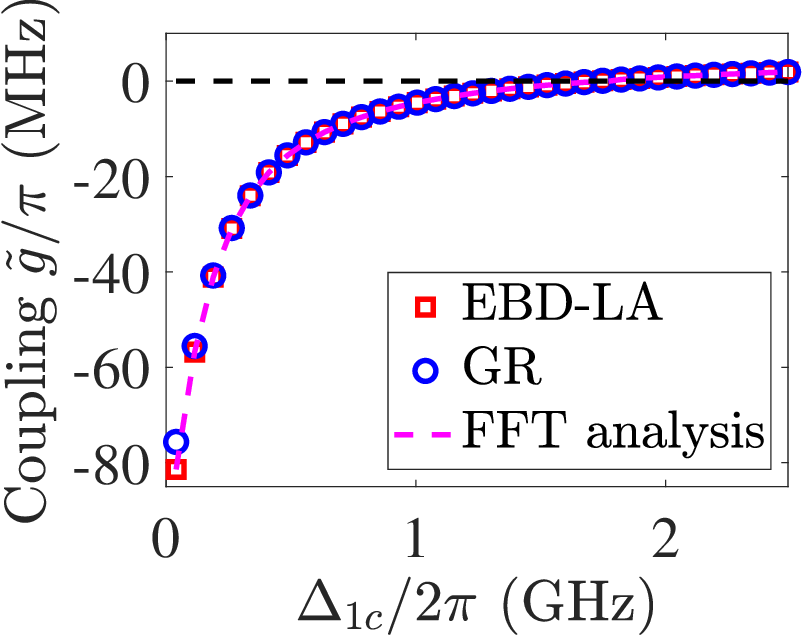}
                \put(0,80){\textbf{(a)}}
            \end{overpic}
        \end{minipage}
    }
    \subfigure{
        \begin{minipage}[b]{0.46\columnwidth}
            \centering
            \begin{overpic}[scale=0.30]{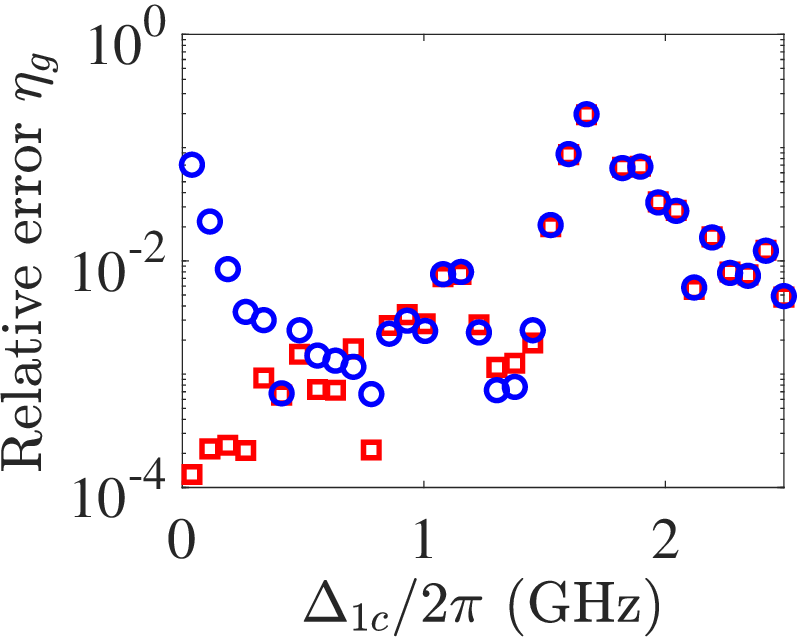}
                \put(0,80){\textbf{(b)}}
            \end{overpic}
        \end{minipage}
    }
    \caption{(Color online) Effective coupling comparison for \(\omega_1=\omega_2\) in the Qubit--Coupler--Qubit system.  
    \textbf{(a)}~Coupling \(2\tilde{g}\) obtained via LAUT, GR, and FFT extraction.  
    \textbf{(b)}~Relative error \(\eta_{g} = |(\tilde{g}-\tilde{g}_{\mathrm{FFT}})/\tilde{g}_{\mathrm{FFT}}|\).  
    LAUT tracks the FFT benchmark closely across the full range, while GR shows large resonance-enhanced errors.}
    \label{fig:compare_LAUT}
\end{figure}

\paragraph*{Full-parameter dynamical benchmarking.-}
To strictly evaluate the performance beyond symmetric configurations, we performed a two-dimensional sweep of the qubit frequencies $(\omega_1, \omega_2)$ with fixed $\omega_{\mathrm{c}}$. The dynamical fidelity was computed by averaging over a $2,000$~ns evolution window and across the initial computational basis states $\{\ket{100}, \ket{001}\}$. We quantify the global performance gap using the relative fidelity advantage, defined as $\eta_{_{\mathcal{F}}}=(\bar{\mathcal{F}}_{\text{LAUT}} - \bar{\mathcal{F}}_{\text{GR}})/\bar{\mathcal{F}}_{\text{LAUT}}$. In Fig.~\ref{fig:Compare_LAUT_GR}, regions of strong hybridization—explicitly identified by the resonance condition $g/\Delta \ge 1$—are masked in shadow. In this regime, the coupling strength becomes comparable to the energy gap, destroying the scale separation required for a valid effective model description~\cite{bravyi2011schrieffer} and leading to a breakdown in model fidelity ($\bar{\mathcal{F}} < 0.8$). Within the physically valid parameter space (colored area), the map is strictly positive ($\eta_{\mathcal{F}} > 0$), demonstrating that LAUT consistently outperforms GR. This confirms the global robustness of the LAUT framework in maintaining dynamical accuracy throughout the domain of physical utility.

\begin{figure}[t]
    \centering
    \includegraphics[width=0.8\columnwidth]{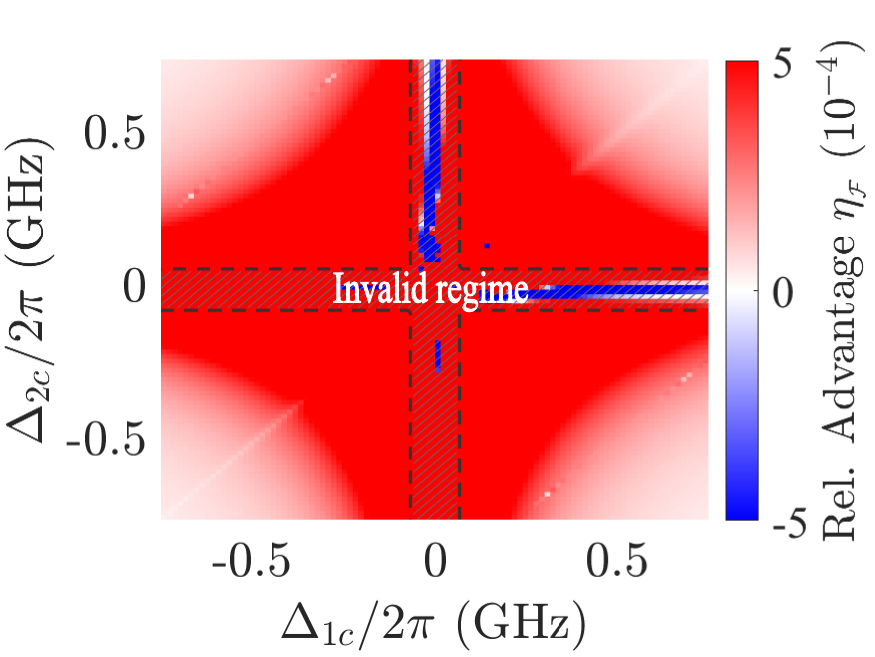}
    \caption{
    (Color online) Global performance comparison over the full two-dimensional parameter space ($\Delta_{1c}, \Delta_{2c}$) in the Qubit--Coupler--Qubit system. 
    The heatmap displays the relative fidelity advantage of the LAUT framework, defined as $\eta_{_{\mathcal{F}}}=(\bar{\mathcal{F}}_{\text{LAUT}} - \bar{\mathcal{F}}_{\text{GR}})/\bar{\mathcal{F}}_{\text{LAUT}}$. 
    The central cross-shaped region, labeled as ``Invalid regime'' (diagonal hatching enclosed by dashed lines), corresponds to areas of strong hybridization ($g/\Delta \gtrsim 1$) where the effective Hamiltonian description naturally breaks down ($\bar{\mathcal{F}} < 0.8$). 
    In the physically valid regime (the unhatched area), the advantage is {strictly positive} (red), demonstrating that the LAUT framework outperforms the GR method throughout the entire applicable parameter space.
}  
    \label{fig:Compare_LAUT_GR}
\end{figure}

Across symmetry tests, coupling extraction, and broad dynamical benchmarking, the LAUT framework uniquely integrates symmetry-preserving, quantitatively accurate, and dynamically robust methods among those evaluated. GR, in contrast, introduces intrinsic symmetry breaking and suffers from resonance-enhanced errors. These results establish LAUT as a principled and reliable foundation for the more demanding case studies examined in the following sections.

\subsection{Cross-Resonance Gate: Experiment-Level Validation}

The driven cross-resonance (CR) gate provides a stringent and experimentally grounded test of effective-Hamiltonian (EH) construction. In fixed-frequency superconducting architectures, a microwave drive on the control qubit induces an effective $ZX$ interaction on the target through their static capacitive coupling. This mechanism operates in a challenging quasi-dispersive regime where (i) the drive amplitude can be a sizable fraction of the qubit frequency (in the rotating frame), (ii) the qubit–qubit detuning is comparable to the anharmonicity, and (iii) static coupling is non-negligible. Under these conditions, standard perturbative methods such as SWT become quantitatively unreliable, while direct numerical simulation of the full multi-level system remains feasible as a reference. The CR gate, therefore, serves as an ideal testbed for evaluating both the nonperturbative LAUT-based block-diagonalization (EBD-LAUT) and its perturbative counterpart (PBD-BB).

\paragraph*{Model setup.-}
We consider two coupled transmons with an un-driven Hamiltonian
\begin{equation}
\begin{aligned}
H_0 = 
\sum_{i=1,2} 
\left( \omega_i\, a_i^{\dagger} a_i 
+ \frac{\beta_i}{2}\, a_i^{\dagger} a_i^{\dagger} a_i a_i \right)
+ J (a_1^{\dagger}+a_1)(a_2^{\dagger}+a_2),
\end{aligned}
\end{equation}
where $\omega_i$ and $\beta_i$ denote the transmon frequencies and anharmonicities, and $J$ is the static coupling.  
A resonant microwave drive applied to qubit~1 has the form
\[
H_{\mathrm{d}}=\frac{\Omega_{\rm CR}}{2}(a_1 + a_1^{\dagger}),
\]
and we work in a rotating frame generated by
\[
H_R=\omega_{\mathrm{d}} (a_1^{\dagger}a_1 + a_2^{\dagger}a_2),
\]
yielding the full Hamiltonian $H_{\rm full}=H_0 + H_{\mathrm{d}} - H_R$.

Counter-rotating terms first induce Bloch-Siegert-type shifts, renormalizing the bare frequencies and anharmonicities as
\begin{equation}
\begin{aligned}
\omega_i' &= \omega_i + \frac{J^2}{\Sigma_{12}} - \frac{2J^2}{\Sigma_{12}+\beta_i}, \\
\beta_i' &= \beta_i + \frac{4J^2}{\Sigma_{12}+\beta_i} - \frac{2J^2}{\Sigma_{12}} - \frac{3J^2}{\Sigma_{12}+2\beta_i},
\end{aligned}
\end{equation}
where $\Sigma_{12}=\omega_1+\omega_2$. In the remainder of this text, we adopt these renormalized parameters and drop the primes for notation simplicity.

The leading-order effective $ZX$ interaction arises from the interference of three distinct second-order virtual tunneling processes:
\begin{equation}
\begin{aligned}
    \ket{00} &\xleftrightarrow[]{\Omega_{\text{CR}}} \ket{10} \xleftrightarrow[]{J} \ket{01},\\
    \ket{10} &\xleftrightarrow[]{\Omega_{\text{CR}}} \ket{20} \xleftrightarrow[]{J} \ket{11},\\  
    \ket{10} &\xleftrightarrow[]{J} \ket{01} \xleftrightarrow[]{\Omega_{\text{CR}}} \ket{11}.
\end{aligned}
\end{equation}
Summing the contributions from these virtual pathways yields the second-order effective coupling strength
\begin{equation}
    \tilde{\nu}_{ZX}^{(2)} = 
    \frac{\Omega_{\rm CR} J}{2}
    \left(
    -\frac{1}{\Delta_{1\mathrm{d}}}
    -\frac{1}{\Delta_{12}}
    +\frac{1}{\Delta_{1\mathrm{d}}+\beta_1}
    +\frac{1}{\Delta_{12}+\beta_1}
    \right),
\end{equation}
where $\Delta_{1\mathrm{d}}=\omega_1-\omega_{\mathrm{d}}$. A detailed derivation via the PBD-BB framework, including a summary of fourth-order contributions, is provided in Appendix~\ref{app:CR_coupling_strength}.

\begin{figure*}[t]
    \centering
    \includegraphics[width=1.9\columnwidth]{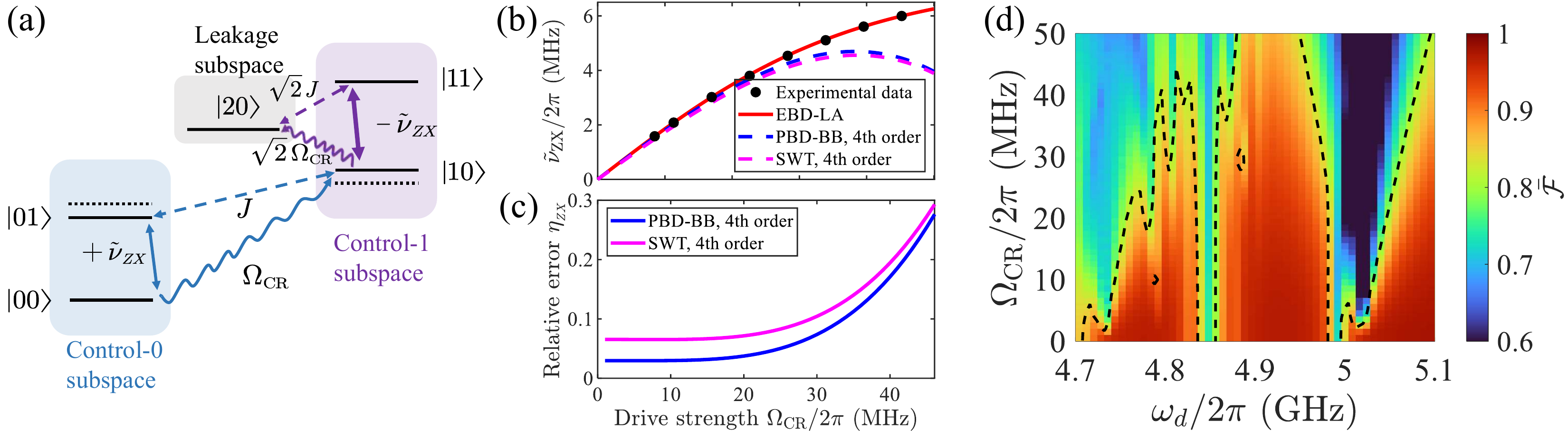}
    \caption{(Color online) Microscopic mechanism and experimental validation of the high-precision effective Hamiltonian.
\textbf{(a)} State-dependent virtual pathways unfolded into Control-0 (blue) and Control-1 (purple) subspaces. The effective $ZX$ interaction arises from the pathway asymmetry, specifically the dominant virtual transition via the leakage state $|20\rangle$ in the Control-1 branch. Solid and dotted lines represent renormalized and bare energy levels, respectively.
\textbf{(b)} Comparison of the effective interaction strength $\tilde{\nu}_{ZX}$ extracted from EBD-LA (red solid line) against experimental data from Ref.~\cite{wei2022hamiltonian} (black dots). The geometric approach maintains excellent agreement across the entire drive range, whereas conventional perturbative expansions (SWT, PBD-BB; dashed lines) diverge significantly beyond $\Omega_{\text{CR}}/2\pi \approx 20$ MHz.
\textbf{(c)} Relative deviation $\eta_{_{ZX}} = |(\tilde{\nu}_{ZX} - \tilde{\nu}_{ZX}^{\rm LAUT})/\tilde{\nu}_{ZX}^{\rm LAUT}|$ reveals that PBD-BB consistently maintains higher accuracy than SWT before the perturbative breakdown in the strong-driving regime.
\textbf{(d)} Time-averaged fidelity heatmap in the $(\omega_{\mathrm{d}},\Omega_{\rm CR})$ plane. The soundness contour $I=0.9$ (black dashed) accurately demarcates the high-fidelity regime, validating the effective model's applicability.}
    \label{fig:CRgate}
\end{figure*}
While the analytical expression includes all three channels, {Fig.~\ref{fig:CRgate} (a)} visually elucidates the two dominant pathways that govern the gate's conditional logic: the standard computational path (blue line) and the anharmonic excursion via the $\ket{20}$ state (purple line). The latter, highlighted in the figure with the $\sqrt{2}$ enhancement, is crucial for the correct sign and magnitude of the interaction and is rigorously captured by our LAUT framework. Higher-order corrections become significant when $\varepsilon=\Omega_{\rm CR}/\Delta_{12} \gtrsim 0.22$.

\paragraph*{Experiment-Level Benchmarking.-}
Figure~\ref{fig:CRgate} (b) compares the predicted $ZX$ interaction against the experimentally extracted values reported in Ref.~\cite{wei2022hamiltonian}.  
The nonperturbative EBD-LAUT method (red solid curve) quantitatively reproduces the experimental data across the full drive range, with average absolute deviations of only $\sim 11$~kHz.  
In contrast, fourth-order SWT (pink dashed) exhibits significant deviations, and PBD-BB (blue dashed) remains accurate only in the moderate-drive region before drifting at large $\Omega_{\rm CR}$.  
The relative error $\eta_{_{ZX}}= |(\tilde{\nu}_{ZX}-\tilde{\nu}_{ZX}^{LAUT})/\tilde{\nu}_{ZX}^{LAUT}|$
in Fig.~\ref{fig:CRgate} (c) confirms that PBD-BB maintains noticeably higher accuracy than SWT; yet, both ultimately break down outside the perturbative domain.


To evaluate dynamical performance, we compute the time-averaged fidelity between the EBD-LAUT effective Hamiltonian and the full evolution over a 1,000-ns evolution window, sweeping the two-dimensional parameter space $(\omega_{\mathrm{d}},\Omega_{\rm CR})$. As shown in Fig.~\ref{fig:CRgate} (d), high fidelity is confined precisely to the region predicted by the soundness measure $I(H_{\rm eff})=0.9$ (black dashed line), while it degrades sharply as $I(H_{\rm eff})$ decreases. This strong correspondence confirms that $I(H_{\rm eff})$ provides a reliable and experimentally meaningful diagnostic for CR-gate operating regimes.

The CR-gate analysis confirms the hierarchical performance observed in the Q--C--Q benchmarks:  
EBD-LAUT provides quantitatively accurate predictions across the entire experimentally relevant regime; PBD-BB reproduces EBD-LAUT in the perturbative domain and outperforms SWT; furthermore, SWT becomes unreliable even under moderate drive strengths. These results demonstrate that the LAUT framework delivers both the accuracy and robustness needed for modeling driven multi-qubit interactions and for guiding high-fidelity CR-gate calibration.

\subsection{Beyond the Rotating-Wave Approximation: Effective Hamiltonian Construction}
\label{subsec:beyond_RWA}

As superconducting circuits push toward stronger coupling, higher coherence, and increasingly sophisticated multi-qubit interactions, operating regimes in which the rotating-wave approximation (RWA) becomes unreliable are gaining practical relevance.  
In parameter ranges characterized by moderate detuning or normalized coupling ratios $g/\Sigma \gtrsim 10^{-2}$, counter-rotating-wave (CRW) terms generate appreciable virtual transitions that renormalize both qubit frequencies and mediated interactions.  
Because these renormalizations directly affect gate calibration, interaction engineering, and hardware-faithful modeling, a systematic framework that remains accurate beyond the RWA is essential.

\paragraph*{Model setup and effective Hamiltonian.-}
We consider a transmon-based qubit–coupler–qubit (Q–C–Q) architecture~\cite{yan2018tunable} described by
\begin{equation}
H_{\mathrm{full}} = H_{\mathrm{sys}} + V_{\mathrm{RW}} + V_{\mathrm{CRW}},
\end{equation}
where
\begin{equation}
H_{\mathrm{sys}}
= \sum_{i=1,2,c}
\left(
\omega_i a_i^\dagger a_i
+ \frac{\beta_i}{2} a_i^\dagger a_i^\dagger a_i a_i
\right)
\end{equation}
specifies the bare oscillators with anharmonicity.  
The interaction Hamiltonian naturally separates into excitation-conserving and non-conserving components:
\begin{align}
V_{\mathrm{RW}} &= \sum_{j=1,2} g_j(a_j^\dagger a_c + a_j a_c^\dagger)
+ g_{12}(a_1^\dagger a_2 + a_1 a_2^\dagger), \\
V_{\mathrm{CRW}} &= - \sum_{j=1,2} g_j(a_j a_c + a_j^\dagger a_c^\dagger)
- g_{12}(a_1 a_2 + a_1^\dagger a_2^\dagger).
\end{align}
Although $V_{\mathrm{CRW}}$ is absent under RWA, the coupler enables off-resonant two-photon pathways that remain energetically accessible and accumulate through second-order processes.

To benchmark the RWA on equal footing, we construct an effective Hamiltonian \emph{without eliminating the coupler}, preserving the operator structure of the RWA but with CRW-renormalized parameters.  
Using second-order perturbative block diagonalization (PBD-BB), we obtain
\begin{equation}
H_{\mathrm{eff}} = H_{\mathrm{sys}}^{\mathrm{eff}} + V_{\mathrm{RW}}^{\mathrm{eff}},
\end{equation}
where
\begin{align}
H_{\mathrm{sys}}^{\mathrm{eff}} &= \sum_{i=1,2,c}
\left(\tilde{\omega}_i a_i^\dagger a_i + \frac{\beta_i}{2} a_i^\dagger a_i^\dagger a_i a_i\right), \\
V_{\mathrm{RW}}^{\mathrm{eff}} &= \sum_{j=1,2} \tilde{g}_j (a_j^\dagger a_c + a_j a_c^\dagger)
+ \tilde{g}_{12}(a_1^\dagger a_2 + a_1 a_2^\dagger).
\end{align}
The renormalized parameters (to leading order in $g/\Sigma$) are:
\begin{align*}
\tilde{\omega}_{\mathrm{c}} &= \omega_{\mathrm{c}} - \frac{2g_1^2}{\Sigma_1 + \beta_c}
- \frac{2g_2^2}{\Sigma_2 + \beta_c}
- \frac{g_{12}^2}{\Sigma_{12}}, \\
\tilde{\omega}_1 &= \omega_1 - \frac{2g_1^2}{\Sigma_1 + \beta_1}
- \frac{g_2^2}{\Sigma_2}
- \frac{2g_{12}^2}{\Sigma_{12} + \beta_1}, \\
\tilde{\omega}_2 &= \omega_2 - \frac{g_1^2}{\Sigma_1}
- \frac{2g_2^2}{\Sigma_2 + \beta_2}
- \frac{2g_{12}^2}{\Sigma_{12} + \beta_2}, \\
\tilde{g}_1 &= g_1 - \frac{g_2 g_{12}}{\Sigma_2}, \qquad
\tilde{g}_2 = g_2 - \frac{g_1 g_{12}}{\Sigma_1}, \\
\tilde{g}_{12} &= g_{12} - \frac{g_1 g_2}{2}\left(\frac{1}{\Sigma_1} + \frac{1}{\Sigma_2}\right),
\end{align*}
with $\Sigma_j = \omega_j + \omega_{\mathrm{c}}$ and $\Sigma_{12} = \omega_1 + \omega_2$.  
These terms arise from virtual CRW-enabled processes such as
\begin{equation}
\ket{100}
\xleftrightarrow[]{g_2}
\ket{111}
\xleftrightarrow[]{g_1}
\ket{001},
\label{eq:CRW_induced_process}
\end{equation}
which contribute to the effective qubit–qubit coupling even though the intermediate state is far detuned.

\paragraph*{Quantitative comparison with RWA.-}
Using representative device parameters
($\omega_{1,2}/2\pi=4$\,GHz,
$\omega_{\mathrm{c}}/2\pi=5$\,GHz,
$\beta_i/2\pi=-0.2$\,GHz,
$g_{1,2}/2\pi=100$\,MHz,
$g_{12}/2\pi=7$\,MHz), we find:
\begin{enumerate}[(1)]
\item CRW renormalizes the direct coupling to  
\[
\tilde{g}_{12}/2\pi \approx 5.9~\mathrm{MHz},
\]
compared with the RWA value of $7$\,MHz;
\item Eliminating the coupler yields a mediated interaction  
\[
g_{\mathrm{decoup}}^{(2)}/2\pi \approx -4.1~\mathrm{MHz},
\]
whereas the RWA predicts $-3.0$\,MHz.
\end{enumerate}
These $g^2/\Sigma$-level corrections are MHz-scale and thus experimentally resolvable.

The RWA relies on the assumption that terms oscillating at frequencies $\sim(\omega_j+\omega_{\mathrm{c}})$ average out and thus contribute negligibly to the dynamics.  
In practice, however, the pathway illustrated in Eq.~\eqref{eq:CRW_induced_process} shows that CRW terms mediate off-resonant virtual transitions with amplitudes scaling as $g^2/(\omega_j+\omega_{\mathrm{c}})$. Although these processes do not populate intermediate states, they generate sizable re-normalizations of both qubit frequencies and mediated couplings.  
Consequently, RWA-based models systematically underestimate interaction strengths, mis-predict dressed-state energy shifts, and fail to capture long-time interference effects—limitations that become particularly pronounced in tunable-coupler architectures. These discrepancies are directly visible in the population and infidelity dynamics shown in Figs.~\ref{subfig:Compare_RWA_Pop} and~\ref{subfig:Compare_RWA_Infid}.

\begin{figure}[t]
    \centering
    \subfigure{
        \begin{minipage}[b]{0.46\columnwidth}
        \centering
        \begin{overpic}[scale=0.3]{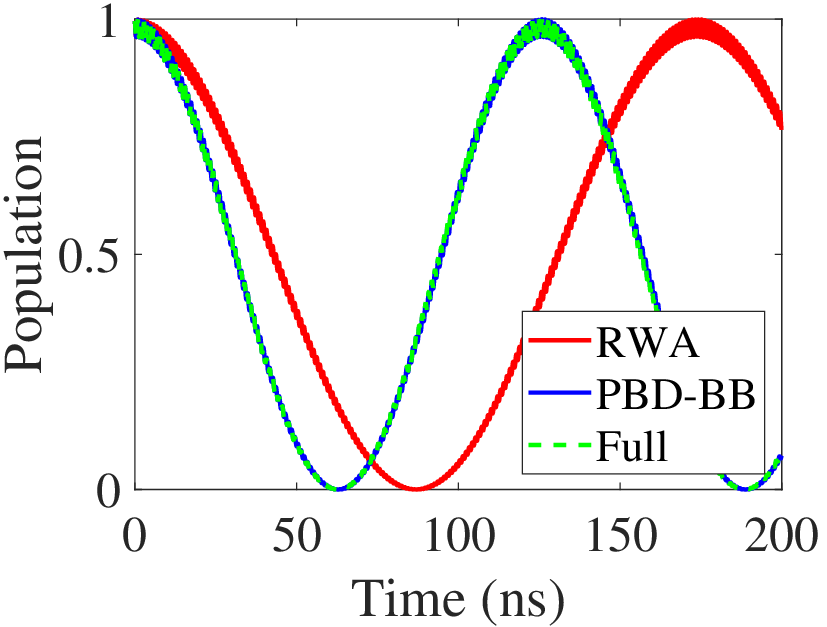}
            \put(0,80){\textbf{(a)}}
        \end{overpic}
        \end{minipage}\label{subfig:Compare_RWA_Pop}
    }
    \subfigure{
        \begin{minipage}[b]{0.46\columnwidth}
        \centering
        \begin{overpic}[scale=0.29]{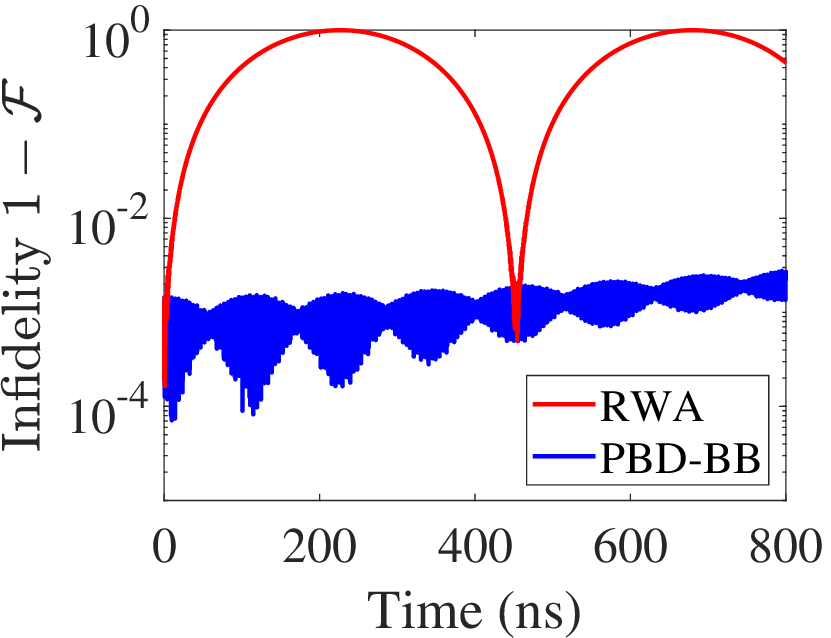}
            \put(0,80){\textbf{(b)}}
        \end{overpic}
        \end{minipage}\label{subfig:Compare_RWA_Infid}
    }
    \subfigure{
        \begin{minipage}[b]{0.6\columnwidth}
        \centering
        \begin{overpic}[scale=0.29]{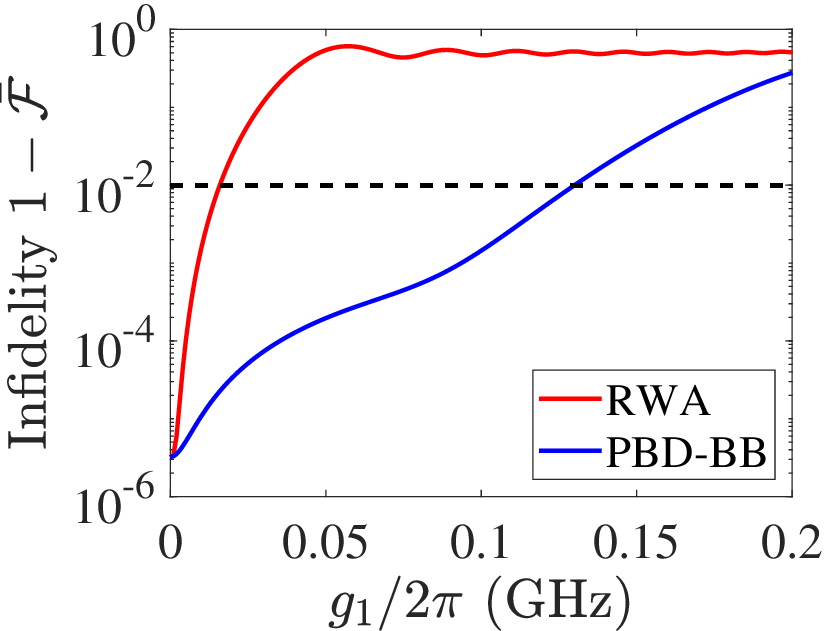}
            \put(0,78){\textbf{(c)}}
        \end{overpic}
        \end{minipage}\label{subfig:Compare_RWA_avg}
    }
    \caption{
    (Color online) Benchmarking the rotating-wave approximation (RWA) against the PBD-BB effective Hamiltonian (blue) in the Q--C--Q model, using full-Hamiltonian evolution as reference.
    \textbf{(a)} Population dynamics of qubit~1 initialized in $\ket{100}$.  
    RWA exhibits a clear Rabi-frequency mismatch, while the PBD-BB EH accurately tracks the exact evolution.
    \textbf{(b)} Instantaneous infidelity relative to the full dynamics.  
    RWA errors accumulate rapidly, whereas PBD-BB maintains low infidelity throughout the evolution.
    \textbf{(c)} Time-averaged infidelity versus coupling strength $g_1=g_2$.  
    PBD-BB remains below the $1\%$ threshold, while RWA breaks down in the moderate-coupling regime.
    }
    \label{fig:compare_RWA}
\end{figure}

For quantitative validation, we evolve for $10,000$ ns
\[
|\psi_\alpha(t)\rangle = e^{-iH_\alpha t} |100\rangle,
\qquad
\alpha \in \{f, r, e\},
\]
where $(f,r,e)$ denotes the full, RWA, and PBD-BB effective Hamiltonian, respectively,
and computes the infidelity $1-\mathcal{F}(t) = 1 - |\langle\psi_f(t) | \psi_\alpha(t)\rangle|^2$.
Figure~\ref{subfig:Compare_RWA_avg} confirms that the effective Hamiltonian faithfully reproduces the full dynamics over extended timescales, while RWA incurs substantial instantaneous and accumulated errors.

Conventional CRW-aware approaches such as GRWA or CHRW~\cite{irish2007dynamics, casanova2010deep} apply to few-mode models and rely on manually chosen unitary transformations. In contrast, the block-diagonalization framework provides a general, symmetry-preserving, and quantitatively accurate method for constructing effective Hamiltonians in multi-mode superconducting circuits. Importantly, previous Q–C–Q analyses almost exclusively derive qubit–qubit interactions \emph{after} eliminating the coupler. Here we extract CRW-renormalized interactions \emph{directly from the full three-mode Hamiltonian}, yielding a more complete and experimentally faithful description of dressed-state hybridization. Incorporating CRW corrections at this level is essential for predicting qubit shifts, calibrating mediated interactions, and designing pulses in tunable-coupler devices. The PBD-BB framework therefore provides an analytically tractable and hardware-relevant approach to modeling beyond-RWA physics in next-generation superconducting platforms.

\subsection{Mediated Three-Body Interaction in a Multilevel Q--C--Q--C--Q Architecture}

A distinctive advantage of the LAUT-guided PBD-BB framework is its ability to reveal higher-order mediated processes that remain inaccessible to conventional effective-Hamiltonian techniques. As an illustrative case study, we analyze a multilevel Q--C--Q--C--Q chain constructed from capacitively shunted flux qubits (CSFQs) and tunable couplers. Although the main text has focused on the resulting effective interaction, its microscopic origin relies on the multilayered connectivity of this five-mode architecture. For completeness, Appendix~\ref{app:device_details} summarizes the device model, bare parameters, and excitation-number structure that underlie the analysis below.

\paragraph*{Model setup.-}We construct the effective Hamiltonian in two stages: first mapping the microscopic circuit to a three-site extended Bose–Hubbard model (EBHM), and subsequently projecting onto the computational subspace. The intermediate EBHM reads
\begin{equation}
\label{eq:EBHM_3_body_full}
\begin{aligned}
    H &= H_{\mathrm{site}} + H_{\mathrm{coupling}},\\
    H_{\mathrm{site}} &= 
    \sum_{i=1}^{3}\!
    \left( \mu_i a_i^\dagger a_i 
    + \tfrac{U_i}{2} a_i^\dagger a_i^\dagger a_i a_i \right),\\
    H_{\mathrm{coupling}} &= 
    J_0 \!\sum_{\langle i,j\rangle}\!(a_i^\dagger a_j+\mathrm{H.c.})
    + V \!\sum_{\langle i,j\rangle}\! n_i n_j ,
\end{aligned}
\end{equation}
where $a_i$ ($a_i^\dagger$) denotes the bosonic annihilation (creation) operator and $n_i=a_i^\dagger a_i$. The detailed mapping from circuit parameters to the EBHM is provided in Appendix~\ref{app:mapping_to_EBHM}. Truncating each site to its three lowest eigenstates $\{0,1,2\}$ defines the relevant Hilbert space
\begin{equation*}
    \mathcal{H}
    = \mathrm{span}\{\ket{ijk}\,|\, i,j,k\in\{0,1,2\}\},
\end{equation*}
while the computational subspace is restricted to
\begin{equation*}
    \mathcal{H}_{\mathrm{C}}
    = \mathrm{span}\{\ket{ijk}\,|\, i,j,k\in\{0,1\}\}.
\end{equation*}
Virtual excitations of the central site into its second excited state $|2\rangle$ mediate a conditional next-nearest-neighbor hopping interaction of the form$$(a_1^\dagger a_3 + a_3^\dagger a_1)n_2 ,$$a term absent in the bare EBHM. Decoupling $\mathcal{H}_{\mathrm{C}}$ from the leakage manifold via the LAUT-based transformation yields the effective Hamiltonian
\begin{equation}
\label{eq:EH_3_body}
\begin{aligned}
    H_{\mathrm{eff}}^{\mathrm{3Q}} &=
    \sum_{i=1}^{3} \tilde{\omega}_i a_i^\dagger a_i 
    + J\!\sum_{\langle i,j\rangle}(a_i^\dagger a_j+\mathrm{H.c.}) \\
    &\quad + K (a_1^\dagger a_3 + a_3^\dagger a_1)n_2 ,
\end{aligned}
\end{equation}
where $K$ quantifies the strength of the mediated three-body interaction.Finally, mapping to qubit operators via $a_i^\dagger\rightarrow\sigma_i^+$ and $n_i\rightarrow(1-Z_i)/2$ leads to
\begin{equation}
    \begin{aligned}
        H_{\mathrm{eff}}^{\mathrm{3Q}} = & \sum_{i=1}^{3} \frac{-\tilde{\omega}_i}{2} Z_i 
        + \frac{J}{2} \sum_{\langle i,j \rangle} (X_i X_j + Y_i Y_j) \\
        & + K \cdot \underbrace{\left[ -\frac{1}{4} \left( X_1 Z_2 X_3 + Y_1 Z_2 Y_3 \right) \right]}_{\hat{\mathcal{O}}_{3\mathrm{b}}},
    \end{aligned}
\end{equation}
which reveals a genuine three-body coupling between the outer qubits conditioned on the state of the central one. To quantify this interaction, we identify the normalized three-body operator $\hat{\mathcal{O}}_{3\mathrm{b}}=-\frac{1}{4}(X_1 Z_2 X_3 + Y_1 Z_2 Y_3)$ and extract the effective coupling strength $\kappa$ via the contrast of conditional hoppings:
\begin{equation}
    \kappa = \bra{110} H_{\mathrm{eff}}^{\mathrm{3Q}} \ket{011} - \bra{100} H_{\mathrm{eff}}^{\mathrm{3Q}} \ket{001}.
\end{equation}
This metric $\kappa$ directly captures the amplitude of the excitation transfer mediated by the central site. Note that under the projection onto $\hat{\mathcal{O}}_{3\mathrm{b}}$, this is consistent with $\kappa = \mathrm{Tr}[\hat{\mathcal{O}}_{3\mathrm{b}}^\dagger H_{\mathrm{eff}}^{\mathrm{3Q}}]$.

\paragraph*{Microscopic origin of the effective $XZX+YZY$ interaction.-}

The origin of the three-body $XZX+YZY$ term becomes transparent in the two-excitation 
manifold (Fig.~\ref{fig:Energy_levels_XZX}).  
Although the boundary states $\ket{110}$ and $\ket{011}$ do not couple directly, each connects 
to the noncomputational state $\ket{020}$ through the nearest-neighbor hopping $J_0$, forming 
the virtual pathway
\[
\ket{110} \;\longleftrightarrow\; \ket{020} \;\longleftrightarrow\; \ket{011}.
\]
The central site thus temporarily hosts two excitations, enabling a correlated boundary-to-boundary 
exchange conditioned on the occupation of site~2.

Block diagonalization makes this mechanism explicit: eliminating the high-energy $\ket{020}$ 
state yields a single effective matrix element between the dressed states 
$\overline{\ket{110}}$ and $\overline{\ket{011}}$ with amplitude $K$.  
In operator form this produces the three-body term
\[
H_{\mathrm{eff}} \supset K\left(X_1 Z_2 X_3 + Y_1 Z_2 Y_3\right),
\]
in which excitation transfer between sites~1 and~3 is activated only when the middle site is 
occupied—precisely the hallmark of a genuine mediated three-body interaction.

We derive the three-body coupling using PBD-BB, which integrates out leakage states while preserving the excitation-number structure. This method organizes virtual pathways order-by-order, explicitly isolating the $\ket{020}$ intermediate state and yielding closed-form expressions for $J$ and $K$. This contrasts with the Schrieffer–Wolff expansion, where nested commutators make higher-order derivations in multilevel systems algebraically intractable and prone to errors. PBD-BB resolves these issues by displaying contributions directly within the block structure. In the EBHM representation, the dominant transition via $\ket{020}$ is captured at second order; however, for the full circuit Hamiltonian, this corresponds to a fourth-order contribution, the derivation of which is provided in the Appendix~\ref{app:derivation of three-body interaction}.

\begin{figure}[t]
    \centering
    \subfigure{
    \begin{minipage}[b]{0.8\columnwidth}\label{fig:Energy_levels_XZX}
        \centering
        \begin{overpic}[scale=0.3]{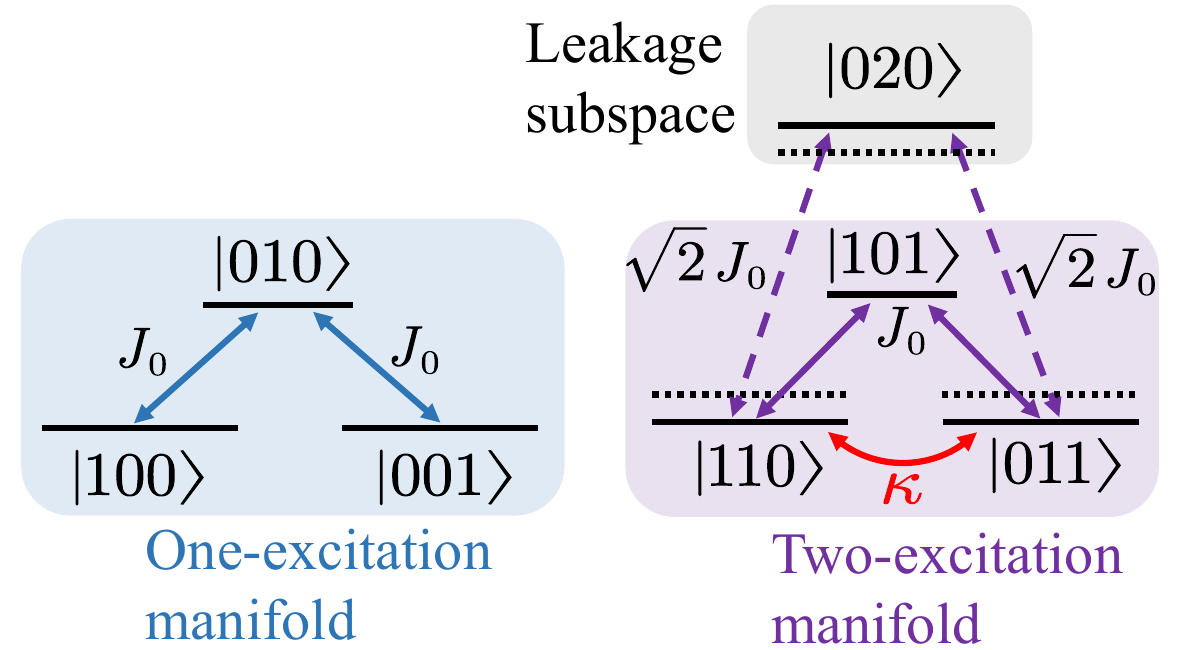}
            \put(-15,55){\textbf{(a)}}
        \end{overpic}
    \end{minipage}}
    \subfigure{
    \begin{minipage}[b]{0.45\columnwidth}\label{fig:Compare_XZX}
        \centering
        \begin{overpic}[scale=0.3]{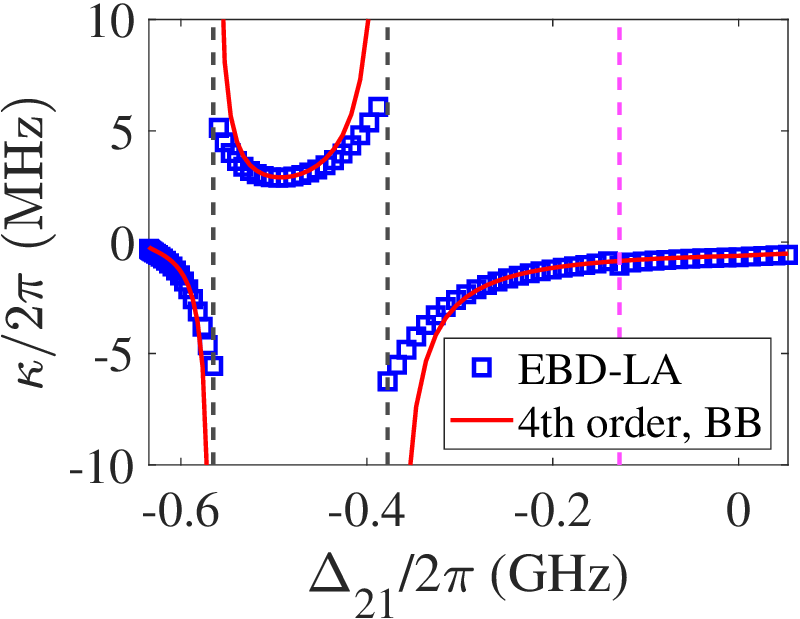}
            \put(0,80){\textbf{(b)}}
        \end{overpic}
    \end{minipage}}
    \subfigure{
    \begin{minipage}[b]{0.44\columnwidth}\label{fig:Compare_XZX_error}
        \centering
        \begin{overpic}[scale=0.3]{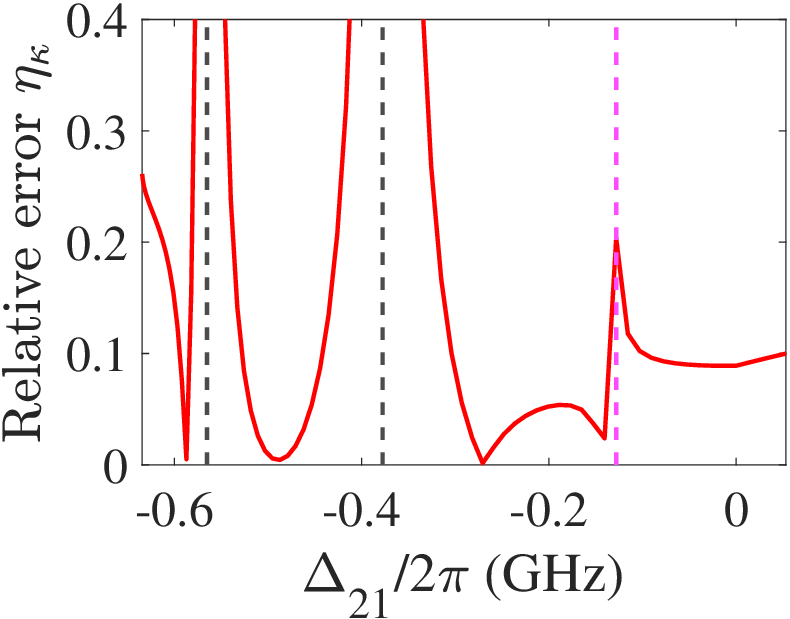}
            \put(0,80){\textbf{(c)}}
        \end{overpic}
    \end{minipage}}
    \subfigure{
    \begin{minipage}[b]{0.44\columnwidth}\label{fig:Compare_XZX_measure}
        \centering
        \begin{overpic}[scale=0.3]{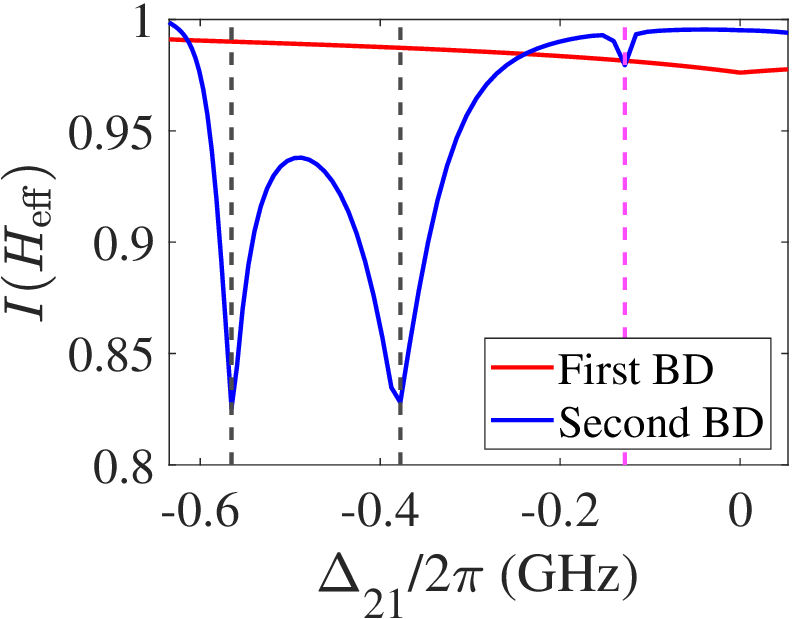}
            \put(0,80){\textbf{(d)}}
        \end{overpic}
    \end{minipage}}
    \subfigure{
    \begin{minipage}[b]{0.44\columnwidth}\label{fig:Compare_XZX_dynamics}
        \centering
        \begin{overpic}[scale=0.3]{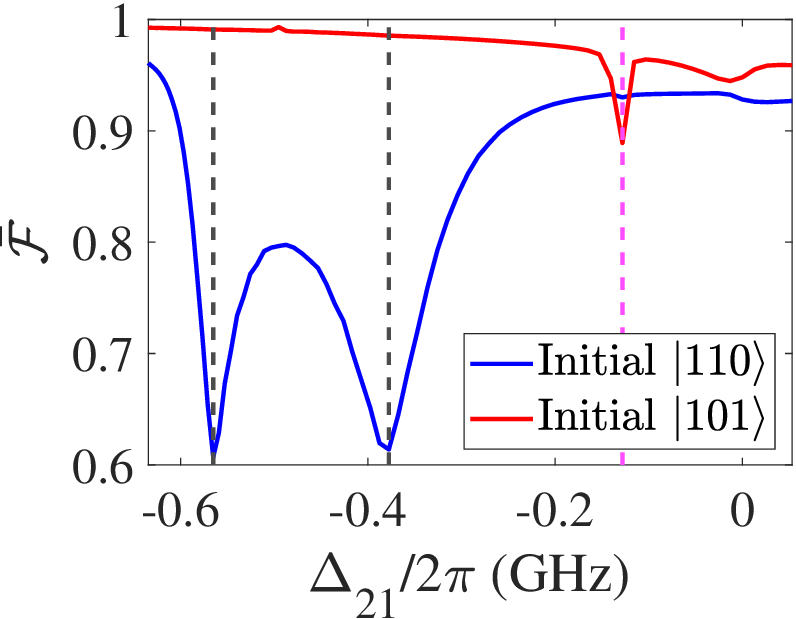}
            \put(0,80){\textbf{(e)}}
        \end{overpic}
    \end{minipage}}
    \caption{
Microscopic origin and benchmarking of the conditional three-body interaction.
\textbf{(a)} Mechanism of the effective coupling. In the one-excitation manifold (blue), dynamics proceed via real nearest-neighbor hopping $J_0$ through the computational state $|010\rangle$. In contrast, the two-excitation manifold (purple) relies on a virtual transition through the high-energy leakage state $|020\rangle$ (gray). Adiabatic elimination of this leakage path generates the effective conditional coupling $\kappa$ (red arrow) connecting $|110\rangle$ and $|011\rangle$. Solid (dotted) lines denote renormalized (bare) energy levels.
\textbf{(b)} Effective coupling strength $\kappa$ extracted via EBD-LA (blue squares) compared to 4th-order perturbation theory (red line). 
\textbf{(c)} Relative error $\eta_{\kappa}$ of the effective model.
\textbf{(d)} Soundness measure $I(H_{\rm eff})$. The sharp dips indicate the breakdown of the effective description at resonances.
\textbf{(e)} Time-averaged fidelity $\bar{\mathcal{F}}$ for different initial states, showing consistent high-fidelity performance in the valid regimes predicted by (d). In panels (b)–(e), vertical dashed lines indicate the dominant fourth-order (black) and weaker eighth-order (pink) resonances.}

    \label{fig:XZX_1D}
\end{figure}

\paragraph*{Numerical validation.-}
We benchmark the tunable three-body interaction $\kappa$ by comparing the EBD-LAUT results against the fourth-order PBD-BB predictions. The effective coupling, arising from higher-order virtual transitions via the central site's second excited state ($\ket{020}$), is controlled by the flux-tunable parameters $\omega_2$ and $\beta_2$.
Figure~\ref{fig:Compare_XZX_error} demonstrates excellent agreement between EBD-LAUT and PBD-BB in the dispersive regime. However, near specific flux points (indicated by vertical lines), PBD-BB diverges due to resonances involving $\ket{110}\leftrightarrow\ket{020}$ and $\ket{011}\leftrightarrow\ket{020}$ transitions, whereas EBD-LAUT remains well-behaved.
This accuracy is quantitatively confirmed in Fig.~\ref{fig:Compare_XZX}, where deviations remain below $0.1\,\mathrm{MHz}$ away from resonances.
Crucially, the static diagnostic $I(H_{\mathrm{eff}})$ in Fig.~\ref{fig:Compare_XZX_measure} provides physical insight into the model's validity limits: while the mapping to the EBHM (red) remains robust, the subsequent isolation of the qubit subspace (blue) exhibits sharp dips near resonances, correctly signaling the breakdown of the qubit approximation due to leakage.
Finally, Fig.~\ref{fig:Compare_XZX_dynamics} validates these findings in the time domain: the dynamical fidelity for the $\ket{101}$ state remains high over $50\,\mu$s, while the leakage-prone $\ket{110}$ state shows degradation near resonance points, consistent with our static metrics.
Despite these resonance constraints, we achieve a substantial three-body coupling strength $\kappa/2\pi \approx 3\,\mathrm{MHz}$, significantly exceeding typical decoherence rates.

\begin{table*}
\centering
\caption{Comparative overview of effective Hamiltonian construction methods. The LAUT-based frameworks (EBD-LAUT and PBD-BB) are unique in simultaneously satisfying symmetry preservation and the least-action (LAUT) principle, ensuring robustness in both perturbative and non-perturbative regimes.}
\label{tab:EH_method_comparison}
\begin{tabular}{l|c|c|c|c}
\hline\hline
\textbf{Method} & \textbf{Sym. Preserving} & \textbf{Structure} & \textbf{Regime} & \textbf{LAUT Consistency} \\
\hline
\textbf{EBD-LAUT} & $\checkmark$ & Numerical & Nonpert. & $\checkmark$ \\
\textbf{PBD-BB} & $\checkmark$ & Series & Pert. & Up to $4^{\text{th}}$ order~\cite{richert1976comparison} \\
\textbf{SWT}    & $\checkmark$ & Series & Pert. & Up to $2^{\text{nd}}$ order~\cite{mankodi2024perturbative} \\
\textbf{GR}     & $\times$     & Numerical & Nonpert. & $\times$ \\
\hline\hline
\end{tabular}
\end{table*}

\subsection{Method Summary and Practical Guidance}
\label{subsec:EH_method_summary}

Our benchmarking analysis establishes a rigorous hierarchy among effective Hamiltonian construction strategies, governed by three critical criteria: symmetry preservation, variational optimality, and perturbative consistency. Table~\ref{tab:EH_method_comparison} summarizes the structural characteristics and validity regimes of the evaluated methods.

Below are the \textit{key insights} from benchmarking:
\begin{itemize}
    \item \textbf{Symmetry and Variational Grounding:} 
    Symmetry preservation is a prerequisite for physical consistency. The GR method suffers from intrinsic symmetry-breaking errors. Conversely, while SWT respects symmetry, it lacks a variational extremal principle, leading to deviations from the optimal LAUT trajectory at higher orders (e.g., mispredicting $ZZ$ shifts~\cite{li2022nonperturbative}). Only the LAUT-based frameworks---EBD-LAUT and its perturbative counterpart, PBD-BB---simultaneously satisfy both constraints.

    \item \textbf{Experimental and Dynamical Fidelity:} 
    Across Q-C-Q architectures and driven CR gates, EBD-LAUT consistently yields the highest fidelity. Notably, PBD-BB accurately reproduces EBD-LAUT in the weak-coupling regime and quantitatively matches experimental $ZX$ interaction rates where standard SWT fails. This confirms that theoretical adherence to the Least-Action criterion directly translates to observable experimental accuracy.

    \item \textbf{Beyond-RWA Capabilities:} 
    In regimes with non-negligible counter-rotating terms, PBD-BB successfully captures MHz-scale renormalizations inherently missed by standard RWA approaches. This establishes LAUT-consistent perturbative methods as essential tools for the high-precision modeling of modern tunable couplers.

    \item \textbf{Discovery of Multi-Body Physics:} 
    Beyond standard two-qubit gates, the framework serves as a powerful discovery tool for higher-order physics. As demonstrated in the multilevel CSFQ architecture, the LAUT-guided PBD-BB approach systematically reveals the microscopic origins of complex mediated couplings. It successfully identified and quantified a genuine three-body interaction ($XZX+YZY$), proving that significant multibody terms can be engineered in standard devices without specialized hardware.
\end{itemize}

\paragraph*{Practical guidance:}
\begin{enumerate}
    \item \textbf{Method Selection:} 
    For strongly hybridized or symmetry-sensitive regimes, \textbf{EBD-LAUT} is the definitive choice. For applications requiring analytical transparency in the weak-coupling limit, \textbf{PBD-BB} offers an optimal balance of accuracy and complexity. Methods lacking variational grounding (such as GR or high-order SWT) may incur higher systematic uncertainties for quantitative calibration.
    
    \item \textbf{Computational Scalability:}
A distinct advantage of the EBD-LAUT framework is that it circumvents full diagonalization. Since the optimal transformation $T$ (Eq.~\eqref{eq:BDT_from_LAUT}) is constructed exclusively from the target subspace eigenstates, the method is highly amenable to sparse-matrix techniques. We specifically leverage the Dual Applications of Chebyshev Polynomials (DACP) algorithm~\cite{guan2021dual}, which acts as a spectral band-pass filter to extract the $d$ relevant eigenstates. This reduces the computational complexity to $\mathcal{O}(N_{nz} \cdot (m+d))$, where $N_{nz}$ is the sparsity count and $m$ is the filter order. By avoiding the prohibitive $\mathcal{O}(D^3)$ scaling of dense solvers, this integration theoretically extends EBD-LAUT to verify specific interactions in systems of $\sim 20$ spins ($D \sim 10^6$) on standard computing resources. This scalability elevates the framework beyond isolated gate calibration, enabling critical system-level tasks such as crosstalk characterization and frequency collision analysis in realistic NISQ processor layouts.
\end{enumerate}

\section{Conclusion and Outlook}
\label{sec:conclusion}

We have introduced the Least Action Unitary Transformation (LAUT) as a rigorous variational foundation for effective Hamiltonian theory. By selecting the unique block‑diagonalizing transformation that minimizes the Frobenius distance to the identity, the LAUT principle resolves the gauge ambiguity intrinsic to conventional approaches and produces effective models that are both physically meaningful and mathematically well‑defined. We proved that this geometric criterion guarantees a lower bound on long‑time dynamical fidelity and enforces symmetry preservation whenever the subspace partition respects the underlying Hamiltonian symmetry. Moreover, we identified the Bloch–Brandow formalism as the natural perturbative counterpart to this variational construction.

Benchmarking on superconducting circuit architectures demonstrates that LAUT‑based methods deliver quantitatively accurate effective models in regimes where standard approximations fail. They reproduce experimental rates in driven cross‑resonance gates, capture counter‑rotating‑wave–induced renormalizations in tunable‑coupler systems, and reveal higher‑order mediated interactions such as the three‑body $XZX+YZY$ term. These results establish the LAUT framework as a practical and reliable tool for both Hamiltonian engineering and the discovery of emergent multi-body physics.

Looking forward, the variational perspective introduced here opens promising avenues. Extending LAUT to non‑Hermitian and open‑system settings may provide consistent gauge fixing in bi-orthogonal representations. Formulating LAUT within the Sambe space could enable accurate Floquet effective Hamiltonians beyond standard high‑frequency expansions. Combining LAUT‑based block diagonalization with embedding strategies may yield scalable effective models in many‑body systems, while the geometric cost functions developed here offer natural objectives for data‑driven Hamiltonian learning. Altogether, the least‑action framework establishes a unified and extensible foundation for high‑precision quantum simulation and control across diverse experimental platforms.

\acknowledgments
We thank Kangqiao Liu, Wenzheng Dong, Yongju Hai, and Shyam Dhamapurkar for their valuable advice and comments, and K. Wei for providing the experimental data used for theoretical verification. This work was supported by the Key-Area Research and Development Program of Guang-Dong Province (Grant No. 2018B030326001), the Quantum Science and Technology-National Science and Technology Major Project (Grants No. 2021ZD0301703), Innovation Program for Quantum Science and Technology (2024ZD0300400), and the Shenzhen Science and Technology Program (KQTD20200820113010023).

\section*{Author Contributions}
X.-H. Deng conceived, supervised, and verified the project; H.-Y. Guan developed the theoretical framework, performed the analytical derivations, and conducted the numerical simulations; Y.-H. Dang performed independent numerical simulations and data analysis; X.-L. Zhu contributed to discussions and verified the analytical derivations; H.-Y. Guan and X.-H. Deng wrote the manuscript with input from all authors, who contributed to the interpretation of the results and approved the final version.

\bibliography{effective_hamiltonian.bib}

\appendix

\section{Detailed Solution to the Least-Action Criterion}
\label{app:ebd-la}

This appendix provides a self-contained derivation of the \emph{exact block diagonalization via least action} (EBD-LAUT) method introduced in Ref.~\cite{cederbaum1989block}. Unlike perturbative approaches that may fail in quasi-dispersive regimes~\cite{goerz2017charting} or under singular perturbations~\cite{holmes2012introduction}, EBD-LAUT yields a fully non-perturbative and mathematically well-defined block-diagonalizing transformation.

We consider effective Hamiltonians of the form
\begin{equation}
    H_{\mathrm{eff}} = T^{\dagger} H T,
\end{equation}
where \( T \) is a unitary operator obtained by minimizing the least‐action (LAUT) functional, i.e., the Frobenius distance between \(T\) and the identity. Let \( S \) denote the matrix whose columns are the exact eigenstates of \( H \), and define its block-diagonal component
\begin{equation}
    S_{\mathrm{BD}} =
    \begin{pmatrix}
        S_{11} & 0 & \cdots & 0 \\
        0 & S_{22} & \cdots & 0 \\
        \vdots & \vdots & \ddots & \vdots \\
        0 & 0 & \cdots & S_{NN}
    \end{pmatrix},
\end{equation}
where each block \(S_{ii}\) spans the target subspace associated with the unperturbed Hamiltonian. Provided that \(S_{\mathrm{BD}}\) is invertible—which is generically satisfied in physical settings—the LAUT minimization has a unique solution. The resulting optimal block-diagonalizing transformation is~\cite{cederbaum1989block}
\begin{equation}\label{eq:BDT_from_LAUT}
    T = S\, S_{\mathrm{BD}}^{\dagger}\, \left(S_{\mathrm{BD}} S_{\mathrm{BD}}^{\dagger}\right)^{-1/2}.
\end{equation}

The structure of Eq.~\eqref{eq:BDT_from_LAUT} admits a clear interpretation.  
The matrix \( S \) diagonalizes \(H\), mapping the bare basis to the exact eigenbasis.  
The factor \( S_{\mathrm{BD}}^{\dagger} \) rotates this eigenbasis so that it aligns with the desired block structure.  
Finally, the factor \( \left(S_{\mathrm{BD}} S_{\mathrm{BD}}^{\dagger}\right)^{-1/2} \) arises from the polar decomposition~\cite{horn2012matrix} and restores unitarity while ensuring that the transformation is as close as possible to the identity in the Frobenius norm.  
Thus, EBD-LAUT performs the minimal global rotation needed to block-diagonalize \(H\), providing both uniqueness and optimality of the resulting unitary.

Unless stated otherwise, the block-diagonalizing transformation \(T\) used throughout this appendix refers to this LAUT-optimal choice. In the corresponding effective Hamiltonian, diagonal entries encode \emph{renormalized energies}, while off-diagonal entries represent \emph{effective couplings} generated by hybridization with eliminated subspaces.

\section{Numerical Illustration of the LAUT-Based Fidelity Bound}
\label{app:bound_verification}

To provide an intuitive demonstration of the fidelity bound derived in 
Sec.~\ref{subsec:LAUT_lower_bound}, we evaluate both sides of the inequality given in Eq.~\eqref{eq:key_bound} for a minimal three-level Q--C--Q model restricted to the single-excitation subspace.
The Hamiltonian takes the form
\begin{equation}
H =
\begin{pmatrix}
\omega_1 & g_1     & 0 \\
g_1      & \omega_{\mathrm{c}} & g_1 \\
0        & g_1     & \omega_2
\end{pmatrix},
\end{equation}
where $\omega_{1,2}$ are qubit frequencies, $\omega_{\mathrm{c}}$ is the coupler frequency,
and $g_{1}$ denotes qubit--coupler couplings.  
The computational subspace is spanned by $\{|100\rangle, |001\rangle\}$, and 
the non-computational state $|010\rangle$ forms an auxiliary level that mediates
hybridization near resonance.  

For each value of $\omega_{\mathrm{c}}$, we compute the exact diagonalization matrix $S$,
construct the least-action block-diagonalizing transformation 
$T$ using Eq.~\eqref{eq:BDT_from_LAUT}, and evaluate the structural deviation
$\|T-\mathcal{I}\|_F^2$.  
The long-time trace fidelity is obtained numerically by computing Eq.~\eqref{eq:long_time_averaged_fidelity} over a time window of $100\,\mu$s,
while the theoretical bound is given by
\begin{equation}
F_{\mathrm{bound}}
= \left(1 - \frac{\|T - \mathcal{I}\|_F^2}{2D} \right)^2,
\qquad D=3.
\end{equation}

Fixing the coupling strength to $g_1/2\pi = 100~\mathrm{MHz}$, we examine both a symmetric configuration ($\omega_1/2\pi = \omega_2/2\pi = 4~\mathrm{GHz}$) and a slightly asymmetric configuration ($\omega_1/2\pi = 4~\mathrm{GHz},\, \omega_2/2\pi = 4.3~\mathrm{GHz}$). In both cases, the coupler frequency is swept over the interval $\omega_{\mathrm{c}}/2\pi \in [3,\, 3.95]~\mathrm{GHz}$. The two configurations exhibit nearly identical global behavior; the inset highlights the pointwise difference between the left- and right-hand sides of the fidelity inequality and reveals the second-order scaling as the norm $\|T-\mathcal{I}\|_F^2$ approaches zero. This minimal model confirms that the LAUT-based fidelity bound is both valid and physically meaningful: it is always satisfied, becomes increasingly tight in the perturbative regime, and remains robust under realistic qubit-frequency asymmetries.

\begin{figure}[t]
    \centering
    \subfigure{
    \begin{minipage}[b]{0.45\columnwidth}
        \centering
        \begin{overpic}[scale=0.3]{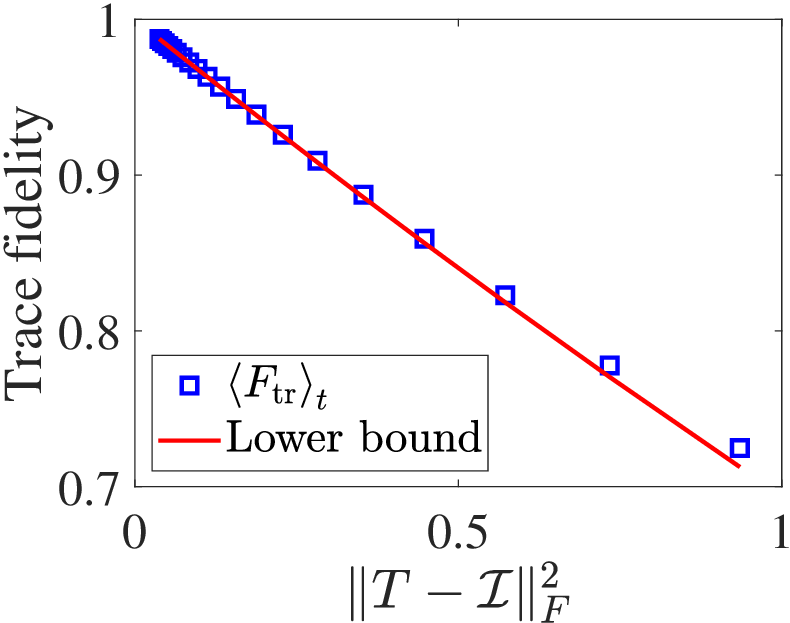}
            \put(0,80){\textbf{(a)}}
        \end{overpic}
    \end{minipage}}
    \subfigure{
    \begin{minipage}[b]{0.44\columnwidth}
        \centering
        \begin{overpic}[scale=0.3]{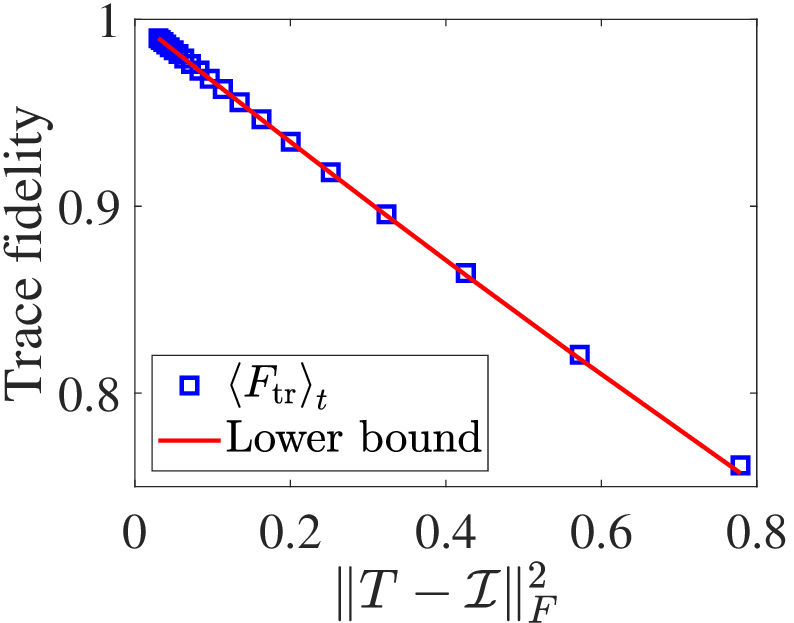}
            \put(0,80){\textbf{(b)}}
        \end{overpic}
    \end{minipage}} 
    \caption{
Numerical validation of the trace-fidelity lower bound.
Scatter points show the long-time trace fidelity 
$\langle F_{\mathrm{tr}} \rangle_t$ as a function of 
$\|T-\mathcal{I}\|_F^2$ for symmetric (a) and asymmetric (b) 
Q--C--Q configurations.  
The red curve is the theoretical LAUT-based lower bound 
$F_{\mathrm{bound}} = (1 - \|T-\mathcal{I}\|_F^2 / 2D)^2$.  
All data satisfy $\langle F_{\mathrm{tr}} \rangle_t \ge F_{\mathrm{bound}}$, 
confirming Eq.~\eqref{eq:key_bound}.  
The inset shows the pointwise difference 
$\Delta F := \langle F_{\mathrm{tr}} \rangle_t - F_{\mathrm{bound}}$.  
In the perturbative regime ($\|T-\mathcal{I}\|_F \ll 1$), 
$\Delta F$ approaches zero quadratically.
}
    \label{fig:fidelity_bound}
\end{figure}

\section{Dynamical Fidelity Bound from the Least Action Unitary Transformation Principle}
\label{app:fidelity_bound}

In this appendix, we provide a rigorous derivation of the fidelity bound presented in Eq.~\eqref{eq:key_bound}. We show that the LAUT criterion is not merely a geometric constraint but establishes a fundamental lower bound on the long-time dynamical fidelity.

\subsection*{Main Logic}
The proof proceeds via three distinct steps:
\begin{enumerate}
    \item \textbf{Geometry (Static):} We prove that the LAUT principle forces the eigenbasis overlap matrix to be Hermitian positive semi-definite. This maximizes the trace overlap and ensures real, non-negative diagonal elements.
    \item \textbf{Dynamics (Kinematic):} We show that under time-averaging, the dynamical fidelity is determined solely by these diagonal overlap elements.
    \item \textbf{Inequality (Synthesis):} We apply the Cauchy-Schwarz inequality to link the dynamical trace directly to the static Frobenius norm minimized by LAUT.
\end{enumerate}

\subsection{Step 1: Geometric Structure of the LAUT Transformation}

Let \( S \) and \( \mathcal{E} \) denote the unitary eigenstate matrices of the full Hamiltonian \( H \) and the effective Hamiltonian \( H_{\mathrm{eff}} \), respectively. The block-diagonalizing transformation connects these bases via \( T = S \mathcal{E}^\dagger \).

The LAUT principle requires finding the transformation \( T \) (and consequently the basis \( \mathcal{E} \)) that minimizes the Frobenius distance to the identity:
\begin{equation}
    \min_{\mathcal{E}} \| T - \mathcal{I} \|_F^2 = \min_{\mathcal{E}} \| S - \mathcal{E} \|_F^2.
\end{equation}
Mathematically, this is the \textit{Orthogonal Procrustes Problem}. To solve it, we perform the Singular Value Decomposition (SVD) of the exact eigenbasis matrix \( S \):
\begin{equation}
    S = U_0 \Sigma V_0^\dagger,
\end{equation}
where \( \Sigma = \mathrm{diag}(\sigma_1, \ldots, \sigma_D) \) contains the singular values. The variationally optimal effective basis is given by the polar factor \( \mathcal{E} = U_0 V_0^\dagger \)~\cite{cederbaum1989block}.

\paragraph*{Crucial Property of the Overlap Matrix.} 
We define the overlap matrix \( W = \mathcal{E}^\dagger S \). Substituting the optimal \( \mathcal{E} \), we find:
\begin{equation}
    W = (V_0 U_0^\dagger) (U_0 \Sigma V_0^\dagger) = V_0 \Sigma V_0^\dagger.
\end{equation}
Since \( \Sigma \) is positive semi-definite, \textbf{\( W \) is a Hermitian positive semi-definite matrix}. This implies two key physical properties:
\begin{enumerate}
    \item The trace of the overlap is real and maximized: \( \mathrm{Tr}(W) = \mathrm{Tr}(\Sigma) = \sum_i \sigma_i \).
    \item The diagonal elements are real and non-negative: \( W_{nn} \geqslant 0 \).
\end{enumerate}

This allows us to relate the trace overlap directly to the LAUT cost function:
\begin{equation}
    \| S - \mathcal{E} \|_F^2 = 2D - 2 \sum_{i=1}^D \sigma_i.
\end{equation}
Rearranging terms yields the static bound:
\begin{equation}
    \mathrm{Tr}(W) = \sum_n W_{nn} = D - \frac{1}{2} \| T - \mathcal{I} \|_F^2.
    \label{eq:static_trace}
\end{equation}

\subsection{Step 2: Dynamical Reduction via Time-Averaging}

We now consider the time-evolution operators \( U(t) = S e^{-iE t} S^\dagger \) and \( U_{\mathrm{eff}}(t) = \mathcal{E} e^{-iE t} \mathcal{E}^\dagger \). The trace fidelity is given by:
\begin{equation}
    \mathrm{Tr}\left[ U^\dagger(t) U_{\mathrm{eff}}(t) \right] = \mathrm{Tr}\left[ W e^{iE t} W^\dagger e^{-iE t} \right].
\end{equation}
Expanding in the eigenbasis indices \(m, n\):
\begin{equation}
    \mathrm{Tr}\left[ U^\dagger(t) U_{\mathrm{eff}}(t) \right] = \sum_{m,n} e^{i(E_n - E_m)t} |W_{mn}|^2.
\end{equation}
Assuming non-degenerate spectra (or that degeneracies are confined within blocks, where the eigenbasis is chosen to diagonalize the Hermitian overlap matrix \( W \)), the long-time average \( \langle \cdot \rangle_t \) eliminates oscillatory terms, while the specific basis choice eliminates static off-diagonal terms. Consequently, the surviving terms are purely diagonal:
\begin{equation}
    \left\langle \mathrm{Tr}\left[ U^\dagger(t) U_{\mathrm{eff}}(t) \right] \right\rangle_t = \sum_n |W_{nn}|^2.
\end{equation}
Applying Jensen’s inequality provides a lower bound on the magnitude:
\begin{equation}
    \left\langle \left| \mathrm{Tr}(U^\dagger(t) U_{\mathrm{eff}}(t)) \right| \right\rangle_t \geqslant \sum_n |W_{nn}|^2.
    \label{eq:dynamic_sum}
\end{equation}

\subsection{Step 3: Deriving the Fidelity Bound}

We now connect the dynamic quantity in Eq.~\eqref{eq:dynamic_sum} to the static quantity in Eq.~\eqref{eq:static_trace}. 
Since \( W \) is Hermitian positive semi-definite (from Step 1), we have \( |W_{nn}|^2 = W_{nn}^2 \). Applying the Cauchy-Schwarz inequality to the real vector of diagonal elements:
\begin{equation}
    \sum_n W_{nn}^2 \geqslant \frac{1}{D} \left( \sum_n W_{nn} \right)^2.
\end{equation}
Substituting Eq.~\eqref{eq:static_trace} into the right-hand side, we obtain:
\begin{equation}
    \sum_n W_{nn}^2 \geqslant \frac{1}{D} \left( D - \frac{1}{2} \| T - \mathcal{I} \|_F^2 \right)^2.
\end{equation}

Finally, combining this with Eq.~\eqref{eq:dynamic_sum}, we arrive at the fidelity bound:
\begin{equation}
    \frac{1}{D} \left\langle \left| \mathrm{Tr}(U^\dagger(t) U_{\mathrm{eff}}(t)) \right| \right\rangle_t \geqslant \left( 1 - \frac{1}{2D} \| T - \mathcal{I} \|_F^2 \right)^2.
\end{equation}

\paragraph*{Conclusion.}
This derivation confirms that the Frobenius norm minimization enforced by the LAUT principle is equivalent to maximizing the lower bound of the dynamical fidelity. The result holds rigorously provided the transformation \( T \) is chosen according to the least-action criterion, which ensures the requisite positive semi-definiteness of the basis overlap matrix.

\section{Gauge Dependence in Local Rotation Schemes}
\label{app:symmetry_GR}
Givens rotations offer a systematic method for iteratively block diagonalizing Hermitian matrices via unitary transformations. However, unless symmetry constraints are explicitly enforced, such operations may inadvertently break conserved quantities intrinsic to the system. We illustrate this phenomenon with two representative examples.

\subsection{Example 1: Two-Qubit System}

Consider a two-qubit Hamiltonian with an XY-type interaction:
\begin{equation}
H = X_1 X_2 + Y_1 Y_2,
\end{equation}
where \( X_i, Y_i, Z_i \) denote the Pauli operators acting on qubit \( i \). This Hamiltonian conserves the total spin along the \( z \)-axis:
\begin{equation}
[H, Z_1 + Z_2] = 0,
\end{equation}
indicating a global \( U(1) \) symmetry associated with rotations about the \( z \)-axis in the spin subspace.

Now apply a single-qubit Givens rotation on qubit 1 about the \( Y \)-axis:
\begin{equation}
G(\theta) = \exp\left(-i \frac{\theta}{2} Y_1\right).
\end{equation}
The rotated Hamiltonian is
\begin{equation}
H' = G^\dagger(\theta) H G(\theta).
\end{equation}
Using standard rotation identities for Pauli operators,
\begin{align}
G^\dagger(\theta) X_1 G(\theta) &= X_1 \cos\theta + Z_1 \sin\theta, \\
G^\dagger(\theta) Y_1 G(\theta) &= Y_1,
\end{align}
we obtain
\begin{align}
H' &= \left(X_1 \cos\theta + Z_1 \sin\theta\right) X_2 + Y_1 Y_2 \notag \\
   &= \cos\theta\, X_1 X_2 + \sin\theta\, Z_1 X_2 + Y_1 Y_2.
\end{align}

To assess symmetry preservation, we compute the commutator with the total \( Z \) operator:
\begin{equation}
[H', Z_1 + Z_2] = 2i (1 - \cos\theta)(Y_1 X_2 + X_1 Y_2) - 2i \sin\theta\, Z_1 Y_2.
\end{equation}
This expression vanishes only when \( \theta = 0 \) (mod \( 2\pi \)), implying that for any nontrivial rotation angle, the rotated Hamiltonian \( H' \) does \emph{not} conserve total \( S_z \). Thus, the Givens rotation explicitly breaks the \( U(1) \) symmetry of the original system.

\subsection{Example 2: Three-Level System}

Consider the symmetric Hamiltonian
\begin{equation}
H = 
\begin{pmatrix}
1 & 0.1 & 0.1 \\
0.1 & 0 & 0 \\
0.1 & 0 & 0 \\
\end{pmatrix},
\end{equation}
defined in the basis \( \{ \ket{0}, \ket{1}, \ket{2} \} \). This Hamiltonian exhibits a \(\mathbb{Z}_2\) symmetry under exchange of levels \( \ket{1} \leftrightarrow \ket{2} \), implemented by the permutation operator
\begin{equation}
\mathcal{P}_{12} = 
\begin{pmatrix}
1 & 0 & 0 \\
0 & 0 & 1 \\
0 & 1 & 0 \\
\end{pmatrix},
\end{equation}
which satisfies
\begin{equation}
[\mathcal{P}_{12}, H] = 0.
\end{equation}
This symmetry implies that the Hilbert space decomposes into symmetric and antisymmetric subspaces under level exchange.

\paragraph{Apply Givens rotation \(G_{01}(\theta)\):}
We now apply a Givens rotation in the \( \ket{0} \)–\( \ket{1} \) subspace:
\begin{equation}
G_{01}(\theta) = 
\begin{pmatrix}
\cos\theta & -\sin\theta & 0 \\
\sin\theta & \cos\theta & 0 \\
0 & 0 & 1 \\
\end{pmatrix}.
\end{equation}
This transformation mixes levels \( \ket{0} \) and \( \ket{1} \) while leaving \( \ket{2} \) unchanged.

The rotated Hamiltonian is
\begin{equation}
H' = G_{01}^\dagger(\theta) H G_{01}(\theta),
\end{equation}
where
\[
G_{01}^\dagger(\theta) = 
\begin{pmatrix}
\cos\theta & \sin\theta & 0 \\
-\sin\theta & \cos\theta & 0 \\
0 & 0 & 1 \\
\end{pmatrix}.
\]

Multiplying out, we find:
\begin{widetext}
\begin{equation}
H' = 
\begin{pmatrix}
0.1 \sin(2\theta) + 0.5 \cos(2\theta) + 0.5 & -0.5 \sin(2\theta) + 0.1 \cos(2\theta) & 0.1 \cos\theta \\
-0.5 \sin(2\theta) + 0.1 \cos(2\theta) & -0.1 \sin(2\theta) - 0.5 \cos(2\theta) + 0.5 & -0.1 \sin\theta \\
0.1 \cos\theta & -0.1 \sin\theta & 0
\end{pmatrix},
\end{equation}
\end{widetext}
with \( \tan(2\theta) = 0.2 \) (i.e., \( \theta \approx 0.0314\pi \)) representing the specific Givens rotation that eliminates the off-diagonal coupling between levels \( \ket{0} \) and \( \ket{1} \).

The resulting matrix elements involving levels \( \ket{1} \) and \( \ket{2} \) are not invariant under exchange,
\[
\mathcal{P}_{12} H' \mathcal{P}_{12}^\dagger \neq H',
\]
$$
H'=\begin{pmatrix}
1.0099 & 0 & 0.0995 \\
0 & -0.0099 & -0.0098 \\
0.0995 & -0.0098 & 0
\end{pmatrix},
$$
demonstrating that the Givens rotation breaks the original \( \mathbb{Z}_2 \) symmetry.

\subsection*{Discussion}

These examples underscore a key limitation of unconstrained Givens rotation schemes: absent explicit symmetry constraints, such transformations can violate physical symmetries intrinsic to the system. In contrast, symmetry-preserving diagonalization approaches—such as those derived from variational principles or least-action formulations—can be designed to respect conserved quantities by construction. This consideration is especially critical in constructing effective Hamiltonians or in symmetry-protected quantum simulation, where preserving the algebraic structure and conserved quantities is essential to ensuring physically faithful modeling.

\section{Validity Conditions and Limitations of Perturbative Block Diagonalization}
\label{app:pbd-validity}
Consider an unperturbed Hamiltonian $H_0$ with eigenvalues in a subset $\ell_0 \subseteq \mathbb{R}$ and the corresponding subspace $P_0 \subseteq \mathcal{H}$. A spectral gap $\Delta$ separating $\ell_0$ from the remainder of the spectrum is defined as
$$|\lambda - \eta| \geq \Delta, \quad \forall \lambda \in \ell_0, \eta \notin \ell_0.$$
For a perturbed Hamiltonian $H = H_0 + \varepsilon V$ with Hermitian $V$ and perturbation parameter $\varepsilon$, an extended interval $\ell \supset \ell_0$ is defined by enlarging $\ell_0$ by $\Delta/2$ on each side. Denoting $P$ as the subspace spanned by eigenstates of $H$ with eigenvalues in $\ell$, perturbative block diagonalization remains valid under the condition~\cite{bravyi2011schrieffer}
$$|\varepsilon| \leqslant \varepsilon_c = \frac{\Delta}{2 \|V\|}.$$
Typically, $\ell_0$ contains the low-energy manifold, so $P_0$ defines a low-energy effective subspace.

Regarding the limitations of the SWT, the complexity arises from the factorial growth in computing $\mathcal{A} = \sum_{n=1}^\infty \mathcal{A}_n \varepsilon^n$, where each $\mathcal{A}_n$ involves increasingly nested commutators~\cite{wurtz2020variational, bravyi2011schrieffer}. Spurious singularities can emerge from divergent series in near-degenerate regimes~\cite{sanz2016beyond}. For multi-block extensions, adaptations like generalized SWT require handling multiple generators, escalating overhead~\cite{hormann2023projective, araya2025pymablock}.

\section{Detailed Review of the Bloch--Brandow Formalism}
\label{app:bb-formalism}
Consider a Hamiltonian of the form $H = H_0 + V$, defined on a $D$-dimensional Hilbert space $\mathcal{H} = \mathcal{H}_C \oplus \mathcal{H}_{NC}$, where $H_0$ is diagonal and $V$ is a weak perturbation. Let $\mathcal{H}C$ denote a computational subspace of dimension $d$, and $\mathcal{H}{NC}$ its orthogonal complement. Define the projectors
$$P = \sum_{i=1}^d |i\rangle\langle i|, \quad Q = \sum_{I=d+1}^D |I\rangle\langle I|,$$
where $H_0 |i\rangle = E_i^{(0)} |i\rangle$ and $H_0 |I\rangle = E_I^{(0)} |I\rangle$. The full eigenstates of $H$ satisfy
$$H |\Psi_k\rangle = E_k |\Psi_k\rangle, \quad k = 1, \dots, D.$$
For the $d$ lowest-energy eigenstates ($k = 1, \dots, d$), decompose
$$|\Psi_k\rangle = |\phi_k\rangle + |\Phi_k\rangle, \quad |\phi_k\rangle = P |\Psi_k\rangle, \quad |\Phi_k\rangle = Q |\Psi_k\rangle.$$

\paragraph*{Connection between EBD-LAUT and PBD-BB approaches}

The LAUT framework admits two variants: one requiring a unitary transformation \(T^{\rm LAUT} = e^G\) with anti-Hermitian \(G\), and one requiring only invertibility transformation $T^{\rm BB} = e^{\hat{\omega}}$. The LAUT method given in Ref.~\cite{cederbaum1989block} adopts the unitary case, while the BB approach aligns with the latter.

In the EBD-LAUT method, the eigenstate matrix of EH is
\begin{equation}
    \mathcal{E}^{\mathrm{LAUT}} = \left(S_{\mathrm{BD}} S_{\mathrm{BD}}^\dagger\right)^{-1/2} S_{\mathrm{BD}},
\end{equation}
providing the closest orthonormal approximation to the exact eigenstates. The corresponding effective Hamiltonian is
\begin{equation}\label{eq:eff_LAUT}
    H_{\mathrm{eff}}^{\mathrm{LAUT}} = P\,e^{-G} H\,e^{G} P,
\end{equation}
which is Hermitian and satisfies the LAUT condition.

In contrast, the BB effective Hamiltonian is defined via the projected eigenvalue equation
\begin{equation}\label{eq:definition_of_EH_BT}
    H_{\mathrm{eff}}^{\mathrm{BB}} \ket{\phi_k} = E_k \ket{\phi_k},
    \qquad k = 1, \dots, d,
\end{equation}
with 
\begin{equation}
    \mathcal{E}^{\mathrm{BB}} = S_{\mathrm{BD}}.
\end{equation}
The associated decoupling operator \(\hat{\omega}\) satisfies
\begin{equation}
    \hat{\omega} = Q \hat{\omega} P, \qquad
    \hat{\omega} P \ket{\Psi_k} = Q \ket{\Psi_k}.
\end{equation}
Using \(\hat{\omega}^2 = 0\) and \(\exp(\pm\hat{\omega}) = \mathcal{I} \pm \hat{\omega}\), the BB effective Hamiltonian becomes
\begin{equation}
    H_{\mathrm{eff}}^{\mathrm{BB}} = P e^{-\hat{\omega}} H e^{\hat{\omega}} P 
    = P H (P + \hat{\omega}),
\end{equation}
which is generally non-Hermitian due to non-orthogonality of the projected eigenstates.

The EBD-LAUT and PBD-BB generators are connected via~\cite{shavitt1980quasidegenerate,kvaal2008geometry}
\begin{equation}
    G = \tanh^{-1}(\hat{\omega} - \hat{\omega}^{\dagger}).
\end{equation}
In practice, Hermiticity can be restored via symmetrization~\cite{richert1976comparison}:
\begin{equation}
    H_{\mathrm{eff}}^{\mathrm{Herm}} = \frac{1}{2}\left(H_{\mathrm{eff}}^{\mathrm{BB}} + H_{\mathrm{eff}}^{\mathrm{BB}\dagger}\right).
\end{equation}

\paragraph*{Perturbative expansion of the effective interaction}

The BB effective Hamiltonian admits the form
\begin{equation}
    H_{\mathrm{eff}}^{\mathrm{BB}} = P H_0 P + V_{\mathrm{eff}}, \qquad
    V_{\mathrm{eff}} = P V_{\mathrm{eff}} P,
\end{equation}
where \(V_{\mathrm{eff}}\) encodes virtual transitions through the non-computational subspace \(\mathcal{H}_{\mathrm{NC}}\). More precisely, it captures virtual excursions from the computational subspace \(\mathcal{H}_{\mathrm{C}}\) into \(\mathcal{H}_{\mathrm{NC}}\) and back, resulting in renormalized frequencies/couplings  within \(\mathcal{H}_{\mathrm{C}}\).  Up to fourth order in \(V\)~\cite{takayanagi2016effective}:
\begin{equation}\label{eq:PBD_BT_results}
\begin{aligned}
    V_{\mathrm{eff}}^{(1)} &= P V P, \\
    V_{\mathrm{eff}}^{(2)} &= P \left[ V (V) \right] P, \\
    V_{\mathrm{eff}}^{(3)} &= P \left[ V (V (V)) - V ((V) V) \right] P, \\
    V_{\mathrm{eff}}^{(4)} &= P \big[ V (V (V (V))) - V (V ((V) V)) \\
    &\quad - V ((V(V)) V) + V (((V) V) V) - V ((V) V (V)) \big] P.
\end{aligned}
\end{equation}
Here, the superoperator \((V)\) is defined as
\begin{equation}\label{eq:superoperator}
    (V)_{I i} = \frac{1}{E^{(0)}_i - E^{(0)}_I} \bra{I} V \ket{i}, \qquad
    i \in \mathcal{H}_{\mathrm{C}},\; I \in \mathcal{H}_{\mathrm{NC}}.
\end{equation}
This formalism yields elementwise closed-form expressions for \(H_{\mathrm{eff}}\), with moderate computational cost for 4th- and 6th-order calculations.

\paragraph*{Advantages of the PBD-BB method}

Compared to the SWT, the PBD-BB framework offers several key advantages:
\begin{enumerate}[(1)]
    \item It satisfies the least-action criterion to all orders of perturbation; upon symmetrization, the resulting effective Hamiltonian matches the traditional perturbation result \( H_{\mathrm{eff}}^{\mathrm{LAUT}} \) up to fourth order~\cite{richert1976comparison}.
    \item It provides explicit, closed-form expressions for \( H_{\mathrm{eff}} \), facilitating both analytical modeling and numerical symbolic implementation.
    \item It naturally extends to multi-block diagonalization, enabling systematic treatment of systems with multiple weakly coupled or symmetry-partitioned sectors.
\end{enumerate}

\section{Symmetry Preservation in the Bloch-Brandow Formalism}
\label{app:symmetry_bb}

In this section, we rigorously prove that the effective Hamiltonian constructed via the Bloch-Brandow (BB) formalism preserves the underlying symmetries of the full system, provided that the subspace partitioning is symmetry-adapted. This formalizes the connection between algebraic consistency and physical conservation laws in our framework.

Consider a full Hamiltonian $H$ acting on a Hilbert space $\mathcal{H}$, and a Hermitian symmetry operator $\mathcal{S}$ such that $[\mathcal{S}, H] = 0$. Let $\{\ket{\Psi_k}\}$ be the common eigenstates of $H$ and $\mathcal{S}$ with eigenvalues $E_k$ and $s_k$, respectively:
\begin{equation}
    H \ket{\Psi_k} = E_k \ket{\Psi_k}, \qquad
    \mathcal{S}\ket{\Psi_k} = s_k \ket{\Psi_k}.
\end{equation}
We partition the Hilbert space into a computational subspace $\mathcal{H}_{\rm C}$ and its orthogonal complement $\mathcal{H}_{\mathrm{NC}}$ using the projector $P = \sum_{i=1}^d \ket{i}\bra{i}$ and $Q = \mathcal{I}-P$. The BB effective Hamiltonian $H_{\mathrm{eff}}^{\mathrm{BB}}$ acting on $\mathcal{H}_{\rm C}$ is defined by the projected eigenvalue equation:
\begin{equation}\label{eq:EH_BB_app_def}
    H_{\mathrm{eff}}^{\mathrm{BB}}\ket{\phi_i} = E_i\ket{\phi_i}, \qquad \ket{\phi_i}=P\ket{\Psi_i}, \quad i=1,\dots,d.
\end{equation}

\noindent \textbf{Corollary~\ref{prop:bb_symmetry}.} \textit{The BB effective Hamiltonian $H_{\mathrm{eff}}^{\mathrm{BB}}$ preserves a system symmetry $\mathcal{S}$ within the target subspace if and only if the projector $P$ commutes with the symmetry operator:}
\begin{equation*}
    [\,H_{\mathrm{eff}}^{\mathrm{BB}},\,\mathcal{S}\,]=0 \;\Longleftrightarrow\; [\,P,\mathcal{S}\,]=0.
\end{equation*}

\begin{proof}
\textbf{Necessity.} Assume $[\,H_{\mathrm{eff}}^{\mathrm{BB}}, \mathcal{S}\,]=0$. Since the projected states $\{\ket{\phi_i}\}$ are eigenstates of $H_{\mathrm{eff}}^{\mathrm{BB}}$ with distinct eigenvalues (assuming no accidental cross-block degeneracy), they must also be eigenstates of $\mathcal{S}$. Given the full system symmetry $\mathcal{S}\ket{\Psi_i} = s_i\ket{\Psi_i}$, we have:
\begin{equation}
    \mathcal{S}\ket{\phi_i} = \mathcal{S}P\ket{\Psi_i} = s_i P\ket{\Psi_i} = P\mathcal{S}\ket{\Psi_i}.
\end{equation}
Rearranging terms yields $[\mathcal{S}, P]\ket{\Psi_i} = 0$. Since the set $\{\ket{\Psi_i}\}$ provides a sufficient basis for the target manifold, it follows that $[\mathcal{S}, P] = 0$.

\textbf{Sufficiency.} Conversely, if $[\,P,\mathcal{S}\,]=0$, the projector $P$ and the symmetry $\mathcal{S}$ share a common eigenbasis. For any projected eigenstate $\ket{\phi_i} = P\ket{\Psi_i}$, we have:
\begin{equation}
    \mathcal{S}\ket{\phi_i} = \mathcal{S}P\ket{\Psi_i} = P\mathcal{S}\ket{\Psi_i} = s_i P\ket{\Psi_i} = s_i \ket{\phi_i}.
\end{equation}
Thus, $\{\ket{\phi_i}\}$ are eigenstates of $\mathcal{S}$. By the definition in Eq.~\eqref{eq:EH_BB_app_def}, $H_{\mathrm{eff}}^{\mathrm{BB}}$ is diagonal in the basis of $\{\ket{\phi_i}\}$, leading to $[\,H_{\mathrm{eff}}^{\mathrm{BB}}, \mathcal{S}\,]=0$.
\end{proof}

\textbf{Remark.} This proof highlights that symmetry breaking in effective models often stems from choosing a computational subspace that ``slices" through symmetry sectors. By ensuring $[P, \mathcal{S}] = 0$, the PBD-BB method inherently respects conserved quantities like the total excitation number in Q-C-Q systems, preventing artificial transitions that plague non-variational approaches.

\section{Benchmarking Model: Qubit--Coupler--Qubit Architecture}
\label{app:QCQ model}
\begin{figure}[tb]
        \centering
        \begin{overpic}[scale=0.3]{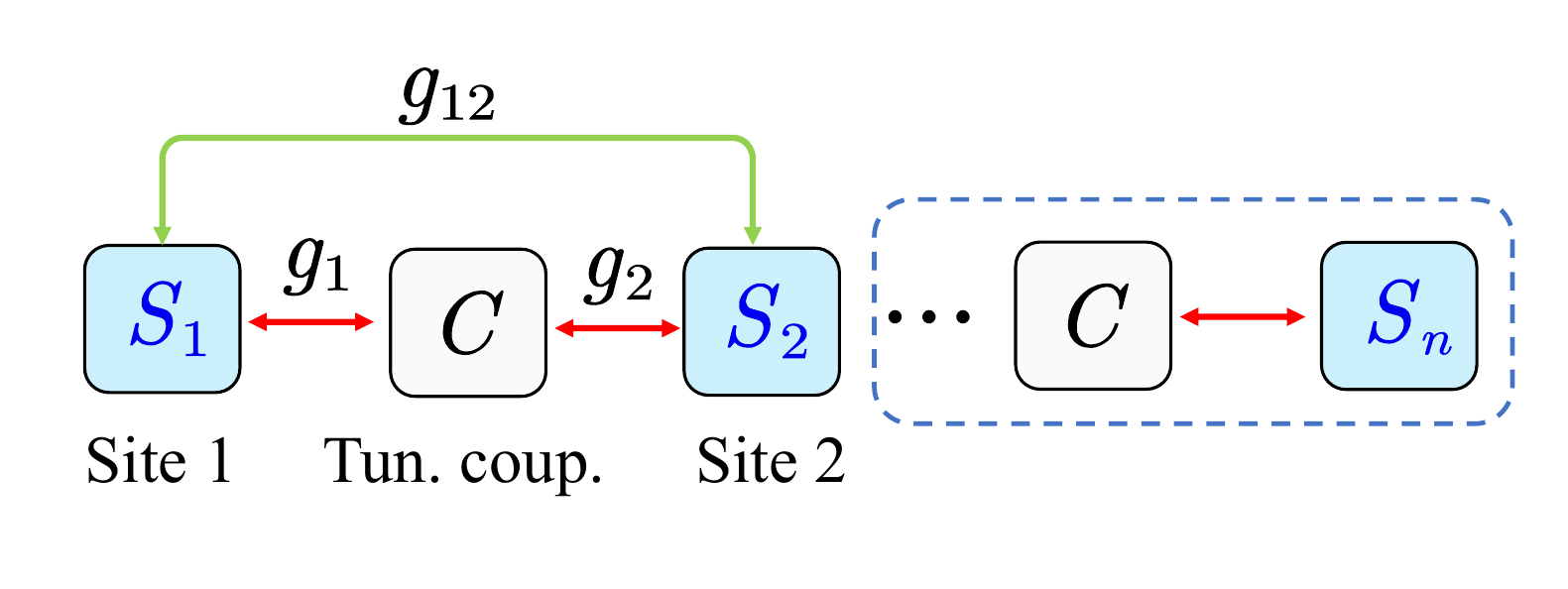}
            \put(0,35){\textbf{(a)}}
        \end{overpic}    
    \caption{
    Schematic of a three-mode system (or generalized \(2n{-}1\)-mode chain, indicated by the dashed extension) with nearest-neighbor exchange couplings (red arrows) and a tunable next-nearest-neighbor coupling (green arrow) mediated by the central mode. In the three-mode configuration, the center serves as a tunable coupler; in the generalized layout, \(n\) lattice sites are connected via \(n{-}1\) mediators.     
    }\label{fig:schematic}
\end{figure}
Our benchmarks employ a hybrid superconducting circuit comprising two multilevel qubits coupled via a frequency-tunable coupler. This qubit--coupler--qubit (Q--C--Q) module serves as the minimal unit for evaluating effective-Hamiltonian constructions, as discussed in Sec.~\ref{sec:benchmarking}. The outer modes act as anharmonic qubits, while the central mode provides a tunable mediated interaction with minimal distortion of local qubit properties. 

The full microscopic Hamiltonian for the three-mode system reads
\begin{equation}
\begin{aligned}
H_{2Q} = &\sum_{i\in\{1,2,c\}}
\left(\omega_i\, a^{\dagger}_i a_i
+ \frac{\beta_i}{2}\, a^{\dagger}_i a^{\dagger}_i a_i a_i\right) \\
&- g_{12}\,\big(a_1 - a_1^{\dagger}\big)
\big(a_2 - a_2^{\dagger}\big) \\
&- \sum_{j\in\{1,2\}} g_j\,
\big(a_j - a_j^{\dagger}\big)
\big(a_c - a_c^{\dagger}\big),
\end{aligned}
\end{equation}
where $i\in\{1,2,c\}$ labels the two qubits and the tunable coupler. The qubit bare frequencies and anharmonicities are denoted by $\omega_{1,2}$ and $\beta_{1,2}$, respectively, and $\omega_{\mathrm{c}}$ is the coupler frequency. In typical devices the qubit--coupler couplings satisfy $g_1, g_2 \gg g_{12}$, ensuring that the mediated interaction dominates over direct qubit--qubit coupling~\cite{yan2018tunable,sung2021realization}.

Operating in the circuit-QED regime~\cite{blais2004cavity,wallraff2004strong}, both qubit and coupler frequencies are flux-tunable. To leading order in circuit capacitances the couplings are approximated by~\cite{yan2018tunable}
\begin{equation}\label{eq:direct_coupling_strength}
\begin{aligned}
g_{j} &\approx \frac{1}{2}\,\frac{C_{jc}}{\sqrt{C_j C_c}}\,
\sqrt{\omega_j\,\omega_{\mathrm{c}}},\quad j=1,2, \\[1mm]
g_{12} &\approx \frac{1}{2}
\left(\frac{C_{12}}{\sqrt{C_1 C_2}}
+ \frac{C_{1c}C_{2c}}{\sqrt{C_1 C_2}\,C_c}\right)
\sqrt{\omega_1\,\omega_2},
\end{aligned}
\end{equation}
where $C_{jc}$ are the qubit--coupler capacitances, $C_{12}$ is the small direct capacitance between the qubits, and $C_1,C_2,C_c$ are the total mode capacitances. This arrangement enables tuning of the effective qubit--qubit interaction via the coupler frequency $\omega_{\mathrm{c}}$ while keeping local qubit properties essentially unaffected, which is critical for the benchmarking presented in the main text.

\section{Relating Effective Coupling to Spectral Frequencies}
\label{app:effective_coupling_FFT}

To connect the effective coupling $\tilde{g}$ extracted from FFT analysis to the underlying Hamiltonian eigenvalues, we derive exact expressions for the case of two identical qubits ($\omega_1 = \omega_2$) and of two small detuned qubits ($|\omega_1 - \omega_2|\sim\varepsilon$). This provides a transparent reference for benchmarking our effective-Hamiltonian methods.

\subsection{Symmetric Qubits: No Detuning}

When the qubit frequencies are equal, the single-excitation subspace spanned by $\{\ket{100}, \ket{010}, \ket{001}\}$ admits a simple analytic solution. The Hamiltonian matrix reads

\begin{equation}
H = \begin{pmatrix}
\Delta & g & g \\
g & 0 & 0 \\
g & 0 & 0
\end{pmatrix},
\end{equation}
where $\Delta = \omega_{\mathrm{c}} - \omega_1$ is the coupler--qubit detuning and $g$ is the coupling strength.

The eigenvalues of this matrix are readily obtained:
\begin{equation}
\lambda_1 = \frac{\Delta + X}{2}, \quad
\lambda_2 = \frac{\Delta - X}{2}, \quad
\lambda_3 = 0,
\end{equation}
with $X = \sqrt{8g^2 + \Delta^2}$. Here $\lambda_1$ corresponds primarily to the coupler excitation, while $\lambda_2$ and $\lambda_3$ span the qubit subspace.

The corresponding normalized eigenvectors are
\begin{equation}
\ket{\phi_1} = \alpha_1 \begin{pmatrix} \frac{\Delta + X}{2g} \\ 1 \\ 1 \end{pmatrix},
\ket{\phi_2} = \alpha_2 \begin{pmatrix} \frac{\Delta - X}{2g} \\ 1 \\ 1 \end{pmatrix},
\ket{\phi_3} = \frac{1}{\sqrt{2}}\begin{pmatrix} 0 \\ -1 \\ 1 \end{pmatrix},
\end{equation}
where $\alpha_{1,2}$ are normalization constants.

Time-domain simulations show a slow oscillation whose frequency corresponds to the energy difference between the two qubit-like eigenstates. This frequency is precisely $\lambda_2 - \lambda_3$, which equals twice the effective coupling:
\begin{equation}
2\tilde{g} = \lambda_2 - \lambda_3.
\end{equation}

To see this explicitly, we use the exact expression for $\tilde{g}$ derived via the least-action criterion (see Appendix~\ref{app:three_level}):
\begin{equation}
\tilde{g} = -\frac{g^2}{\lambda_1} = -\frac{2g^2}{\Delta + X}.
\end{equation}
Direct substitution verifies that
\begin{equation}
-\frac{2g^2}{\Delta + X} = \frac{1}{2}(\Delta - X),
\end{equation}
confirming that the frequency extracted from FFT analysis indeed corresponds to $2\tilde{g}$ in the symmetric case. This relationship serves as a rigorous benchmark for our numerical EH constructions.

\subsection{Small Detuning: Perturbative Treatment}

Next, we introduce a small asymmetry by setting \(\omega_2 = \omega_1 + \varepsilon\), so that the system Hamiltonian becomes
\begin{equation*}
    H = \begin{pmatrix}
        \Delta & g & g \\
        g & 0 & 0 \\
        g & 0 & \varepsilon
    \end{pmatrix}.
\end{equation*}

Using non-degenerate perturbation theory, the eigenvalues are shifted as follows:
\begin{equation}
\begin{aligned}
    \lambda_1 &= \frac{\Delta + X}{2} + \alpha_1^2 \varepsilon + \left( \frac{\alpha_1^2 \alpha_2^2}{X} + \frac{\alpha_1^2}{\Delta + X} \right)\varepsilon^2, \\
    \lambda_2 &= \frac{\Delta - X}{2} + \alpha_2^2 \varepsilon + \left( -\frac{\alpha_1^2 \alpha_2^2}{X} + \frac{\alpha_2^2}{\Delta - X} \right)\varepsilon^2, \\
    \lambda_3 &= \frac{\varepsilon}{2} - \left( \frac{\alpha_1^2}{\Delta + X} + \frac{\alpha_2^2}{\Delta - X} \right)\varepsilon^2.
\end{aligned}
\end{equation}

Assuming \(g \ll \Delta\), we find that the relevant eigenvalue difference becomes
\begin{equation}
    \begin{aligned}
        \lambda_2 - \lambda_3 =& \frac{\Delta - X}{2} + (\alpha_2^2 - 0.5)\varepsilon+ \mathcal{O}(\varepsilon^3) \\
    &+\left( -\frac{\alpha_1^2 \alpha_2^2}{X} + \frac{\alpha_2^2}{\Delta - X} + \frac{\alpha_1^2}{\Delta + X} + \frac{\alpha_2^2}{\Delta - X} \right)\varepsilon^2
    .
    \end{aligned}
\end{equation}

Meanwhile, the effective coupling derived from the least action method is
\begin{equation}
    \tilde{g} = -\frac{g^2}{2} \left( \frac{1}{\lambda_1} + \frac{1}{\lambda_1 - \varepsilon} \right),
\end{equation}
which, after expansion in \(\varepsilon\), yields
\begin{equation}
\begin{aligned}
    2\tilde{g} &= (\lambda_2 - \lambda_3) 
    + \frac{4g^2}{(\Delta + X)^2}(2\alpha_1^2 - 1)\varepsilon 
    - (\alpha_2^2 - 0.5)\varepsilon \\
    &\quad - \left( -\frac{\alpha_1^2 \alpha_2^2}{X} + \frac{\alpha_2^2}{\Delta - X} + \frac{\alpha_1^2}{\Delta + X} + \frac{\alpha_2^2}{\Delta - X} \right)\varepsilon^2,
\end{aligned}
\end{equation}
where the coefficients \(\alpha_1^2\) and \(\alpha_2^2\) are given by
\begin{equation*}
\begin{aligned}
    \alpha_1^2 &= \frac{4g^2}{(\Delta + X)^2 + 8g^2}, \\
    \alpha_2^2 &= \frac{4g^2}{(\Delta - X)^2 + 8g^2}.
\end{aligned}
\end{equation*}

These results confirm that the frequency extracted via Fourier analysis corresponds exactly to \(2\tilde{g}\) in the symmetric limit, and remains analytically connected to \(2\tilde{g}\) through perturbative expressions under weak detuning.

\section{Three-Level Q--C--Q System: Benchmarking Effective-Hamiltonian Methods}
\label{app:three_level}

The three-level Q--C--Q architecture provides a minimal yet analytically tractable setting for benchmarking effective-Hamiltonian (EH) constructions. We focus on three widely used approaches: (i) exact block diagonalization via the least-action criterion (EBD-LAUT), (ii) its perturbative counterpart in the bare basis (PBD-BB), and (iii) the Schrieffer--Wolff transformation (SWT). Our main goal is to establish that \emph{LAUT and PBD-BB coincide order-by-order up to fourth order}, whereas \emph{SWT deviates starting at fourth order}, revealing the structural advantage of LAUT-aligned methods.

\subsection{Model}
Two transmons, \(Q_1\) and \(Q_2\), are coupled via a tunable coupler \(C\). Restricting to the single-excitation subspace \(\{\ket{100},\ket{010},\ket{001}\}\), the Hamiltonian reads
\begin{equation}
H_1 =
\begin{pmatrix}
0 & g_1 & g_2 \\
g_1 & \Delta_1 & 0 \\
g_2 & 0 & \Delta_2
\end{pmatrix}
+ \omega_{\mathrm{c}} \mathcal{I}_3,
\end{equation}
with \(\Delta_i = \omega_i - \omega_{\mathrm{c}}\). Eliminating the coupler state \(\ket{010}\) defines an effective two-level Hamiltonian with coupling \(\tilde g\) and renormalized qubit energies \(\tilde\omega_{1,2}\).

\subsection{Exact Block Diagonalization via the Least-Action Criterion (EBD-LAUT)}
In the LAUT approach, the coupler eigenvalue \(\lambda\) is fixed by continuity as \(g_{1,2}\rightarrow0\), so that \(\lambda\to\omega_{\mathrm{c}}\). The non-Hermitian block-diagonalization
\[
\mathcal{H} = U^{-1} H U
\]
yields the qubit block
\begin{equation}
\begin{pmatrix}
\omega_1+\dfrac{g_1^2}{\omega_1-\lambda} & \dfrac{g_1 g_2}{\omega_2-\lambda} \\[1mm]
\dfrac{g_1 g_2}{\omega_1-\lambda} & \omega_2+\dfrac{g_2^2}{\omega_2-\lambda}
\end{pmatrix}.
\end{equation}
Defining \(\delta_i = \omega_i - \lambda\) and expanding \(\lambda\) perturbatively up to fourth order in \(g_i/\Delta_i\) gives
\[
\lambda^{(4)} = \omega_{\mathrm{c}} - \frac{g_1^2}{\Delta_1} - \frac{g_2^2}{\Delta_2} + \left(\frac{g_1^2}{\Delta_1^2} + \frac{g_2^2}{\Delta_2^2}\right) \left(\frac{g_1^2}{\Delta_1} + \frac{g_2^2}{\Delta_2}\right).
\]
The resulting fourth-order effective parameters are
\begin{align}
\tilde{g}_{LAUT}^{(4)} &= \frac{g_1 g_2}{2} \Big[\left(\frac{1}{\Delta_1} + \frac{1}{\Delta_2}\right)
- \left(\frac{g_1^2}{\Delta_1} + \frac{g_2^2}{\Delta_2}\right)\left(\frac{1}{\Delta_1^2} + \frac{1}{\Delta_2^2}\right) \Big], \\
\tilde{\omega}_{1,LAUT}^{(4)} &= \omega_1 + \frac{g_1^2}{\Delta_1} - \frac{g_1^4}{\Delta_1^3} - \frac{g_1^2 g_2^2}{\Delta_1^2 \Delta_2}, \\
\tilde{\omega}_{2,LAUT}^{(4)} &= \omega_2 + \frac{g_2^2}{\Delta_2} - \frac{g_2^4}{\Delta_2^3} - \frac{g_1^2 g_2^2}{\Delta_1 \Delta_2^2}.
\end{align}

\subsection{Perturbative Block Diagonalization (PBD-BB)}
In the PBD-BB approach, the effective coupling arises from two-step virtual processes:
\begin{equation}
\ket{100}\xrightarrow{g_1}\ket{010}\xrightarrow{g_2}\ket{001}, \quad
\ket{001}\xrightarrow{g_2}\ket{010}\xrightarrow{g_1}\ket{100},
\end{equation}
giving a second-order contribution \(\tilde g^{(2)}_{BB} = \frac{g_1 g_2}{2} (\frac{1}{\Delta_1} + \frac{1}{\Delta_2})\). Extending to fourth order by summing all relevant four-step paths reproduces exactly the LAUT results:
\[
\tilde g_{BB}^{(4)} = \tilde g_{LAUT}^{(4)}, \quad \tilde{\omega}_{i,BB}^{(4)} = \tilde{\omega}_{i,LAUT}^{(4)},\ i=1,2.
\]

\subsection{Analytic Derivation via the Schrieffer–Wolff Transformation}\label{app:3_level_SWT}

The SWT expands the generator \(\mathcal{A}\) as
\[
\mathcal{A}=\sum_{n=1}^\infty \mathcal{A}_n\varepsilon^n.
\]
Following Ref.~\cite{bravyi2011schrieffer}, we define the off-diagonal interaction
\[
V_{\text{od}}=\begin{pmatrix}
	0 & 0 & g_1\\[1mm]
	0 & 0 & g_2\\[1mm]
	g_1 & g_2 & 0
\end{pmatrix},
\]
and obtain the first-order generator
\[
\mathcal{A}_1=\begin{pmatrix}
	0 & 0 & \dfrac{g_1}{\Delta_1}\\[1mm]
	0 & 0 & \dfrac{g_2}{\Delta_2}\\[1mm]
	-\dfrac{g_1}{\Delta_1} & -\dfrac{g_2}{\Delta_2} & 0
\end{pmatrix}.
\]
Its commutator with \(V_{\text{od}}\) reproduces the second-order corrections equivalent to those of the PBD-BB method. At fourth order, the effective hopping comprises two contributions:
\begin{widetext}
\begin{align*}
    T_1 &= \frac{g_1g_2^3\Delta_1\left(4\Delta_1^2+\Delta_1\Delta_2+3\Delta_2^2\right)
    +g_1^3g_2\Delta_2\left(3\Delta_1^2+\Delta_1\Delta_2+4\Delta_2^2\right)}
    {-6\Delta_1^3\Delta_2^3},\\[1mm]
    T_2 &= \frac{g_1g_2\left(g_1^2\Delta_2^2\left(7\Delta_1+\Delta_2\right)
    +g_2^2\Delta_1^2\left(\Delta_1+7\Delta_2\right)\right)}
    {24\Delta_1^3\Delta_2^3}.
\end{align*}
The combined effective hopping is then
\[
\tilde{g}_{SWT}^{(4)}=\frac{g_1g_2^3\Delta_1\left(5\Delta_1^2-\Delta_1\Delta_2+4\Delta_2^2\right)
+g_1^3g_2\Delta_2\left(4\Delta_1^2-\Delta_1\Delta_2+5\Delta_2^2\right)}
{-8\Delta_1^3\Delta_2^3}.
\]
\end{widetext}
Similarly, the renormalized frequencies are renormalized as
\begin{align*}
    \tilde{\omega}_{1,SWT}^{(4)} &= \omega_1+\frac{g_1^2}{\Delta_1}-\frac{g_1^4}{\Delta_1^3}-\frac{g_1^2g_2^2(\Delta_1+3\Delta_2)}{4\Delta_1^2\Delta_2^2},\\[1mm]
    \tilde{\omega}_{2,SWT}^{(4)} &= \omega_2+\frac{g_2^2}{\Delta_2}-\frac{g_2^4}{\Delta_2^3}-\frac{g_1^2g_2^2(3\Delta_1+\Delta_2)}{4\Delta_1^2\Delta_2^2}.
\end{align*}

\begin{figure}[t]
	\centering
    \subfigure{\begin{minipage}[b]{0.47\columnwidth}
        \centering
            \begin{overpic}[scale=0.3]{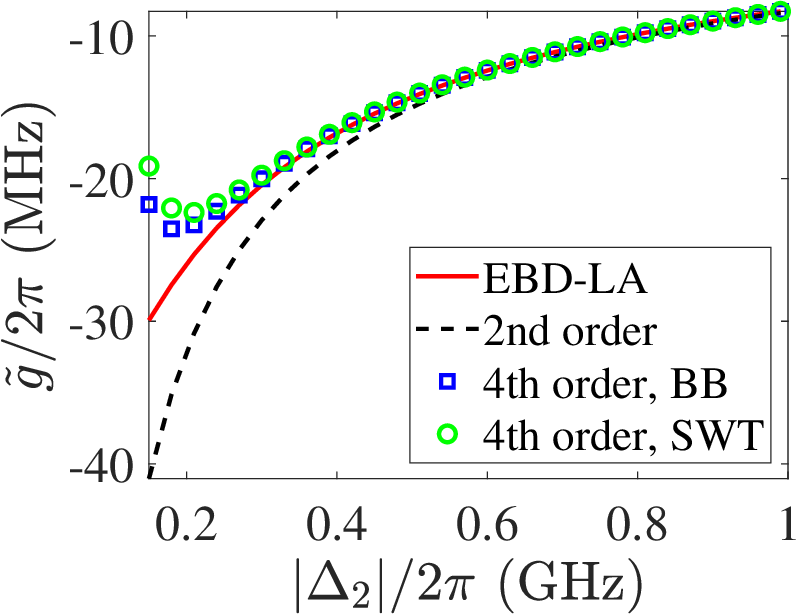}
            \put(0,80){\textbf{(a)}}
        \end{overpic}
    \end{minipage}
    \label{fig:compare_three_coupling}
    }
    \subfigure{\begin{minipage}[b]{0.47\columnwidth}
        \centering
            \begin{overpic}[scale=0.3]{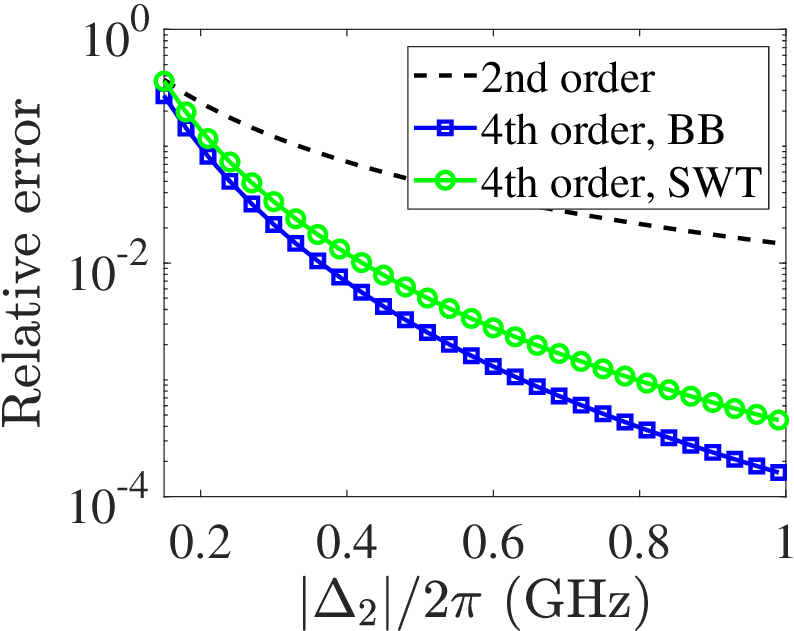}
            \put(0,80){\textbf{(b)}}
        \end{overpic}
    \end{minipage}
    \label{fig:compare_three_coupling_error}
    }
    \subfigure{\begin{minipage}[b]{0.47\columnwidth}
        \centering
            \begin{overpic}[scale=0.3]{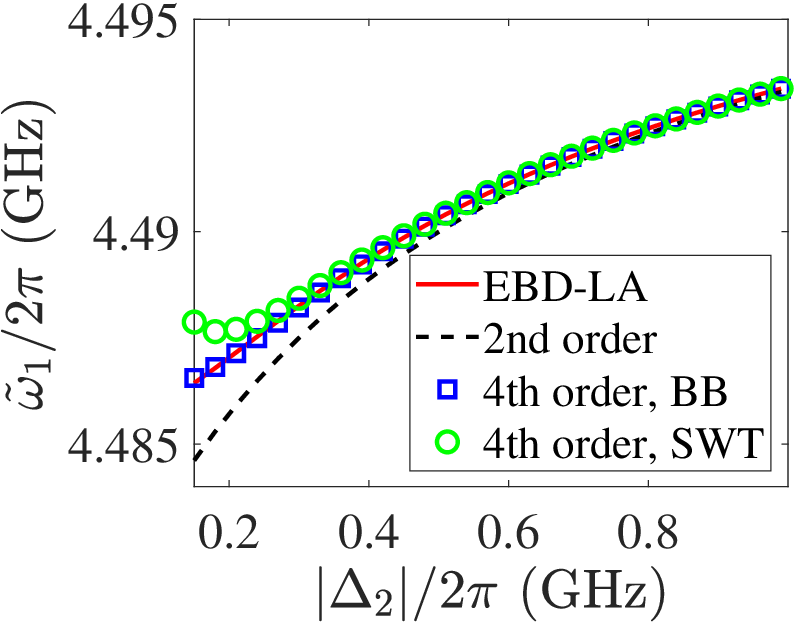}
            \put(0,80){\textbf{(c)}}
        \end{overpic}
    \end{minipage}
    \label{fig:compare_three_omega1}
    }
    \subfigure{\begin{minipage}[b]{0.47\columnwidth}
        \centering
            \begin{overpic}[scale=0.3]{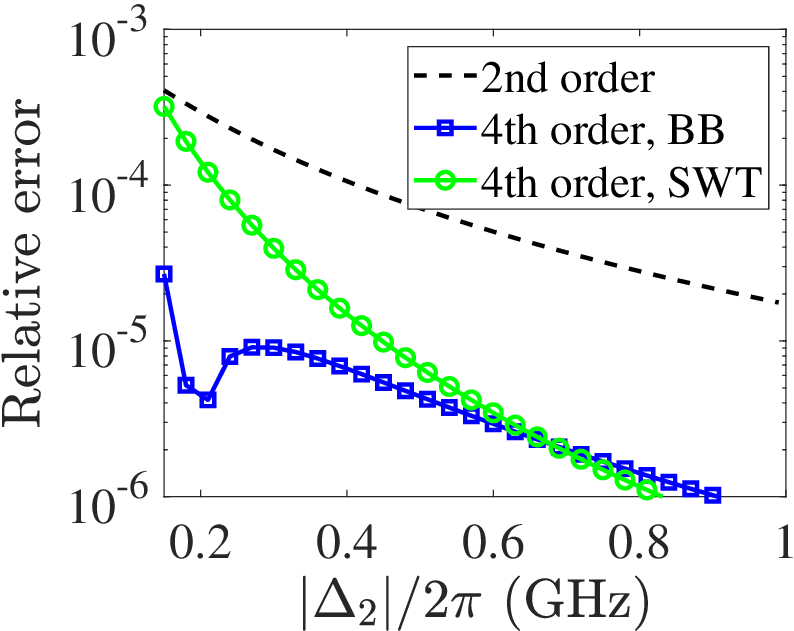}
            \put(0,80){\textbf{(d)}}
        \end{overpic}
    \end{minipage}
    \label{fig:compare_three_omega1_error}
    }

    \subfigure{\begin{minipage}[b]{0.47\columnwidth}
        \centering
            \begin{overpic}[scale=0.3]{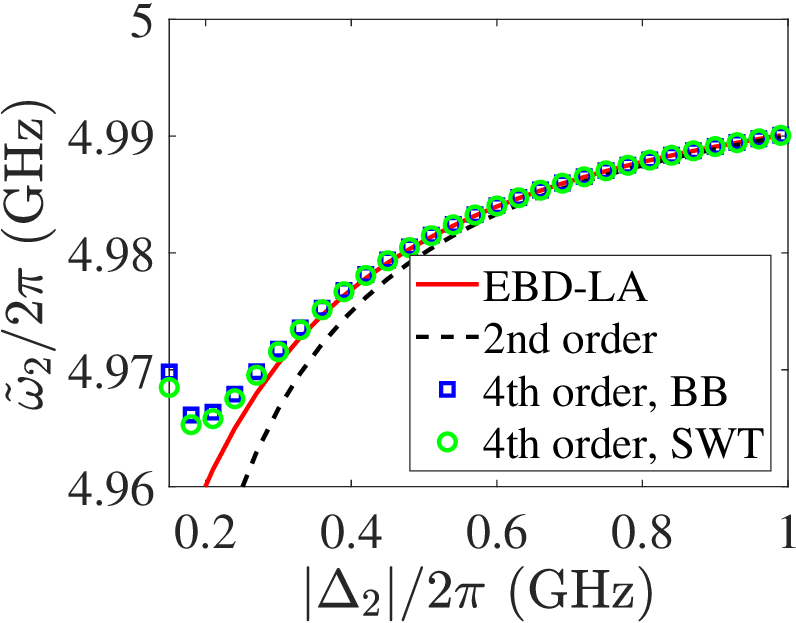}
            \put(0,80){\textbf{(e)}}
        \end{overpic}
    \end{minipage}
    \label{fig:compare_three_omega2}
    }
    \subfigure{\begin{minipage}[b]{0.47\columnwidth}
        \centering
            \begin{overpic}[scale=0.3]{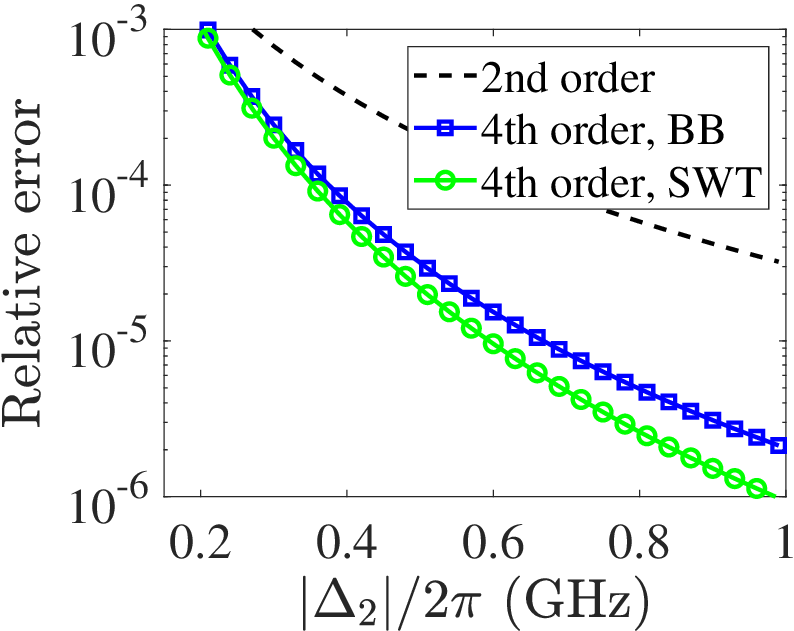}
            \put(0,80){\textbf{(f)}}
        \end{overpic}
    \end{minipage}
    \label{fig:compare_three_omega2_error}
    }
    \subfigure{\begin{minipage}[b]{0.47\columnwidth}
        \centering
            \begin{overpic}[scale=0.3]{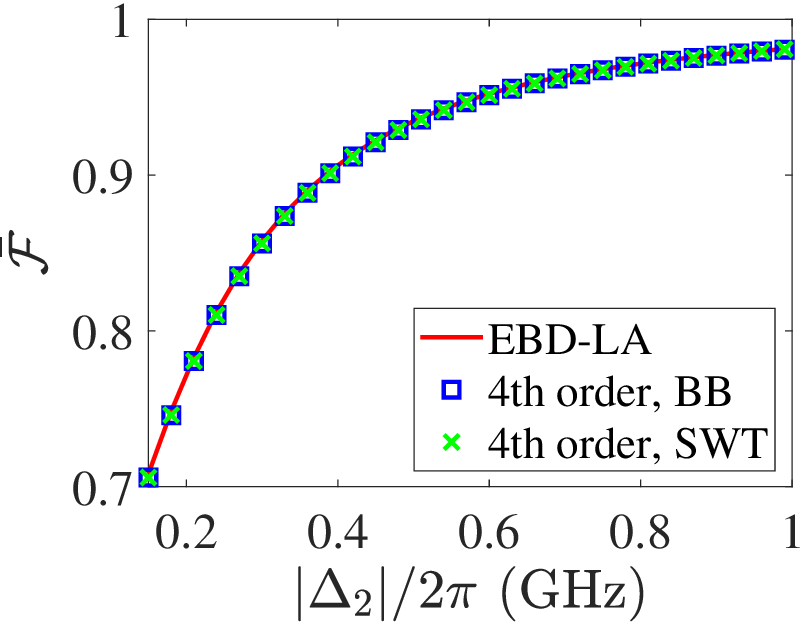}
            \put(0,80){\textbf{(g)}}
        \end{overpic}
    \end{minipage}
    \label{fig:compare_three_dynamics}
    }
    \subfigure{\begin{minipage}[b]{0.47\columnwidth}
        \centering
            \begin{overpic}[scale=0.3]{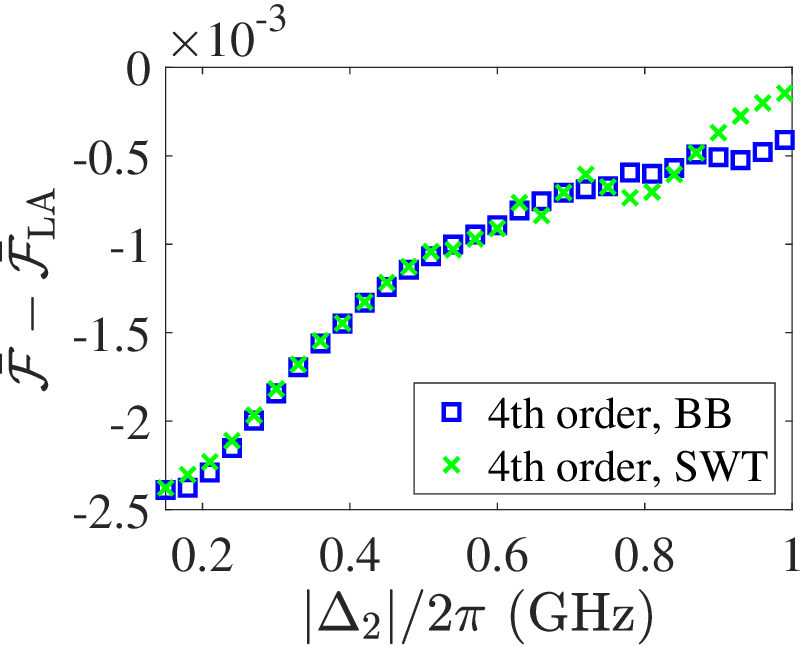}
            \put(0,80){\textbf{(h)}}
        \end{overpic}
    \end{minipage}
    \label{fig:compare_three_dynamics_error}
    }
    \caption{(Color online) 
Comparison of effective parameters extracted for the three-level system using four appraoches: exact block diagonalization via the least action unitary transformation principle (EBD-LAUT, red solid), second-order perturbation theory (black dashed), fourth-order PBD-BB (blue squares), and fourth-order SWT (green circles).  
\textbf{(a,c,e)}~Effective coupling strengths \(\tilde{g}\), and renormalized frequencies \(\tilde{\omega}_1\), \(\tilde{\omega}_2\).  
\textbf{(b,d,f)}~Relative errors with respect to the LAUT result, defined as \(\left|\left(\tilde{g} - \tilde{g}_{\mathrm{LAUT}}\right) / \tilde{g}_{\mathrm{LAUT}}\right|\), and analogously for \(\tilde{\omega}_1\) and \(\tilde{\omega}_2\). The PBD-BB method shows significantly better agreement with the EBD-LAUT benchmark than SWT.  
\textbf{(g)}~Time-averaged fidelity \(\bar{\mathcal{F}}\) between the full and effective Hamiltonian dynamics over a total evolution time of \(T = 100\,\mu\mathrm{s}\).  
\textbf{(h)}~Fidelity difference \(\bar{\mathcal{F}} - \bar{\mathcal{F}}_{\mathrm{LAUT}}\), highlighting comparable dynamical performance of PBD-BB and SWT in this simple model.
}

	\label{fig:compare_three_level}
\end{figure}
\subsection{Numerical Comparison and Summary}
Figures~\ref{fig:compare_three_level}(a–f) show representative detuning sweeps. Key observations:
\begin{itemize}
    \item PBD-BB tracks EBD-LAUT almost exactly for \(\tilde g\) and \(\tilde{\omega}_{1,2}\).
    \item SWT exhibits systematic deviations at fourth order, particularly near small detuning or strong hybridization.
\end{itemize}
Time-domain fidelity comparisons [Figs.~\ref{fig:compare_three_level}(g–h)] indicate that all methods yield qualitatively similar dynamics, but  EBD-LAUT and PBD-BB provide mutually similar effective parameters.

\noindent
\textbf{Summary.} This three-level analysis establishes a clear hierarchy of accuracy:
\begin{enumerate}
    \item \emph{LAUT and PBD-BB coincide exactly up to fourth order}, confirming the reliability of perturbative LAUT constructions.
    \item \emph{SWT deviates starting at fourth order}, a limitation of fixed-generator expansions.
    \item \emph{LAUT provides a consistent, symmetry-aware baseline} for effective models of tunable-coupler architectures.
\end{enumerate}

\section{Calculation Details on CR Coupling} 

\subsection{Derivation of Bare Parameters}\label{app:bare_parameters_CR}

We consider the full Hamiltonian of a coupled transmon system in the absence of the CR drive:
\begin{equation}
\begin{aligned}
    H_0 =& \sum_{i\in\{1,2\}} \Biggl[ \omega_i\, a_i^{\dagger}a_i 
    + \frac{\beta_i}{2}\, a_i^{\dagger}a_i \left( a_i^{\dagger}a_i - 1 \right) \Biggr]\\
    &+ J \left( a_1^{\dagger}+a_1 \right) \left( a_2^{\dagger}+a_2 \right),
\end{aligned}
\end{equation}
where \(\omega_i\) and \(\beta_i\) denote the bare frequencies and anharmonicities, respectively, and \(J\) is the direct coupling strength. 

From Ref.~\cite{wei2022hamiltonian}, the experimentally measured parameters (dressed by the coupling $J$) are:
\(
\tilde{\omega}_1/2\pi = 5.016\,\text{GHz}, \tilde{\omega}_2/2\pi = 4.960\,\text{GHz}, \tilde{\beta}_1/2\pi = -0.287\,\text{GHz}, \tilde{\beta}_2/2\pi = -0.283\,\text{GHz}.
\)

A standard perturbative expansion allows us to relate the bare parameters in \(H_0\) to the experimentally measured ones:
\begin{equation}
\begin{aligned}
\tilde{\omega}_1 =& \omega_1 + \frac{J^2}{\Delta_{12}} - \frac{2J^2}{\Sigma_{12}+\beta_1}, \\
\tilde{\omega}_2 =& \omega_2 - \frac{J^2}{\Delta_{12}} - \frac{2J^2}{\Sigma_{12}+\beta_2}, \\
\tilde{\beta}_1 =& \beta_1 + 2J^2 \left( \frac{1}{\Delta_{12}+\beta_1} - \frac{1}{\Delta_{12}} \right)
\\ &+ J^2 \left( \frac{4}{\Sigma_{12}+\beta_1} - \frac{3}{\Sigma_{12}+2\beta_1} \right), \\
\tilde{\beta}_2 =& \beta_2 + 2J^2 \left( \frac{1}{\Delta_{12}} - \frac{1}{\Delta_{12}-\beta_2} \right)
\\ &+ J^2 \left( \frac{4}{\Sigma_{12}+\beta_2} - \frac{3}{\Sigma_{12}+2\beta_2} \right),
\end{aligned}
\end{equation}
where \(\Delta_{12} = \omega_1 - \omega_2\) and \(\Sigma_{12} = \omega_1 + \omega_2\).

Solving these expressions yields the following estimates for the bare parameters:
\(
\omega_1/2\pi \approx 5.0149\,\text{GHz}, \omega_2/2\pi \approx 4.9611\,\text{GHz}, 
\beta_1/2\pi \approx -0.2847\,\text{GHz}, \beta_2/2\pi \approx -0.2853\,\text{GHz},
\)
corresponding to a detuning of \(\Delta_{12}/2\pi \approx 53\,\text{MHz}\).

\subsection{Derivation of CR Coupling Strength}
\label{app:CR_coupling_strength}

Due to parity constraints, third-order contributions to the CR interaction vanish. The leading corrections arise at fourth order, involving multiple virtual excitation pathways. For example, the state $\ket{00}$ in the first computational subspace $\mathcal{H}_1$ contributes five dominant routes, while $\ket{10}$ in the second subspace $\mathcal{H}_2$ contributes twelve. These processes are essential for accurately capturing the effective ZX interaction mediated by the interplay between the drive and inter-qubit coupling. Consequently, the effective strength $\tilde{\nu}_{ZX}$ is determined by the net difference between the aggregate fourth-order contributions from the twelve pathways in $\mathcal{H}_2$ and the five in $\mathcal{H}_1$.

\paragraph*{Effective coupling in $\mathcal{H}_1$: five fourth-order routes between $\ket{00}$ and $\ket{01}$.}
We enumerate the five contributing paths and their corresponding amplitudes:

\begin{enumerate}[(1)]
    \item Route: $\ket{00} \xleftrightarrow{\Omega} \ket{10} \xleftrightarrow{J} \ket{01} \xleftrightarrow{J} \ket{10} \xleftrightarrow{J} \ket{01}$  
    \begin{equation}
        \frac{\Omega J^3}{4\Delta_{12}} \left( \frac{1}{\Delta_{1\mathrm{d}}^2} + \frac{1}{\Delta_{12}^2} \right)
    \end{equation}
    
    \item Route: $\ket{00} \xleftrightarrow{\Omega} \ket{10} \xleftrightarrow{J} \ket{01} \xleftrightarrow{\Omega} \ket{11} \xleftrightarrow{\Omega} \ket{01}$  
    \begin{equation}
        \frac{\Omega^3 J}{16\Delta_{1\mathrm{d}}} \left( \frac{1}{\Delta_{1\mathrm{d}} \Sigma_{12}} + \frac{1}{\Delta_{12}^2} \right)
    \end{equation}

    \item Route: $\ket{00} \xleftrightarrow{\Omega} \ket{10} \xleftrightarrow{\Omega} \ket{20} \xleftrightarrow{J} \ket{11} \xleftrightarrow{\Omega} \ket{01}$  
    \begin{equation}
        -\frac{\Omega^3 J}{8\Delta_{1\mathrm{d}}} \left[ \frac{1}{\Sigma_{12}(2\Delta_{1\mathrm{d}}+\beta_1)} + \frac{1}{\Delta_{12}(\Delta_{1\mathrm{d}}+\Delta_{12}+\beta_1)} \right]
    \end{equation}

    \item Route: $\ket{00} \xleftrightarrow{\Omega} \ket{10} \xleftrightarrow{\Omega} \ket{20} \xleftrightarrow{\Omega} \ket{10} \xleftrightarrow{J} \ket{01}$  
    \begin{equation}
        -\frac{\Omega^3 J}{8} \left[ \frac{1}{\Delta_{1\mathrm{d}}^2 (2\Delta_{1\mathrm{d}} + \beta_1)} + \frac{1}{\Delta_{12}^2 (\Delta_{1\mathrm{d}} + \Delta_{12} + \beta_1)} \right]
    \end{equation}

    \item Route: $\ket{00} \xleftrightarrow{\Omega} \ket{10} \xleftrightarrow{\Omega} \ket{00} \xleftrightarrow{\Omega} \ket{10} \xleftrightarrow{J} \ket{01}$  
    \begin{equation}
        \frac{\Omega^3 J}{16\Delta_{1\mathrm{d}}} \left( \frac{1}{\Delta_{1\mathrm{d}}^2} + \frac{1}{\Delta_{12}^2} \right)
    \end{equation}
\end{enumerate}

\paragraph*{Effective coupling in $\mathcal{H}_2$: twelve fourth-order routes between $\ket{10}$ and $\ket{11}$.}
The following virtual processes contribute to the effective ZX interaction in the second excitation subspace:

\begin{enumerate}[(1)]
    \item $\ket{10} \xleftrightarrow{J} \ket{01} \xleftrightarrow{J} \ket{10} \xleftrightarrow{J} \ket{01} \xleftrightarrow{\Omega} \ket{11}$
    \begin{equation}
        -\frac{\Omega J^3}{4\Delta_{12}} \left( \frac{1}{\Delta_{1\mathrm{d}}^2} + \frac{1}{\Delta_{12}^2} \right)
    \end{equation}

    \item $\ket{10} \xleftrightarrow{J} \ket{01} \xleftrightarrow{J} \ket{10} \xleftrightarrow{\Omega} \ket{20} \xleftrightarrow{J} \ket{11}$
    \begin{equation}
        \frac{\Omega J^3}{2\Delta_{12}} \left( \frac{1}{\Delta_{1\mathrm{d}}(\Delta_{12}+\beta_1)} - \frac{1}{(\Delta_{1\mathrm{d}} + \beta_1)^2} \right)
    \end{equation}

    \item $\ket{10} \xleftrightarrow{J} \ket{01} \xleftrightarrow{\Omega} \ket{11} \xleftrightarrow{J} \ket{20} \xleftrightarrow{J} \ket{11}$
    \begin{equation}
        \frac{\Omega J^3}{2(\Delta_{12}+\beta_1)} \left[ \frac{1}{\Delta_{1\mathrm{d}}^2} - \frac{1}{\Delta_{12}(\Delta_{1\mathrm{d}}+\beta_1)} \right]
    \end{equation}

    \item $\ket{10} \xleftrightarrow{J} \ket{01} \xleftrightarrow{\Omega} \ket{11} \xleftrightarrow{J} \ket{02} \xleftrightarrow{J} \ket{11}$
    \begin{equation}
        \frac{\Omega J^3}{2(\Delta_{12}-\beta_2)} \left[ \frac{1}{\Delta_{12}(\Delta_{2\mathrm{d}}+\beta_2-\Delta_{12})} - \frac{1}{\Delta_{1\mathrm{d}}^2} \right]
    \end{equation}

    \item $\ket{10} \xleftrightarrow{\Omega} \ket{00} \xleftrightarrow{\Omega} \ket{10} \xleftrightarrow{J} \ket{01} \xleftrightarrow{\Omega} \ket{11}$
    \begin{equation}
        -\frac{\Omega^3 J}{16\Delta_{1\mathrm{d}}} \left( \frac{1}{\Delta_{12}^2} + \frac{1}{\Delta_{1\mathrm{d}} \Sigma_{12}} \right)
    \end{equation}

    \item $\ket{10} \xleftrightarrow{\Omega} \ket{00} \xleftrightarrow{\Omega} \ket{10} \xleftrightarrow{\Omega} \ket{20} \xleftrightarrow{J} \ket{11}$
    \begin{equation}
        \frac{\Omega^3 J}{8\Delta_{1\mathrm{d}}} \left[ \frac{1}{\Sigma_{12}(\Delta_{12}+\beta_1)} - \frac{1}{(\Delta_{1\mathrm{d}}+\beta_1)^2} \right]
    \end{equation}

    \item $\ket{10} \xleftrightarrow{\Omega} \ket{20} \xleftrightarrow{J} \ket{11} \xleftrightarrow{J} \ket{20} \xleftrightarrow{J} \ket{11}$
    \begin{equation}
        \frac{\Omega J^3}{\Delta_{12}+\beta_1} \left[ \frac{1}{(\Delta_{1\mathrm{d}}+\beta_1)^2} + \frac{1}{(\Delta_{12}+\beta_1)^2} \right]
    \end{equation}

    \item $\ket{10} \xleftrightarrow{\Omega} \ket{20} \xleftrightarrow{J} \ket{11} \xleftrightarrow{J} \ket{02} \xleftrightarrow{J} \ket{11}$
    \begin{equation}
        -\frac{\Omega J^3}{\Delta_{12} - \beta_2} \left[ \frac{1}{(\Delta_{1\mathrm{d}}+\beta_1)(\Delta_{2\mathrm{d}}+\beta_2-\Delta_{12})} + \frac{1}{(\Delta_{12}+\beta_1)^2} \right]
    \end{equation}

    \item $\ket{10} \xleftrightarrow{\Omega} \ket{20} \xleftrightarrow{\Omega} \ket{30} \xleftrightarrow{J} \ket{21} \xleftrightarrow{\Omega} \ket{11}$
    \begin{equation}
    \begin{aligned}
        -\frac{3\Omega^3 J}{8(\Delta_{1\mathrm{d}} + \beta_1)} &\left[ \frac{1}{(2\Delta_{1\mathrm{d}} + 3\beta_1)(\Sigma_{12} + \beta_1)} \right.\\
        &\left.+ \frac{1}{(\Delta_{12} + \beta_1)(2\Delta_{1\mathrm{d}} + 3\beta_1 - \Delta_{2\mathrm{d}})} \right]
    \end{aligned}
    \end{equation}

    \item $\ket{10} \xleftrightarrow{\Omega} \ket{20} \xleftrightarrow{\Omega} \ket{30} \xleftrightarrow{\Omega} \ket{20} \xleftrightarrow{J} \ket{11}$
    \begin{equation}
    \begin{aligned}
        -\frac{3\Omega^3 J}{8} &\left[ \frac{1}{(\Delta_{1\mathrm{d}} + \beta_1)^2 (2\Delta_{1\mathrm{d}} + 3\beta_1)} \right.\\
        &\left.+ \frac{1}{(\Delta_{12} + \beta_1)^2 (2\Delta_{1\mathrm{d}} + 3\beta_1 - \Delta_{2\mathrm{d}})} \right]
    \end{aligned}
    \end{equation}

    \item $\ket{10} \xleftrightarrow{\Omega} \ket{20} \xleftrightarrow{\Omega} \ket{10} \xleftrightarrow{J} \ket{01} \xleftrightarrow{\Omega} \ket{11}$
    \begin{equation}
        \frac{3\Omega^3 J}{8(\Delta_{1\mathrm{d}} + \beta_1)} \left[ \frac{1}{\Delta_{12}^2} - \frac{1}{\Delta_{1\mathrm{d}}(\Delta_{12} + \beta_1)} \right]
    \end{equation}

    \item $\ket{10} \xleftrightarrow{\Omega} \ket{20} \xleftrightarrow{\Omega} \ket{10} \xleftrightarrow{\Omega} \ket{20} \xleftrightarrow{J} \ket{11}$
    \begin{equation}
        \frac{3\Omega^3 J}{4(\Delta_{1\mathrm{d}} + \beta_1)} \left[ \frac{1}{(\Delta_{1\mathrm{d}} + \beta_1)^2} + \frac{1}{(\Delta_{12} + \beta_1)^2} \right]
    \end{equation}
\end{enumerate}

\section{Model Setup for Mediated Three-Body Interaction}
\label{app:device_details}

We investigate a minimal superconducting architecture designed to engineer effective three-body interactions within an extended Bose--Hubbard framework. The system comprises three capacitively shunted flux qubits (CSFQs, refer to Appendices.~\ref{app:single CSFQ},~\ref{app:CSFQ Hamiltonian}), labeled \(S_1\), \(S_2\), and \(S_3\), arranged in a linear chain and coupled via two intermediate transmon couplers, \(C_1\) and \(C_2\).

This subsection introduces the full microscopic Hamiltonian, defines the relevant low-energy subspace structure, and derives the effective Hamiltonian obtained via block diagonalization. Particular emphasis is placed on the emergence of a conditional three-body interaction of the form \(XZX + YZY\) when restricting the logical subspace to the three-qubit manifold. We analyze the physical origin, symmetry structure, and tunability of this interaction in detail.

The full system Hamiltonian is given by
\begin{equation}
    \label{eq:three_qudit_system_full}
    \begin{aligned}
        H_{\mathrm{full}} &= \sum_{i\in\{1,2,3,c_1,c_2\}} \left( \omega_i\, a^{\dagger}_i a_i + \frac{\beta_i}{2}\, a^{\dagger}_i a^{\dagger}_i a_i a_i \right) \\
        &- g_{12}\,\left(a_1 - a_1^{\dagger}\right)\left(a_2 - a_2^{\dagger}\right) - g_{23}\,\left(a_2 - a_2^{\dagger}\right)\left(a_3 - a_3^{\dagger}\right) \\
        &- \sum_{j\in\{1,2\}} g_j\,\left(a_j - a_j^{\dagger}\right)\left(a_{c1} - a_{c1}^{\dagger}\right)\\ 
        &- \sum_{j\in\{2,3\}} g_j\,\left(a_j - a_j^{\dagger}\right)\left(a_{c2} - a_{c2}^{\dagger}\right),
    \end{aligned}
\end{equation}
where \(a_i^\dagger\) creates an excitation in mode \(i\), and each oscillator is characterized by a Kerr-type nonlinearity \(\beta_i\). The Hamiltonian includes linear couplings for both direct CSFQ--CSFQ and CSFQ–coupler interactions.

In the dispersive limit where the coupler modes are far detuned, they can be adiabatically eliminated, yielding an effective three-site extended Bose--Hubbard model (EBHM) of the form:
\begin{equation}
    \begin{aligned}
        H &= H_{\mathrm{site}} + H_{\mathrm{coupling}}, \\
        H_{\mathrm{site}} &= \sum_{i=1}^{3} \left( \mu_i\, a_i^{\dagger} a_i + \frac{U_i}{2}\, a_i^{\dagger} a_i^{\dagger} a_i a_i \right), \\
        H_{\mathrm{coupling}} &= J_0 \sum_{\langle i,j \rangle} \left(a_i^\dagger a_j + \mathrm{H.c.}\right) + V \sum_{\langle i,j \rangle} n_i n_j,
    \end{aligned}
\end{equation}
where \(J_0\) represents the effective hopping amplitude and \(V\) denotes the strength of the nearest-neighbor density--density interaction. A detailed derivation mapping the microscopic circuit parameters to these EBHM couplings is provided in Appendix~\ref{app:mapping_to_EBHM}.

For the remainder of this discussion, we employ the simplified EBHM form (Eq.~\eqref{eq:EBHM_3_body_full}) to elucidate the mechanism driving the effective three-body interactions. We truncate the local Hilbert space to \(\ket{S_1 S_2 S_3}\) with \(S_i \in \{0,1,2\}\), resulting in a \(3^3 = 27\)-dimensional manifold. To isolate the qubit-relevant low-energy physics, we define the computational subspace \(\mathcal{H}_{\mathrm{C}}\) as the span of states with at most single occupancy per site:
\begin{equation}
    \mathcal{H}_{\mathrm{C}} = \mathrm{span} \left\{ \ket{ijk} \,\middle|\, i,j,k \in \{0,1\} \right\},
\end{equation}
which spans an 8-dimensional basis naturally mapping to a logical three-qubit system.

We target an effective interaction of the form
\begin{equation}
    \left(a_1^\dagger a_3 + a_1 a_3^\dagger\right) n_2,
\end{equation}
which describes a conditional hopping process between the outer sites \(S_1\) and \(S_3\), gated by the occupation of the central site \(S_2\). This interaction arises from virtual transitions through the doubly excited state of \(S_2\) in a specifically detuned EBHM configuration. We assume mirror symmetry between \(S_1\) and \(S_3\), with \(S_2\) detuned to function as a nonlinear mediator.

Applying a block-diagonalizing transformation (BDT) to decouple the computational subspace \(\mathcal{H}_{\mathrm{C}}\) from the leakage subspace \(\mathcal{H}_{\mathrm{NC}}\), we obtain the effective Hamiltonian projected onto \(\mathcal{H}_{\mathrm{C}}\):
\begin{equation}
    \label{eq: EH_3_body}
    \begin{aligned}
        H_{\mathrm{eff}}^{\mathrm{3S}} = & \sum_{i=1}^{3} \tilde{\omega}_i\, a_i^{\dagger} a_i 
        + J \sum_{\langle i,j \rangle} \left( a_i^{\dagger} a_j + \mathrm{H.c.} \right) \\
        & + K \left( a_1^{\dagger} a_3 + a_3^{\dagger} a_1 \right) n_2,
    \end{aligned}
\end{equation}
where \(\tilde{\omega}_i\) are the renormalized on-site energies, \(J\) is the nearest-neighbor hopping strength, and \(K\) quantifies the mediated three-body interaction.

Further projecting this model onto the logical qubit subspace via the mappings \(a_i^\dagger \to \sigma_i^+\) and \(n_i \to (1 - Z_i)/2\) yields the effective spin Hamiltonian:
\begin{equation}
    \begin{aligned}
        H_{\mathrm{eff}}^{\mathrm{3Q}} = & \sum_{i=1}^{3} \frac{-\tilde{\omega}_i}{2} Z_i 
        + \frac{J}{2} \sum_{\langle i,j \rangle} (X_i X_j + Y_i Y_j) \\
        & - \frac{K}{4} \left( X_1 Z_2 X_3 + Y_1 Z_2 Y_3 \right).
    \end{aligned}
\end{equation}
This Hamiltonian describes an \(XY\) spin chain augmented by a conditional three-body \(XZX + YZY\) interaction coupling the outer qubits, dependent on the state of the central spin.

To quantitatively characterize this interaction, we define the normalized operator
\begin{equation}
    \hat{K} = -\frac{1}{4} \left( X_1 Z_2 X_3 + Y_1 Z_2 Y_3 \right),
\end{equation}
which satisfies \(\Tr(\hat{K}^\dagger \hat{K}) = 1\) within \(\mathcal{H}_{\mathrm{C}}\). The effective interaction strength \(\kappa\) is extracted as
\begin{equation}
    \label{eq:kappa_definition}
    \kappa = \Tr\left( \hat{K}^\dagger \Heff \right),
\end{equation}
providing a basis-independent measure of the three-body term.

Crucially, \(\kappa\) is distinct from the raw matrix element \(K = \bra{110} \Heff \ket{011}\). Instead, \(\kappa\) isolates the genuine three-body contribution by subtracting the background two-body hopping:
\begin{equation}
    \kappa = \bra{110} \Heff \ket{011} - \bra{100} \Heff \ket{001}.
\end{equation}
This definition filters out direct exchange processes, serving as a robust diagnostic for identifying and quantifying conditional hopping effects in multi-qubit systems.

\section{Construction of Improved Capacitively Shunted Flux Qubit}
\label{app:single CSFQ}

Superconducting qudits, as multi-level quantum systems, offer a rich platform for quantum simulation and computation. In this section, we introduce an improved  CSFQ design that provides enhanced tunability of both its transition frequency and anharmonicity. By replacing the conventional small junction with a symmetric SQUID and exploiting dual magnetic flux control, the proposed CSFQ enables precise adjustment of its energy spectrum while maintaining robust coherence. This section details the circuit architecture, quantization, and control mechanisms of the CSFQ.

\begin{figure}[tb]
	\centering
    \subfigure{
        \begin{minipage}[b]{1\columnwidth}
            \centering
            \begin{overpic}[scale=0.6]{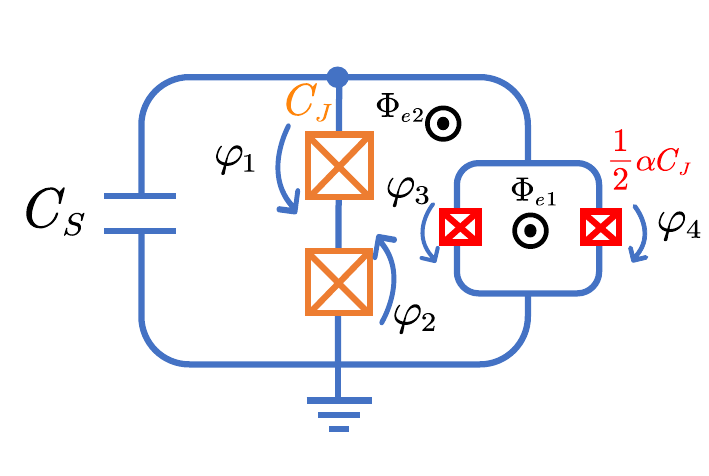}
                \put(0,65){\textbf{(a)}}
            \end{overpic}
        \end{minipage}
        \label{fig:CSFQ_circuit}
    }
    \hspace{5pt}
    \subfigure{
        \begin{minipage}[b]{0.45\columnwidth}
            \centering
            \begin{overpic}[scale=0.3]{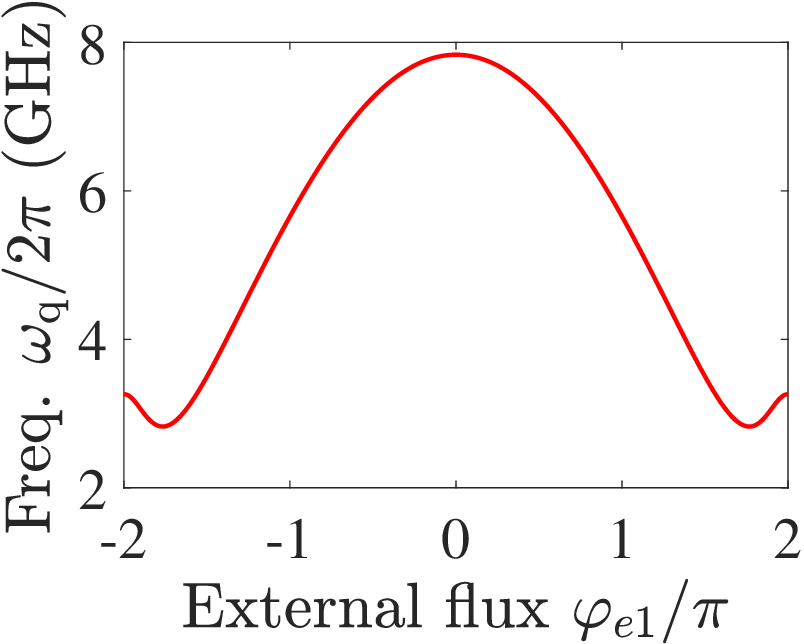}
                \put(3,85){\textbf{(b)}}
            \end{overpic}
        \end{minipage}
        \label{fig:omega_q}
    }
    \hspace{2pt}
    \subfigure{
        \begin{minipage}[b]{0.45\columnwidth}
            \centering
            \begin{overpic}[scale=0.3]{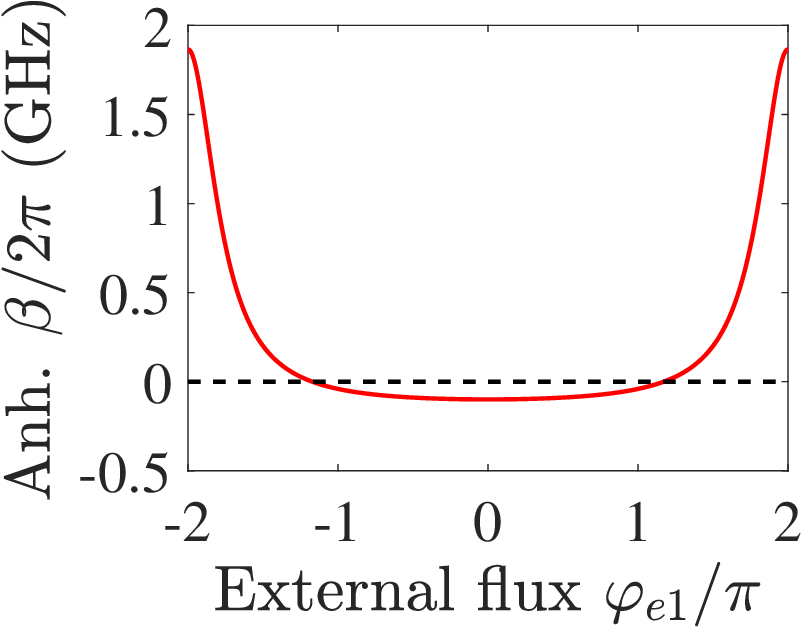}
                \put(3,85){\textbf{(c)}}
            \end{overpic}
        \end{minipage}
        \label{fig:anharmonicity}
    }
\caption{
The improved CSFQ circuit. (a) The circuit diagram shows two identical in-series junctions (orange) in the middle branch, shunted by a capacitor \(C_S\) and a superconducting quantum interference device (SQUID). The SQUID contains two identical junctions (red) with capacitance \(\frac{1}{2}\alpha C_J\), reflecting the junction asymmetry \(E_{J2}/E_{J1} = \frac{\alpha}{2}\). Two external magnetic fluxes, \(\varphi_{\rm e1}\) and \(\varphi_{\rm e2}\), are applied to enable simultaneous tuning of the qubit's frequency and anharmonicity. (b) The dependence of the qubit frequency \(\omega_{\rm q}\) on \(\varphi_{\rm e1}\). (c) The corresponding tuning of the anharmonicity \(\beta\). Circuit parameters are \(C_J = 20\,\text{fF}\), \(C_S = 100\,\text{fF}\), \(E_{J} = 50\,\text{GHz}\), \(E_J/E_C \approx 308\), and \(\alpha = 0.47\). The condition \(\varphi_{e2} + \varphi_{e1}/2 = 0\) is applied. Notice that \(\beta\) can be tuned to both positive and negative values.
}
\label{fig:CSFQ}
\end{figure}

Superconducting qubits inherently possess multiple energy levels and are classified according to the sign of their anharmonicity. In qubits with negative anharmonicity, the third level lies relatively close to the first excited state (as in typical transmons), while positive anharmonicity yields a larger energy separation. A conventional single Josephson junction (JJ) transmon employs a large shunt capacitance to suppress charge noise. In contrast, a CSFQ utilizes a small principal JJ shunted by two larger series JJs and an additional capacitance~\cite{xu2021zz, ku2020suppression}, resulting in extended coherence and enhanced anharmonicity~\cite{krantz2019quantum}. Our improved CSFQ design replaces the small principal JJ with a symmetric SQUID, significantly enhancing tunability.

At the classical level, a single-JJ transmon is described by the magnetic flux \(\Phi\) and the charge \(Q = 2en\), which form a pair of canonically conjugate variables. With the definition \(\Phi_0 = h/(2e)\), one introduces the reduced flux \(\phi = \frac{2\pi\Phi}{\Phi_0}\) and reduced charge \(n = Q/(2e)\). The potential energy is then expressed as \(-E_J \cos(\phi)\), where \(E_J = I_c\Phi_0/(2\pi)\) and \(I_c\) is the critical current. By expanding the cosine potential up to fourth order, the transmon Hamiltonian can be approximated as that of a Duffing oscillator:
\[
H = \omega \hat{a}^\dagger \hat{a} + \frac{\beta}{2} \hat{a}^\dagger \hat{a}^\dagger \hat{a} \hat{a},
\]
with \(\beta\) representing the anharmonicity, and where the canonical operators are defined by \(\hat{\phi} = \phi_{\text{zpf}} (\hat{a}+\hat{a}^\dagger)\) and \(\hat{n} = i\, n_{\text{zpf}} (\hat{a}-\hat{a}^\dagger)\) satisfying \(\phi_{\text{zpf}}\, n_{\text{zpf}} = 1/2\).

In our improved CSFQ (Fig.~\ref{fig:CSFQ_circuit}), two identical in-series JJs are integrated into a closed loop with a SQUID. The SQUID, comprising two symmetric JJs, has an effective capacitance of \(\alpha C_J\) (with \(\alpha < 1\)). Two external magnetic fluxes, \(\Phi_{e1}\) and \(\Phi_{e2}\), control the system; the potential energy is optimized at the “sweet spot” defined by \( \varphi_{e2} + \varphi_{e1}/2 = k\pi \) (with \(k\) an integer and \(\varphi_{ei} = 2\pi\Phi_{ei}/\Phi_0\)). The potential energy is given by
\begin{align}
    \mathcal{U} =& -E_{J} \cos\varphi_{1} - E_{J} \cos\varphi_{2} \nonumber \\
    & - \frac{1}{2} \alpha E_{J} \cos\Bigl( \varphi_{1} - \varphi_{2} + \varphi_{e2} \Bigr) \nonumber \\
    & - \frac{1}{2} \alpha E_{J} \cos\Bigl( \varphi_{1} - \varphi_{2} + \varphi_{e2} + \varphi_{e1} \Bigr),
\end{align}
where \(\varphi_{1,2}\) are the reduced fluxes. By introducing the normal modes \(\phi_+ = \varphi_1+\varphi_2\) and \(\phi_- = \varphi_1-\varphi_2\), and noting that \(\phi_+\) has a negligible impact on the potential, we focus on \(\phi_-\). Canonical quantization, with \([\hat{\phi}_-, \hat{N}_-] = i\) (setting \(\hbar = 1\)), leads to the effective Hamiltonian
\begin{equation}
\begin{split}
    \hat{\mathcal{H}} &= 4E_C\, \hat{N}_-^2 - E_J \Biggl[2 \cos\left(\frac{\phi_+}{2}\right) \cos\left(\frac{\hat{\phi}_-}{2}\right) \\
    &\quad + \alpha \cos\left(\frac{\varphi_{e1}}{2}\right) \cos\left(\hat{\phi}_- + \varphi_{e2} + \frac{\varphi_{e1}}{2}\right)\Biggr],
\end{split}
\end{equation}
where \(\hat{N}_- = \hat{P}_-/(2e)\) and \(E_C = e^2/(2C_-)\) with \(C_- = \alpha C_J + C_S + C_J/2\).

Expressed in the second quantization formalism, the CSFQ Hamiltonian is written as
\begin{equation}\label{eq:qudit}
    \hat{\mathcal{H}} = \omega_{\rm q}\, a^{\dagger} a + \frac{\beta}{2}\, a^{\dagger} a^{\dagger} a a,
\end{equation}
with the qudit frequency
\begin{equation}
    \omega_{\rm q} = \sqrt{8E_J^{\prime} E_C} - E_C + \frac{3E_J E_C}{8E_J^{\prime}},
\end{equation}
and the anharmonicity
\begin{equation}
    \beta = -E_C + \frac{3E_J E_C}{8E_J^{\prime}}.
\end{equation}
Here, \(E_J^{\prime} = E_J \left[\frac{1}{2} + (-1)^k \alpha \cos\left(\frac{\varphi_{e1}}{2}\right)\right]\) is the effective Josephson energy. Figures~\ref{fig:omega_q} and \ref{fig:anharmonicity} demonstrate that by varying \(\varphi_{e1}\) (with \(k=0\) and under the condition \(\varphi_{e2} + \varphi_{e1}/2 = 0\)), the frequency can be tuned continuously from \(3\) to \(8\) GHz, while the anharmonicity varies from \(-0.1\) to \(1.8\) GHz. For additional details on the quantization procedure, see Appendix~\ref{app:CSFQ Hamiltonian}.

In summary, our improved CSFQ design leverages a symmetric SQUID and dual-flux control to provide unprecedented tunability in superconducting qudits. This architecture allows for adjustment of both the qudit frequency and anharmonicity over a broad range, providing a powerful platform for engineering versatile quantum simulators. The ability to precisely control these parameters not only enhances device performance but also facilitates the exploration of diverse quantum phenomena in engineered multi-level systems.

\section{Derivation of Improved Capacitively Shunted Flux Qubit Hamiltonian}
\label{app:CSFQ Hamiltonian}

This section details the derivation of the effective Hamiltonian for our improved CSFQ, which serves as a highly tunable superconducting qudit. By integrating a symmetric SQUID and dual external magnetic fluxes, the design allows independent adjustment of the qudit's transition frequency and anharmonicity. We begin with the circuit Lagrangian, employ fluxoid quantization to reduce the degrees of freedom, and then proceed to canonical quantization. The resulting Hamiltonian, expressed in the form of a Duffing oscillator, forms the basis for precise qudit control in quantum simulation and computing applications.

The circuit (see Fig.~\ref{fig:CSFQ_circuit}) comprises four Josephson junctions: two identical junctions (with Josephson energy \(E_J\)) in series on the left, and two identical junctions in parallel forming a SQUID on the right, each with Josephson energy \(\frac{1}{2}\alpha E_J\). The Lagrangian is defined as
\begin{equation}
	\mathcal{L} = \mathcal{T} - \mathcal{U},
\end{equation}
with the kinetic and potential energies given by
\begin{equation}\label{eq:appendix_kinetic}
\begin{aligned}
    \mathcal{T} &= \frac{1}{2} C_J \dot{\varphi}_1^2 + \frac{1}{2} C_J \dot{\varphi}_2^2 
    + \frac{1}{2} (\alpha C_J + C_S)(\dot{\varphi}_1 - \dot{\varphi}_2)^2, \\
    \mathcal{U} &= -E_J \cos\varphi_1 - E_J \cos\varphi_2 
    - \frac{1}{2}\alpha E_J \cos\left(\varphi_1 - \varphi_2 + \varphi_{e2}\right) \\
    &\quad - \frac{1}{2}\alpha E_J \cos\left(\varphi_1 - \varphi_2 + \varphi_{e2} + \varphi_{e1}\right),
\end{aligned}
\end{equation}
where \(\varphi_i\) denote the reduced fluxes across the junctions, and \(\varphi_{e1,e2}\) are the reduced external fluxes (\(\Phi_{e1,e2}/\Phi_0\)).

Fluxoid quantization imposes the constraints
\[
\varphi_1 - \varphi_2 - \varphi_3 + \varphi_{e2} = 2\pi k_1,\quad \varphi_3 - \varphi_4 + \varphi_{e1} = 2\pi k_2,
\]
with \(k_{1,2}\) integers, which allows us to eliminate the degrees of freedom \(\varphi_3\) and \(\varphi_4\). Defining the new variables 
\[
\phi_+ = \varphi_1 + \varphi_2,\quad \phi_- = \varphi_1 - \varphi_2,
\]
the Lagrangian becomes
\begin{equation}
\begin{aligned}
    \mathcal{T} &= \frac{1}{4} C_J\, \dot{\phi}_+^2 + \frac{1}{2} \left(\alpha C_J + C_S + \frac{C_J}{2}\right) \dot{\phi}_-^2, \\
    \mathcal{U} &= -2E_J \cos\left(\frac{\phi_+}{2}\right) \cos\left(\frac{\phi_-}{2}\right) \\
    &\quad - \alpha E_J \cos\left(\frac{\varphi_{e1}}{2}\right) \cos\left(\phi_- + \varphi_{e2} + \frac{\varphi_{e1}}{2}\right).
\end{aligned}
\end{equation}

The generalized momenta are obtained as
\[
P_{+} = \frac{\partial \mathcal{L}}{\partial \dot{\phi}_+} = \frac{1}{2} C_J\, \dot{\phi}_+,\quad 
P_{-} = \frac{\partial \mathcal{L}}{\partial \dot{\phi}_-} = \left(\alpha C_J + C_S + \frac{C_J}{2}\right) \dot{\phi}_-.
\]
Thus, the classical Hamiltonian is
\begin{equation}
	\mathcal{H} = \frac{P_+^2}{2C_+} + \frac{P_-^2}{2C_-} + \mathcal{U},
\end{equation}
with the effective capacitances \(C_+ = C_J/2\) and \(C_- = \alpha C_J + C_S + C_J/2\).

Because the shunting capacitance \(C_S\) is much larger than \(C_J\) (\(C_S \gg C_J\)), the \(\phi_+\) mode oscillates at a much higher frequency and can be adiabatically eliminated via the Born-Oppenheimer approximation, allowing us to treat \(\phi_+\) as a fixed parameter. This reduction simplifies the potential to a one-dimensional function of \(\phi_-\).

Using canonical quantization with
\[
[\hat{\phi}_-, \hat{N}_-] = i,
\]
where \(\hat{N}_- = \hat{P}_-/(2e)\) is the Cooper-pair number operator and \(E_C = e^2/(2C_-)\) is the charging energy, we obtain the quantum Hamiltonian
\begin{equation}
\begin{split}
    \hat{\mathcal{H}} &= 4E_C\, \hat{N}_-^2 - E_J \Biggl[2 \cos\left(\frac{\phi_+}{2}\right) \cos\left(\frac{\hat{\phi}_-}{2}\right) \\
    &\quad + \alpha \cos\left(\frac{\varphi_{e1}}{2}\right) \cos\left(\hat{\phi}_- + \varphi_{e2} + \frac{\varphi_{e1}}{2}\right)\Biggr].
\end{split}
\end{equation}

At the qubit’s sweet spot—where the potential is symmetric and flux noise is minimized—the condition \(\phi_- = 0\) and \(\varphi_{e2} + \frac{\varphi_{e1}}{2} = k\pi\) (with \(k\) an integer) holds. Expanding the cosine terms around \(\phi_- = 0\) up to fourth order, the Hamiltonian separates into quadratic and quartic components:
\begin{equation}\label{eq:Hfourthorder}
\begin{aligned}
    \hat{\mathcal{H}}^{(2)} &= 4E_C\, \hat{N}_-^2 + \frac{E_J'}{2}\, \hat{\phi}_-^2, \\
    \hat{\mathcal{H}}^{(4)} &= -\frac{1}{24}\left(E_J' - \frac{3}{8}E_J\right) \hat{\phi}_-^4,
\end{aligned}
\end{equation}
where the effective Josephson energy is defined as
\[
E_J' = E_J\left[\frac{1}{2} + (-1)^k \alpha \cos\left(\frac{\varphi_{e1}}{2}\right)\right].
\]
We require a sufficiently large energy ratio \(E_J'/E_C \geq 50\) to suppress charge sensitivity.

Expressing the operators in terms of the creation and annihilation operators,
\[
\hat{\phi}_- = \sqrt{\frac{Z_{\text{imp}}}{2}}(a + a^\dagger),\quad \hat{N}_- = \frac{i}{\sqrt{2Z_{\text{imp}}}}(a - a^\dagger),
\]
with the impedance \(Z_{\text{imp}} = \sqrt{8E_C/E_J'}\), the Hamiltonian can be recast in the form of a Duffing oscillator:
\begin{equation}
	{{H}} = \omega_{\rm q}\, a^{\dagger}a + \frac{\beta}{2}\, a^{\dagger}a^{\dagger}aa,
\end{equation}
where, in the rotating frame (neglecting fast oscillating terms),
\begin{equation}
\begin{split}
	\omega_{\rm q} &= \sqrt{8E_J' E_C} - E_C + \frac{3E_J E_C}{8E_J'}, \\
	\beta &= -E_C + \frac{3E_J E_C}{8E_J'}.
\end{split}
\end{equation}
Notably, both the qudit frequency and anharmonicity are highly tunable via the external flux \(\varphi_{e1}\), typically ranging from \(\omega_{\rm q}/2\pi \approx 3\) to \(8\) GHz and \(\beta/2\pi \approx -0.1\) to \(2\) GHz, respectively.

\section{Derivation of the Extended Bose-Hubbard Model in a Rotating Frame}
\label{app:mapping_to_EBHM}

This subsection outlines how the effective Hamiltonian for a two-mode system can be transformed into the canonical form of the EBHM by applying a rotating frame transformation. The goal is to eliminate the dominant mode frequencies (on the order of GHz), thereby isolating the low-energy interactions (MHz scale) relevant for quantum simulation.

We start from the effective two-body Hamiltonian:
\begin{equation}\label{eq:effective_Hamiltonian_appendix}
\begin{aligned}
    {H}^{\rm eff} = & \sum_{i=1,2} \left( \tilde{\omega}_i\, {n}_i + \frac{\tilde{\beta}_i}{2}\, {b}_i^\dagger {b}_i^\dagger {b}_i {b}_i \right)
    + \tilde{g} \left( {b}_1^\dagger {b}_2 + {b}_1 {b}_2^\dagger \right)\\
    &+ V\, {n}_1 {n}_2,
\end{aligned}
\end{equation}
where \(\tilde{\omega}_i\) denote the renormalized mode frequencies, \(\tilde{\beta}_i\) are the anharmonicities, \(\tilde{g}\) is the effective hopping strength, and \(V\) represents the ZZ interaction.

To remove the large energy offsets \(\tilde{\omega}_i\), we perform a rotating frame transformation defined by
\begin{equation}
    {U}_R(t) = e^{-i\tilde{\omega}_1 ({n}_1+{n}_2)t}.
\end{equation}
The rotating-frame Hamiltonian is given by
\begin{equation}
    {H}_R = {U}_R^\dagger {H}^{\rm eff} {U}_R + i \, \frac{d{U}_R^\dagger}{dt} {U}_R.
\end{equation}
Evaluating this expression yields
\begin{equation}
\begin{aligned}
    {H}_R = & \Delta_{21}\, {n}_2 + \sum_{i=1,2} \frac{\tilde{\beta}_i}{2}\, {b}_i^\dagger {b}_i^\dagger {b}_i {b}_i 
    + \tilde{g} \left( {b}_1^\dagger {b}_2 + {b}_1 {b}_2^\dagger \right)
    + V\, {n}_1 {n}_2,
\end{aligned}
\end{equation}
where \(\Delta_{21} = \tilde{\omega}_2 - \tilde{\omega}_1\) is typically in the MHz range. All remaining terms now reside within the relevant low-energy scale.

To express this in the standard EBHM form,
\begin{equation}
    {H} = -J \sum_{\langle i,j \rangle} {c}_i^\dagger {c}_j + \frac{U}{2} \sum_i {n}_i ({n}_i - 1) - \sum_i \mu_i {n}_i + V \sum_i {n}_i {n}_{i+1},
\end{equation}
where $c_j$ denotes the annihilation operator in EBHM,
we identify the parameters as
\begin{equation}
    J = -\tilde{g},\quad U = \tilde{\beta}_1 = \tilde{\beta}_2,\quad \mu_1 = 0,\quad \mu_2 = -\Delta_{21}.
\end{equation}
This mapping directly connects the effective Hamiltonian in the rotating frame to the EBHM, with all energy scales appropriately reduced to the MHz domain.

This derivation demonstrates how a rotating frame transformation serves as a powerful tool to isolate interaction terms from large mode energies, providing a clean and tractable effective model. The resulting EBHM establishes a direct correspondence between circuit-level parameters and bosonic lattice models, enabling programmable many-body simulations in superconducting architectures.

\section{Derivation on Three-Body Interaction Strength}\label{app:derivation of three-body interaction}
To concretely establish the effective coupling result, we start from the full five-mode Hamiltonian described in Eq.~\eqref{eq:EBHM_3_body_full}, with mode ordering $\ket{Q_1,C_1,Q_2,C_2,Q_3}$. For simplicity, we consider a symmetric configuration with $Q_1$ and $Q_3$ identical, while $Q_2$ is slightly detuned. The effective three-body interaction strength is defined as
\begin{equation}
    \kappa = \bra{110} \Heff \ket{011} - \bra{100} \Heff \ket{001} \equiv \kappa_{110} - \kappa_{100}.
\end{equation}
We focus here on the first term, \(\kappa_{110}\), which can be computed within the two-excitation manifold using perturbation theory up to fourth order.

\subsection{Two-Excitation Manifold}

\paragraph*{Second-Order Contribution.}
There is a single second-order process connecting the states \(\ket{10100} \to \ket{00101}\), mediated by the intermediate leakage state \(\ket{00200}\):
\begin{equation*}
\ket{10100} \xleftrightarrow[]{g_{12}} \ket{00200} \xleftrightarrow[]{g_{12}} \ket{00101}, \quad 
\text{amplitude: } \frac{2g_{12}^2}{\Delta_{12} - \beta_2}.
\end{equation*}

\paragraph*{Third-Order Contributions.}
Third-order processes are categorized according to the position of the \(g_{12}\) coupling along the sequence.

\begin{enumerate}[(1)]

\item \textbf{Route with \(g_{12}\) at the first step:}
\begin{align*}
\ket{10100} &\xleftrightarrow[]{g_{12}} \ket{00200} \xleftrightarrow[]{g_1} \ket{00110} \xleftrightarrow[]{g_2} \ket{00101}, \\
&\text{amplitude: } \frac{2g_{12} g_1 g_2}{(\Delta_{12} - \beta_2)\Delta_{1c}}.
\end{align*}

\begin{align*}
\ket{10100} &\xleftrightarrow[]{g_{12}} \ket{10001} \xleftrightarrow[]{g_1} \ket{01001} \xleftrightarrow[]{g_2} \ket{00101}, \\
&\text{amplitude: } -\frac{g_{12} g_1 g_2}{2\Delta_{1c} \Delta_{2c}}.
\end{align*}

\item \textbf{Route with \(g_{12}\) at the second step:}
\begin{align*}
\ket{10100} &\xleftrightarrow[]{g_1} \ket{01100} \xleftrightarrow[]{g_{12}} \ket{01001} \xleftrightarrow[]{g_2} \ket{00101}, \\
&\text{amplitude: } \frac{g_{12} g_1 g_2}{\Delta_{1c} \Delta_{2c}}.
\end{align*}

\begin{align*}
\ket{10100} &\xleftrightarrow[]{g_1} \ket{10010} \xleftrightarrow[]{g_{12}} \ket{00110} \xleftrightarrow[]{g_2} \ket{00101}, \\
&\text{amplitude: } \frac{g_{12} g_1 g_2}{\Delta_{1c} \Delta_{2c}}.
\end{align*}

\item \textbf{Route with \(g_{12}\) at the third step:}
\begin{align*}
\ket{10100} &\xleftrightarrow[]{g_1} \ket{01100} \xleftrightarrow[]{g_1} \ket{00200} \xleftrightarrow[]{g_{12}} \ket{00101}, \\
&\text{amplitude: } \frac{2g_{12} g_1 g_2}{(\Delta_{12} - \beta_2)\Delta_{1c}}.
\end{align*}

\end{enumerate}

\paragraph*{Fourth-Order Contributions.}
We focus on fourth-order processes involving only \(g_1\) and \(g_2\), and categorize them as PQQQP and PQPQP types.

\begin{enumerate}[(1)]

\item \textbf{PQQQP processes:}

\begin{align*}
\ket{10100} &\xleftrightarrow[]{g_1} \ket{10010} \xleftrightarrow[]{g_2} \ket{01010} 
\xleftrightarrow[]{g_1} \ket{00110} \xleftrightarrow[]{g_2} \ket{00101}, \\
&\text{amplitude: } \frac{g_1^2 g_2^2}{\Delta_{2c} (\Delta_{1c} + \Delta_{2c})} \left( \frac{1}{\Delta_{1c}} + \frac{1}{\Delta_{2c}} \right).
\end{align*}

\begin{align*}
\ket{10100} &\xleftrightarrow[]{g_1} \ket{01100} \xleftrightarrow[]{g_1} \ket{00200} 
\xleftrightarrow[]{g_2} \ket{00110} \xleftrightarrow[]{g_2} \ket{00101}, \\
&\text{amplitude: } \frac{2 g_1^2 g_2^2}{\Delta_{1c}^2 (\Delta_{12} - \beta_2)}.
\end{align*}

\begin{align*}
\ket{10100} &\xleftrightarrow[]{g_1} \ket{01100} \xleftrightarrow[]{g_2} \ket{01010} 
\xleftrightarrow[]{g_1} \ket{00110} \xleftrightarrow[]{g_2} \ket{00101}, \\
&\text{amplitude: } \frac{g_1^2 g_2^2}{\Delta_{1c} (\Delta_{1c} + \Delta_{2c})} \left( \frac{1}{\Delta_{1c}} + \frac{1}{\Delta_{2c}} \right).
\end{align*}

\item \textbf{PQPQP process:}

\begin{align*}
\ket{10100} &\xleftrightarrow[]{g_1} \ket{10010} \xleftrightarrow[]{g_2} \ket{10001} 
\xleftrightarrow[]{g_1} \ket{01001} \xleftrightarrow[]{g_2} \ket{00101}, \\
&\text{amplitude: } -\frac{g_1^2 g_2^2}{\Delta_{2c}^2 \Delta_{1c}}.
\end{align*}

\end{enumerate}

\paragraph*{Total result.}
Combining all contributions up to the fourth order, we obtain
\begin{equation}
\begin{aligned}
    \kappa_{110} =&\frac{2g_{12}^2}{\Delta_{12} - \beta_2} 
    + g_{12} g_1 g_2 \left[ \frac{4}{(\Delta_{12} - \beta_2)\Delta_{1c}} + \frac{3}{2 \Delta_{1c} \Delta_{2c}} \right] \\
    &+ g_1^2 g_2^2 \left[ \frac{1}{\Delta_{1c} + \Delta_{2c}} \left( \frac{1}{\Delta_{1c}} + \frac{1}{\Delta_{2c}} \right)^2 \right.\\
    &\left.+ \frac{2}{\Delta_{1c}^2 (\Delta_{12} - \beta_2)} - \frac{1}{\Delta_{1c} \Delta_{2c}^2} \right].
\end{aligned}
\end{equation}

\subsection{One-Excitation Manifold}

\paragraph*{Second-Order Contribution.}
There is a single second-order process connecting \(\ket{10000} \to \ket{00001}\) via intermediate state \(\ket{00100}\):
\begin{equation*}
\ket{10000} \xleftrightarrow[]{g_{12}} \ket{00100} \xleftrightarrow[]{g_{12}} \ket{00001}, \quad 
\text{amplitude: } \frac{g_{12}^2}{\Delta_{12}}.
\end{equation*}

\paragraph*{Third-Order Contributions.}
Third-order processes are of PPQP type and involve intermediate couplings via \(g_1, g_2\), and \(g_{12}\):

\begin{align*}
\ket{10000} &\xleftrightarrow[]{g_1} \ket{00100} \xleftrightarrow[]{g_2} \ket{00010} \xleftrightarrow[]{g_{12}} \ket{00001}, \\
\ket{00001} &\xleftrightarrow[]{g_{12}} \ket{00010} \xleftrightarrow[]{g_2} \ket{00100} \xleftrightarrow[]{g_1} \ket{10000}, \\
&\text{amplitude: } -\frac{g_1 g_2 g_{12}}{\Delta_{1c} \Delta_{2c}}.
\end{align*}

\paragraph*{Fourth-Order Contribution.}
There is a single PQPQP-type fourth-order process:

\begin{align*}
\ket{10000} &\xleftrightarrow[]{g_1} \ket{01000} \xleftrightarrow[]{g_1} \ket{00100} 
\xleftrightarrow[]{g_2} \ket{00010} \xleftrightarrow[]{g_2} \ket{00001}, \\
&\text{amplitude: } -\frac{g_1^2 g_2^2}{\Delta_{1c}^2 \Delta_{2c}}.
\end{align*}

\paragraph*{Total result.}
Combining all contributions up to the fourth order, we obtain
\begin{align*}
&\kappa_{100} = \frac{g_{12}^2}{\Delta_{12}} - \frac{g_1 g_2 g_{12}}{\Delta_{1c} \Delta_{2c}} - \frac{g_1^2 g_2^2}{\Delta_{1c}^2 \Delta_{2c}}.
\end{align*}

\paragraph*{Effective Coupling Strength.}
Finally, the full effective three-body coupling strength is given by
\begin{widetext}
    \begin{equation}
\begin{aligned}
\kappa =\ &\kappa_{110} - \kappa_{100} \\
=\ &\frac{2 g_{12}^2}{\Delta_{12} - \beta_2} - \frac{g_{12}^2}{\Delta_{12}} 
+ g_{12} g_1^2 \left[ \frac{4}{(\Delta_{12} - \beta_2)\Delta_{1c}} + \frac{5}{2 \Delta_{1c} \Delta_{2c}} \right] \\
&+ g_1^4 \left[ \frac{1}{\Delta_{1c} + \Delta_{2c}} \left( \frac{1}{\Delta_{1c}} + \frac{1}{\Delta_{2c}} \right)^2 
+ \frac{2}{\Delta_{1c}^2 (\Delta_{12} - \beta_2)} - \frac{1}{\Delta_{1c} \Delta_{2c}^2} + \frac{1}{\Delta_{1c}^2 \Delta_{2c}} \right].
\end{aligned}
\label{eq: full Hamiltonian perturbation}
\end{equation}
\end{widetext}

\end{document}